\newif\ifextended

\extendedtrue

\ifextended
\documentclass[acmsmall,screen,authorversion,nonacm]{acmart}
\else
\documentclass[acmsmall,screen]{acmart}
\fi

\ifextended%
\else%
\setcopyright{cc}
\setcctype{by}
\acmDOI{10.1145/3776684}
\acmYear{2026}
\acmJournal{PACMPL}
\acmVolume{10}
\acmNumber{POPL}
\acmArticle{42}
\acmMonth{1}
\received{2025-07-10}
\received[accepted]{2025-11-06}
\fi

\usepackage[misc,geometry]{ifsym}
\usepackage{amsmath,stmaryrd}

\usepackage{tcolorbox}
\usepackage{hyperref}
\usepackage[capitalise,nameinlink]{cleveref}
\usepackage{xspace}
\usepackage{subcaption}
\usepackage{tcolorbox}

\usepackage{wrapfig}
\usepackage{listings}

\usepackage{multirow}

\ifextended%
  \usepackage{apxproof}
\else%
  \usepackage[appendix=strip]{apxproof}
\fi
\theoremstyle{acmplain}
\newtheoremrep{theorem}{Theorem}[section]
\newtheoremrep{lemma}[theorem]{Lemma}
\newtheoremrep{proposition}[theorem]{Proposition}
\newtheoremrep{observation}[theorem]{Observation}
\theoremstyle{acmdefinition}
\newtheoremrep{definition}[theorem]{Definition}
\newtheoremrep{example}[theorem]{Example}

\usepackage{halloweenmath}
\usepackage[
group-digits=integer, group-minimum-digits=4, %
free-standing-units, unit-optional-argument, %
]{siunitx}[=v2]
\usepackage{pgfplots}
\pgfplotsset{
	compat=1.9,
	table/col sep=tab, %
	unbounded coords=jump, %
	filter discard warning=false, %
}
\SendSettingsToPgf %

\usetikzlibrary{shapes,arrows,positioning,calc,fit,backgrounds}
\usetikzlibrary{decorations.pathmorphing,automata,petri}
\usetikzlibrary{shapes.multipart}
\usetikzlibrary{positioning}
\usetikzlibrary{arrows}
\usetikzlibrary{calc}
\tikzstyle{transition}=[rectangle,thick,fill=black,minimum height=3.5pt,minimum width=6mm,inner ysep=0]
\tikzstyle{empstate}=[draw, inner sep=2pt, outer sep=2pt, align=center, font=\footnotesize, draw, rectangle split, rectangle split parts=2, rounded corners=3pt]

\usetikzlibrary{external}
\tikzexternaldisable

\newcommand\init{\mathrm{init}}

\newenvironment{externalize}[1]{%
}{%
}

\makeatletter
\newcommand*{\propdef}[2][\@nil]{%
  \def\tmp{#1}%
  \ifx\tmp\@nnil%
    \expandafter\gdef\csname propref#2\endcsname{\hyperref[prop:#2]{\textbf{(#2)}}}%
    \textbf{(#2)}%
  \else%
    \expandafter\gdef\csname propref#2\endcsname{\hyperref[prop:#2]{\textbf{(#1)}}}%
    \textbf{(#1)}%
  \fi%
  \phantomsection\label{prop:#2}%
}
\makeatother

\newcommand\prop[1]{%
  \ifcsname propref#1\endcsname%
    \csname propref#1\endcsname%
  \else%
    \errmessage{Property #1 has not been defined (use \textbackslash{}propdef to define it).}%
  \fi%
}

\tikzset{st/.style={font=\ttfamily,shape=rectangle,rounded corners=.5em,fill=gray!40,inner xsep=.3em,inner ysep=0em,text height=2ex,text depth=.6ex}}
\newcommand{\stsmcol}[2]{\tikz[baseline]{\node[st,fill=#2] at (0,.64ex){\hspace{.3em}\texttt{\strut\footnotesize #1}\hspace{.1em}\strut};}}
\definecolor{main}{RGB}{250,236,115}
\definecolor{t1}{RGB}{250,102,87}
\definecolor{t2}{RGB}{106,195,240}

\newcommand\hoareTriple[3]{{\left\{#1\middle\}\;#2\;\middle\{#3\right\}}}

\renewcommand\P{\mathcal{P}}

\renewcommand\succ[1]{{#1}^\bullet}
\newcommand\pred[1]{{}^\bullet {#1}}

\newcommand\fire[1]{\mathrel{\triangleright_{#1}}}

\newcommand\reachMarkings{\textsc{Reach}}

\newcommand\errorPlaces{P_\mathsf{err}}

\newcommand\st{s\!t}

\newcommand\programVars{V_\P}

\newcommand\postFun[1][]{\mathit{post}_{#1}}
\newcommand\post[3][]{\postFun[#1]\left(#2, #3\right)}

\newcommand\absReachConf[2]{\textsc{ReachConf}(#1,#2)}

\newcommand\enabled[1]{\mathit{enabled}(#1)}

\newcommand\formulas[1]{\mathbf{Fm}(#1)}

\newcommand\ghost{\mathghost}
\newcommand\ghostVars{V_{\ghost}}

\newcommand\ghostUpdates{\mathbf{GhostUpdates}}

\newcommand\thenha[1]{\tilde{\theta}(#1)}

\newcommand\og{\mathcal{O\!G}}

\newcommand\Regions{\mathbf{Regions}}

\newcommand\Territories{\mathbf{Terr}}

\newcommand\places[1]{\mathit{places}(#1)}

\newcommand\treaty[1]{\mathit{treaty}(#1)}

\newcommand\bystanders[2]{\mathit{bystanders}_{#1}(#2)}

\newcommand\lawFun[1][]{\mathit{law}_{#1}}
\newcommand\law[2][]{\lawFun[#1](#2)}

\newcommand\terrFun{\mathit{terr}}
\newcommand\terr[1]{\terrFun(#1)}

\newcommand\co[2]{{#1} \mathrel{\mathsf{co}} {#2}}

\newcommand\saturatedSuccs[3]{\mathit{saturate}_{#1}^{#2}(#3)}

\newcommand\extended[2]{\mathit{extended}_{#1}(#2)}

\newcommand\replaced[2]{\mathit{replaced}_{#1}(#2)}

\newcommand\legalFun{\mathfrak{f}}
\newcommand\legalIndices[2]{\legalFun(#1, #2)}
\newcommand\legal[3]{#3 \in \legalIndices{#1}{#2}}

\newcommand\uautomizer{\textsc{Ultimate Automizer}\xspace}
\newcommand\automizer{\textsc{Automizer}\xspace}

\newcommand\totalBenchmarks{1\,829}
\newcommand\totalProofs{1\,109}

\newcommand\naiveSucc{904}
\newcommand\naiveTO{194}
\newcommand\naiveOOM{10}
\newcommand\naiveOutOfRessources{204}
\newcommand\naiveGenTime{39.29}

\newcommand\naiveValid{316}
\newcommand\naiveValidPercent{35}
\newcommand\naiveValTO{781}
\newcommand\naiveValOOM{8}
\newcommand\naiveValUnknown{0}
\newcommand\naiveValTime{55.20}

\newcommand\imperialSucc{1\,109}
\newcommand\imperialTO{0}
\newcommand\imperialOOM{0}

\newcommand\imperialGenTime{0.33}
\newcommand\imperialMaxGenTime{68.40}
\newcommand\imperialValid{941}
\newcommand\imperialValidPercent{85}
\newcommand\imperialValTO{148}
\newcommand\imperialValOOM{0}
\newcommand\imperialValUnknown{19}
\newcommand\imperialValTime{2.09}

\newcommand\modularFocusSucc{1\,109}
\newcommand\modularFocusTO{0}
\newcommand\modularFocusOOM{0}

\newcommand\modularFocusGenTime{0.38}

\newcommand\modularFocusValid{949}

\newcommand\modularFocusValTO{130}
\newcommand\modularFocusValOOM{0}
\newcommand\modularFocusValUnknown{28}
\newcommand\modularFocusValTime{1.54}

\newcommand\additionalModularFocus{8}
\newcommand\naiveToImperialSizeRatio{15}
\newcommand\focusSmallerPercent{23}
\newcommand\imperialUnderOneSec{90}
\newcommand\numCrashes{8}
\newcommand\threeTimesLessGU{98}
\newcommand\tenTimesLessGU{62}
\newcommand\minPlaces{31}
\newcommand\maxPlaces{1\,717}

\newcommand\minTransitions{27}
\newcommand\maxTransitions{1\,776}

\newcommand\percValFasterVeriOverall{78}

\newcommand\percValSlowerVeriTenSeconds{15}
\newcommand\percValFasterVeriTenSeconds{85}

\newcommand\percValThriceFasterVeriTenSeconds{52}

\begin{document}
\ifextended%
  \title[The Ghosts of Empires: Extracting Modularity from Interleaving-Based Proofs]{The Ghosts of Empires: Extracting Modularity from Interleaving-Based Proofs (Extended Version)}
\else%
  \title[The Ghosts of Empires: Extracting Modularity from Interleaving-Based Proofs]{The Ghosts of Empires: \\Extracting Modularity from Interleaving-Based Proofs}
\fi
\author{Frank Sch\"ussele}
\orcid{0000-0002-5656-306X}
\email{schuessf@informatik.uni-freiburg.de}
\affiliation{%
  \institution{University of Freiburg}
  \city{Freiburg im Breisgau}
  \country{Germany}
}

\author{Matthias Zumkeller}
\orcid{0009-0004-3696-6269}
\email{zumkellm@informatik.uni-freiburg.de}
\affiliation{%
  \institution{University of Freiburg}
  \city{Freiburg im Breisgau}
  \country{Germany}
}

\author{Miriam Lagunes-Rochin}
\orcid{0009-0008-9482-3900}
\email{miriam.lagunes-rochin@siemens.com}
\affiliation{%
  \institution{University of Freiburg}
  \city{Freiburg im Breisgau}
  \country{Germany}
}

\author{Dominik Klumpp}
\orcid{0000-0003-4885-0728}
\email{klumpp@lix.polytechnique.fr}
\affiliation{%
  \institution{LIX -- CNRS -- \'Ecole Polytechnique}
  \city{Palaiseau}
  \country{France}
}
\affiliation{
  \institution{University of Freiburg}
  \city{Freiburg im Breisgau}
  \country{Germany}
}

\begin{abstract}
  Implementation bugs threaten the soundness of algorithmic software verifiers.
  Generating correctness certificates for correct programs
  allows for efficient independent validation of verification results,
  and thus
  helps to reveal such bugs.
	Automatic generation of small, compact correctness proofs for concurrent programs is challenging,
	as the correctness arguments may depend on the particular interleaving, which can lead to exponential explosion.
	We present an approach that converts an interleaving-based correctness proof, as generated by many algorithmic verifiers, into a thread-modular correctness proof in the style of Owicki and Gries.
	We automatically synthesize \emph{ghost variables} that capture the relevant interleaving information, and abstract away irrelevant details.
	Our evaluation shows that the approach is efficient in practice and generates compact proofs, compared to a baseline.
\end{abstract}


\maketitle

\section{Introduction}

Formal verification promises to eliminate certain kinds of defects, and to ensure correctness of verified systems with \emph{absolute, mathematical certainty}.
Yet, to deliver on these promises,
in practice we rely on implementations of such formal approaches in verification \emph{tools}.
Though based on sound methodology,
verification tools (that are sophisticated enough to handle real-world applications)
are themselves complex software systems,
and are prone to subtle implementation bugs.
To address this issue and increase the trustworthiness of formal methods tools generally,
there is an increasing movement towards requiring such tools to accompany their output
(e.g., a verdict of \emph{correct} or \emph{incorrect}) with an additional certificate
that allows a user to efficiently check the accuracy of the output (i.e., to \emph{validate} the certificate).
Such certificates, sometimes also called \emph{proofs} or \emph{witnesses},
have long become standard practice for propositional SAT solvers~\cite{Wetzler:DRAT},
and are a topic of interest also in hardware model checking~\cite{Froyleks:Certificates-HWMC,Froyleks:Certifying-PA},
satisfiability modulo theories (SMT)~\cite{Hoenicke:Proof-Format,Schurr:Alethe},
classical planning~\cite{Eriksson:Unsolvability-Certificates},
and (as in this work) software verification~\cite{Ayaziova:Witnesses20,Beyer:Witnesses-Where-We-Go,VMCAI24:Ghost-Witnesses}.

We consider in particular correctness proofs
for the algorithmic verification of concurrent programs.
While generating and checking correctness proofs for sequential programs is already difficult,
in particular for real-world programming languages~\cite{Ayaziova:Witnesses20,Beyer:Witnesses-Where-We-Go},
correctness proofs for concurrent programs present an even greater challenge.
As correctness must be established for all possible interleavings of a concurrent program,
a na\"ive approach to correctness proofs falls victim to the \emph{interleaving problem}:
the proof is typically exponentially large in the size of the program.
This explosion of size runs counter to the idea of certificates,
as validating such proofs is prohibitively expensive.
Therefore, we present an approach to convert the internal, interleaving-based proofs produced by many algorithmic verifiers into compact \emph{thread-modular} proofs for the verified concurrent programs.
Rather than explicitly encoding all interleaving information in a single global proof, thread-modular proofs instead reason over each thread locally.
Where necessary, interleaving information can be encoded through the use of \emph{ghost variables}.
Our algorithm detects and takes advantage of modularity in the given interleaving-based proof,
to encode only as little interleaving information as needed.
The generated compact correctness proofs can be efficiently validated,
and thereby
enable testing the consistency of a verifier's internal reasoning,
revealing implementation bugs or increasing trust in the verifier's soundness.
The key advantage is that no ground truth regarding the correctness of the analysed program is required.
Moreover, validation can increase the confidence in the correctness of a particular, highly-critical piece of software.

Since the seminal work of \citeauthor{Owicki-Gries}~\cite{Owicki-Gries},
there has been a wealth of research into suitable program logics for correctness proofs of concurrent programs.
This includes,
among others, the \emph{rely-guarantee} approach~\cite{Jones:rely-guarantee-thesis},
and many works building on \emph{concurrent separation logics}~\cite{OHearn:CSL,Brookes:CSL-Semantics,Brookes:CSL-2016}.
These works have typically targeted interactive or deductive verification.
By contrast, our work targets algorithmic verification.
Yet, we do not propose a new verification algorithm:
Instead, our approach is compatible with many existing verification algorithms (and mature tools implementing these algorithms)
that explicitly consider all the possible interleavings of a concurrent program.
We propose to analyse the resulting interleaving-based proof,
and construct from it a more compact correctness proof
specifically for the purpose of exporting a certificate.
We base our approach on the work of \citeauthor{Owicki-Gries}~\cite{Owicki-Gries},
as the given proof rule is conceptually simple yet complete.
The key challenge addressed by our approach is to identify which interleaving information is required for the correctness argument,
and which information can be abstracted away.
In this sense, our approach differs from classical constructions of Owicki-Gries proofs from interleaving-based proofs,
as employed in completeness arguments~\cite{PrensaNieto:OG-Complete}.
Such constructions typically capture an excessive amount of interleaving information (e.g.\ the entire history, or the entire global control state), and are not intended for practical use.
Our approach operates on a generic representation of interleaving-based proofs that can capture classical techniques such as abstract interpretation~\cite{Cousot:AI}, predicate abstraction~\cite{Graf:Pred-Abs}, Impact~\cite{McMillan:Impact} or trace abstraction~\cite{Heizmann:TAR}.

\begin{wrapfigure}[11]{r}{0.54\textwidth}
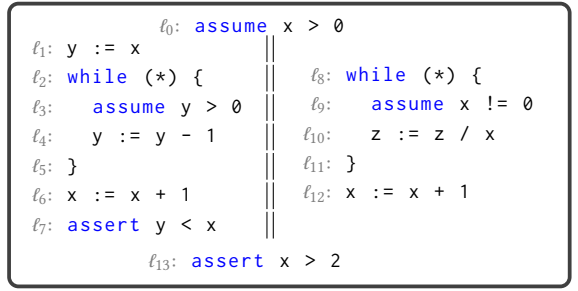

  \vspace*{-1.5em}
  \centering
  \lstset{
      basicstyle=\ttfamily\footnotesize,        %
      breakatwhitespace=false,         %
      captionpos=t,                    %
      commentstyle=\color{green!50!black},    %
      deletekeywords={...},            %
      extendedchars=true,              %
      firstnumber=1,                %
      frame=single,	                   %
      frameround=tttt,
      framexleftmargin=0pt,            %
      keepspaces=true,                 %
      keywordstyle=\color{blue},       %
      language=C,                 %
      morekeywords={*,assume, assert,havoc,atomic,var,...},            %
      numbers=left,                    %
      numbersep=5pt,                   %
      numberstyle=\color{gray}, %
      rulecolor=\color{black},         %
      showlines=false,                  %
      showspaces=false,                %
      showstringspaces=false,          %
      showtabs=false,                  %
      stepnumber=1,                    %
      stringstyle=\color{Periwinkle},      %
      tabsize=2,	                   %
      xleftmargin=0pt                 %
  }
  \renewcommand*\thelstnumber{$\ell_{\the\value{lstnumber}}$:}
  \newcommand\ann[1]{$\color{green!50!black}\{#1\}$}
  \begin{tcolorbox}[width=\linewidth,colback=white]
  \centering
  \vspace*{-1em}
  \begin{minipage}{0.4\textwidth}
  \begin{lstlisting}[firstnumber=0,frame=none]
assume x > 0\end{lstlisting}
  \end{minipage}\\[-1.5em]
  \begin{minipage}[t]{0.4\textwidth}
  \begin{lstlisting}[frame=R,firstnumber=1]
y := x
while (*) {
  assume y > 0
  y := y - 1
}
x := x + 1
assert y < x\end{lstlisting}
  \end{minipage}
  \qquad\quad
  \begin{minipage}[t]{0.35\textwidth}
  \vspace*{1em}
  \begin{lstlisting}[frame=none,firstnumber=8]
while (*) {
  assume x != 0
  z := z / x
}
x := x + 1\end{lstlisting}
  \end{minipage}\\[-0.5em]
  \begin{minipage}{0.4\textwidth}
  \begin{lstlisting}[firstnumber=13,frame=none]
assert x > 2\end{lstlisting}
  \end{minipage}
  \vspace*{-3mm}
  \end{tcolorbox}
  \vspace*{-1em}
  \caption{Example Program}
  \label{fig:example-pseudocode}
\end{wrapfigure}
\paragraph{Example.}
Consider the example program in \cref{fig:example-pseudocode}.
It begins with the assumption that the variable $x$ is positive.
Then, two threads run in parallel.
The left-hand thread performs a computation on its local variable $y$,
and finally asserts that $y<x$ holds.
The right-hand thread performs a computation on its local variable $z$.
After their respective computations, each thread increments the variable $x$.
When both threads have terminated, the program asserts that $x>2$.
An interleaving-based verifier might first show that the local assert in the left-hand thread never fails, using the inductive assertions $\top$ (initially), $y\leq x$ (once $y$ has been initialized to $x$), and $y<x$ (once $x$ has been incremented by the thread).
Then, the verifier might show that the assert statement $x>2$ holds, using the inductive assertions $\top$ (initially), $x>1$ (after either of the two increments), and $x>2$ (after the second increment).
Our goal is to construct an Owicki-Gries proof from these assertions.
I.e., we annotate every line $\ell_i$ with a formula $\omega(\ell_i)$,
such that the annotation is inductive within each thread,
and executing statements in the other thread is guaranteed to preserve the formula (interference-freedom).

However, this is not always possible.
In particular, to show the final assert $x>2$,
the annotation at $\ell_{12}$ needs to distinguish whether the left-hand thread has already incremented $x$ (in which case, $x>1$ holds) or not (in which case, we only know that $x>0$).
Without this case distinction, it is not possible to give an Owicki-Gries proof for $x>2$~\cite{Apt-Olderog:Book}.
Hence, \citeauthor{Owicki-Gries} allow introducing \emph{ghost variables} that are updated throughout the program, but do not otherwise influence the execution.
Such ghost variables can record information about the executed interleaving (such as, whether the other thread has incremented $x$ already), which can then be used in the proof.
In our case, we introduce two boolean ghost variables $\ghost_1$ and $\ghost_2$,
both initialized to $\bot$, which are updated to $\top$ when the left-hand resp.\ the right-hand thread increment $x$.
We then annotate $\ell_{12}$ with the formula $\lnot \ghost_2 \land ((\ghost_1 \land x>1) \lor (\lnot \ghost_1 \land x>0))$.
Analogously, after the increment in line~$\ell_{12}$ has been executed,
the formula $\ghost_2 \land ((\ghost_1 \land x>2) \lor (\lnot \ghost_1 \land x>1))$ holds.
The assertions concerning~$y$ are not relevant to the right-hand thread, and thus need only appear in the proof of the left-hand thread.
For instance, line~$\ell_7$ is annotated with $\lnot \ghost_1 \land y<x \land ((\ghost_2 \land x>1) \lor (\lnot\ghost_2 \land x > 0))$.
We %
return to this example program throughout the paper to demonstrate our approach.

\paragraph{Related Work.}
Beyond the aforementioned works in deductive and interactive verification,
thread-modular correctness proofs for concurrent programs have also been studied in algorithmic verification.
All of the works below differ from our approach in that they present specialized verification algorithms,
rather than extracting thread-modularity from a proof after verification has completed.

\citeauthor{Gupta:OwickiGries-CHC}~\cite{Gupta:OwickiGries-CHC} present an abstraction refinement approach that verifies programs using a rely-guarantee method.
I.e., the approach uses transition predicate abstraction to over-approximate the transitions of environment threads.
They allow the proof of a thread to use interleaving information in the form of control locations and local variables of other threads,
but bias the method to first try finding a solution without any such interleaving information.

\citeauthor{Mine:ThreadModularAI}~\cite{Mine:ThreadModularAI} introduces thread-modular abstract interpretation,
also following a rely-guarantee scheme.
In later works~\cite{Mine:RelationalAI-ThreadModular,Monat:PreciseTM-AI},
this is extended to \emph{flow-sensitive} analyses,
which allow for case distinctions based on the locations of other threads.
They group control locations into \emph{abstract control locations},
which resembles our abstraction of interleaving information.
Here, the abstraction is an input parameter of the analysis, whereas we derive it automatically.
The thread-modular abstract interpreter \textsc{Goblint}~\cite{Saan:Goblint-SVCOMP23,Vojdani:Goblint-Approach} is similarly based on thread-modular abstract interpretation.

In verification of programs with an unbounded number of threads (parameterized programs),
thread-modularity has been relaxed and generalized to \emph{thread-modularity at level~$k$}~\cite{Hoenicke:Thread-Modularity-Levels}.
Rather than proving correctness of individual threads,
these methods instead associate invariants to $k$-tuples of control locations from different threads.
Several works have studied constrained Horn-clause systems to derive such proofs~\cite{Grebenshchikov:Synthesize-Verifiers,Hojjat:CHC-Timed,Gurfinkel:SMT-Based-Parameterized}.
Here, the local states of additional threads can be seen to serve as a particular kind of ghost state.
However, resulting the proof rule remains incomplete.

Finally, previous work has also considered the synthesis of ghost variables not based on capturing the state of other threads.
For instance, \citeauthor{Farzan:Proofs-Count}~\cite{Farzan:Proofs-Count} develop an abstraction-refinement approach that integrates the synthesis of ghost \emph{counter} variables.

\paragraph{Contributions and Roadmap.}
The remainder of the paper is organized as follows.
\Cref{sec:background} defines our formal model of concurrent programs.
In \cref{sec:proof-formats}, we present a general notion of interleaving-based proofs and an adaptation of Owicki-Gries proofs to our formal setting; thereby defining input and output of our approach.
The subsequent four sections contain our core technical contributions:
\begin{itemize}
\item In \cref{sec:empires}, we introduce \emph{empires}, an abstract representation of the interleaving information relevant to a proof,
  and show how to construct Owicki-Gries proofs from empires.
\item \Cref{sec:algo} presents our algorithm to automatically extract such an empire from an interleaving-based proof,
  abstracting away irrelevant interleaving information.
\item We refine the approach in \cref{sec:focus},
  with a criterion to identify which components of a modular interleaving-based proof are relevant for particular threads.
\item The evaluation in \cref{sec:evaluation} demonstrates the practicability of our approach and the quality of the generated Owicki-Gries proofs, compared to a baseline.
\end{itemize}
We conclude with a discussion and an outlook to future work (\cref{sec:discussion}).
\ifextended\else An extended version of this paper including proofs is available online~\cite{popl26:extendedVersion}.\fi

\section{Background: Concurrent Programs}
\label{sec:background}

Our model of concurrent programs uses Petri nets to represent the possible control flow.
Note that typical concurrent programs, such as those based on POSIX threads (pthreads), can be transformed to such a Petri net-based model~\cite{Heizmann:Petrification}.
In \cref{sec:discussion}, we discuss in more detail how our approach also applies to verifiers whose internal program model does not rely on Petri nets.

\begin{definition}
	A \emph{Petri program} is a tuple $\P = (\Sigma, P, T, F, m_\init, \lambda, \errorPlaces)$ with a finite alphabet~$\Sigma$ of statements,
	a finite set of \emph{places}~$P$,
	a finite set of \emph{transitions}~$T$ with $P \cap T = \emptyset$,
	a \emph{flow relation} $F \subseteq (P \times T) \cup (T \times P)$,
	an \emph{initial marking} $m_\init : P \to \mathbb{N}$,
	a \textit{labeling} which maps transitions to statements $\lambda : T \to \Sigma$,
	and a set of \emph{error places} $\errorPlaces \subseteq P$.
\end{definition}
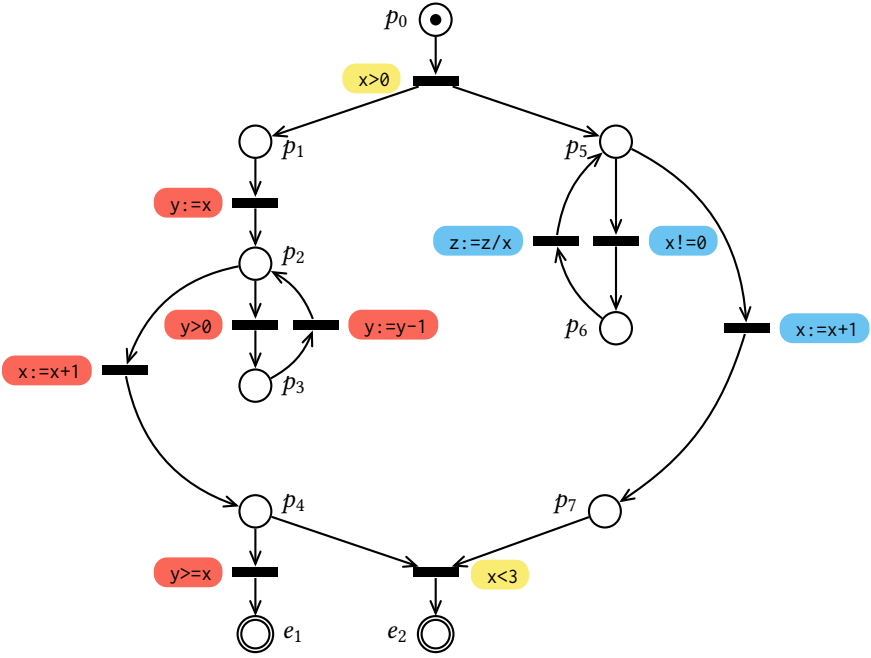
\begin{figure}
	\centering
	\begin{externalize}{petri-program}
	\begin{tikzpicture}[thick,scale=1,node distance=0.5cm,
		trans/.style={->,>=angle 45,thick},
		place/.style={draw,circle,inner sep=1.5mm}
		]

		\node[place,label=left:$p_0$] (p0) {};
		\node[token] (token) at (p0) {};
		\node[transition,below=5mm of p0,label={[xshift=-0.3mm,yshift=0.3mm]left:$\stsmcol{x>0}{main}$}] (t0) {};

		\node[place,below=5mm of t0,label={[yshift=-1mm]right:$p_1$}, xshift=-24mm] (p1) {};
		\node[transition,below=5mm of p1,label=left:$\stsmcol{y:=x}{t1}$] (t1) {};
		\node[place,below=5mm of t1,label={[yshift=1mm]right:$p_2$}] (p2) {};
		\node[transition,below=5mm of p2,label=left:$\stsmcol{y>0}{t1}$] (t2) {};
		\node[transition,below=5mm of p2,label=right:$\stsmcol{y:=y-1}{t1}$, xshift=8mm] (t3) {};
		\node[place,below=5mm of t2,label=right:$p_3$] (p3) {};
		\node[transition,left=12mm of p3,label=left:$\stsmcol{x:=x+1}{t1}$,yshift=2mm] (t4) {};
		\node[place,below=12mm of p3,label={[yshift=1mm]right:$p_4$}] (p4) {};
		\node[transition,below=5mm of p4,label=left:$\stsmcol{y>=x}{t1}$] (te1) {};

		\node[place,below=5mm of t0,label={[yshift=-1mm]left:$p_5$}, xshift=24mm] (p5) {};
		\node[transition,below=10mm of p5,label=right:$\stsmcol{x!=0}{t2}$] (t5) {};
		\node[transition,below=10mm of p5,label=left:$\stsmcol{z:=z/x}{t2}$, xshift=-8mm] (t6) {};
		\node[place,below=20mm of p5,label=left:$p_6$] (p6) {};
		\node[transition,right=12mm of p6,label=right:$\stsmcol{x:=x+1}{t2}$] (t7) {};
		\node[place,right=42mm of p4,label={[yshift=1mm]left:$p_7$}] (p7) {};

		\node[transition,below=5mm of p4,label={[xshift=0.3mm,yshift=-0.3mm]right:$\stsmcol{x<3}{main}$}, xshift=24mm] (te2) {};

		\node[place,below=5mm of te1,label=right:$e_1$, accepting] (e1) {};
		\node[place,below=5mm of te2,label=left:$e_2$, accepting] (e2) {};

		\draw [trans] (p0) to (t0);
		\draw [trans] (t0) to (p1);
		\draw [trans] (p1) to (t1);
		\draw [trans] (t1) to (p2);
		\draw [trans] (p2) to (t2);
		\draw [trans] (t2) to (p3);
		\draw [trans,bend right=20] (p3) to (t3);
		\draw [trans,bend right=20] (t3) to (p2);
		\draw [trans,bend right=30] (p2) to (t4);
		\draw [trans,bend right=30] (t4) to (p4);
		\draw [trans] (p4) to (te1);
		\draw [trans] (te1) to (e1);
		\draw [trans] (p4) to (te2);

		\draw [trans] (t0) to (p5);
		\draw [trans] (p5) to (t5);
		\draw [trans] (t5) to (p6);
		\draw [trans,bend left=20] (p6) to (t6);
		\draw [trans,bend left=20] (t6) to (p5);
		\draw [trans,bend left=30] (p5) to (t7);
		\draw [trans,bend left=20] (t7) to (p7);
		\draw [trans] (p7) to (te2);
		\draw [trans] (te2) to (e2);
	\end{tikzpicture}
	\end{externalize}
	\caption{Petri program corresponding to~\cref{fig:example-pseudocode}. Transitions are drawn as black boxes, places as circles. Double circles indicate error places. Arrows represent the flow relation. Initially, only $p_0$ has a token (black dot).}
	\label{fig:petri-program}
\end{figure}
Each statement in~$\Sigma$ consists of guards and assignments over program variables, and is executed atomically.
$\programVars$~denotes the set of all program variables occurring in statements of $\P$.
\Cref{fig:petri-program} presents the program from~\cref{fig:example-pseudocode} as a Petri program in a graphical form.

\paragraph{Concurrent Control Flow.}
The control flow of the program follows the Petri net semantics.
Specifically, each control state of the program is represented by a \emph{marking},
i.e., a map $m : P \to \mathbb{N}$ that assigns to each place the number of \emph{tokens} that are currently in the place.
The set of predecessor places $\pred{t}$ of a transition $t \in T$ consists of all places $p \in P$ such that $(p,t)\in F$ holds.
Similarly, the set of successors $\succ{t}$ consists of all places $p$ such that $(t,p) \in F$ holds.
When all predecessor places of a transition contain at least one token, we say that the transition is \emph{enabled} in the current marking.
The transition can then be \emph{fired}, resulting in a new marking.
\begin{definition}
	\label{def:firing}
	A transition $t \in T$ is \textit{enabled} in a marking~$m$ if each place~$p \in \pred{t}$ contains a token ($m(p) > 0$).
	Firing the transition results in the successor marking~$m'$, denoted $m \fire{t} m'$, if for each~$p \in P$ we have
    $m'(p) = m(p)-1$ if $p \in \pred{t}$ and $p\notin \succ{t}$,
    $m'(p)=m(p)+1$ if $p \in \succ{t}$ and $p\notin \pred{t}$,
    and $m'(p)=m(p)$ otherwise.
\end{definition}
A sequence $m_0 \fire{t_1} m_1 \fire{t_2} \dots \fire{t_n} m_n$ beginning at the initial marking $m_0 = m_\init$ is called a \emph{firing sequence}.
A marking $m$ is \emph{reachable} if there exists a firing sequence $m_0 \fire{t_1} m_1 \fire{t_2} \dots \fire{t_n} m_n$ such that $m_n = m$.
We write $\reachMarkings$ for the set of all reachable markings.

We consider only \emph{one-safe} Petri programs, i.e., Petri programs where in all reachable markings, every place has at most one token.
To simplify notation, we identify each reachable marking $m : P \to \{0, 1\}$ with the set $\{\,p\mid m(p) = 1 \,\}$ of all places that have a token.
The firing relation for reachable markings then simplifies: $m \fire{t} m'$ holds if and only if\pagebreak[3] $\pred{t} \subseteq m$ and $m' = (m\setminus \pred{t}) \cup \succ{t}$.
We \pagebreak[3] further assume that every transition in the Petri program is enabled in some reachable marking.%
\footnote{
  This assumption simplifies the technical details of the presentation.
  Removing transitions that can never be enabled does not affect correctness of the Petri program.
  If desired, an Owicki-Gries proof of the full Petri program can be easily constructed from a proof of the program where such transitions have been removed.
}

While we sometimes refer to ``threads'' on an intuitive level,
Petri programs have no such concept.
To formalize key conditions for thread-modular proofs,
such as two places (or transitions) ``belonging to different threads'',
we instead rely on the following definitions:

\begin{definition}
	Two places $p, p' \in P$ are  \emph{co-related}, denoted as $\co{p}{p'}$, if $p\neq p'$ and there exists a reachable marking $m$ such that $p \in m$ and $p' \in m$.
	A place $p$ and a transition $t$ are \emph{co-marked} if $p\notin\pred{t}$ and there exists a reachable marking $m$ such that $p \in m$ and $t$ is enabled in $m$ (i.e., $\pred{t} \subseteq m$).
\end{definition}

\paragraph{Correctness.}
Informally speaking, a Petri program is correct if it cannot feasibly reach an error place.
To formalize this, we use the notion of \emph{Hoare triples}. %

For a given set of variables $V$, we write $\formulas{V}$ for the set of all first-order formulas whose free variables are included in $V$.
Given assertions $\varphi$ and $\psi$ over variables $V$, i.e., $\varphi,\psi\in\formulas{V}$, and a statement $\st$ that refers only to variables in $V$,
a \textit{Hoare triple} has the form $\hoareTriple{\varphi}{\st}{\psi}$.
The formulas $\varphi$ and $\psi$ are referred to as precondition resp.\ postcondition.
The Hoare triple \emph{holds} if, whenever $\varphi$ holds and $\st$ is executed and terminates, the assertion $\psi$ must hold afterwards.
We extend this notion to sequences of statements.
Note in particular that the Hoare triple $\hoareTriple{\top}{\st_1\ldots\st_n}{\bot}$ holds
if and only if the sequence $\st_1\ldots\st_n$ cannot be executed
(because the guard of one of the statements cannot be satisfied, and the statement blocks the execution).
\begin{definition}
	\label{def:petri-program-correctness}
	A Petri program $\P = (\Sigma, P, T, F, m_\init, \lambda, \errorPlaces)$ is \emph{correct} if for all firing sequences $m_0 \fire{t_1} \dots \fire{t_n} m_n$ with $m_n \cap \errorPlaces \neq \emptyset$,
	the Hoare triple $\hoareTriple{\top}{\lambda(t_1) \dots \lambda(t_n)}{\bot}$ holds.
\end{definition}
\goodbreak

\section{Interleaving-Based vs.\ Thread-Modular Correctness Proofs}
\label{sec:proof-formats}

Correctness proofs of concurrent programs take different forms.
\emph{Interleaving-based proofs} operate on the interleavings of all concurrent threads, and do not separate local states of different threads from each other or from the global state.
By contrast, \emph{thread-modular proofs} establish correctness in a thread-by-thread fashion, and minimize the required information about global state and interference by other threads.
In the following, we define a general notion of interleaving-based proofs, which we call \emph{invariant domains}, characterizing the kind of proof our approach takes as input.
We then discuss \emph{Owicki-Gries annotations}, the thread-modular proof format produced by our approach.

\subsection{Interleaving-Based Correctness Proofs}
\label{sec:interleaving-proofs}

The approach presented in this paper works on correctness proofs for concurrent programs generated by verification algorithms that prove correctness for all program interleavings.
To capture a general notion of such interleaving-based proofs, we introduce \emph{invariant domains}.
\begin{definition}
	An \emph{invariant domain} for $\P$ is a tuple $(A, \postFun)$
	consisting of a finite set $A \subseteq \formulas{\programVars}$ with $\bot, \top \in A$
	and a function $\postFun: A\times T \rightarrow A$
	such that for all $\varphi\in A$ and $t\in T$,
	the Hoare triple $\hoareTriple{\varphi}{\lambda(t)}{\post{\varphi}{t}}$ holds.
\end{definition}
	For an invariant domain~$(A, \postFun)$,
	the set of \emph{reachable abstract configurations} $\absReachConf{A}{\postFun}$
	is the smallest subset of $\reachMarkings\times A$ that contains $\langle m_\init, \top \rangle$,
	such that for every transition~$t\in T$ and for each $\langle m, \varphi \rangle \in \absReachConf{A}{\postFun}$,
	where $t$~is enabled in~$m$ and $m\fire{t} m'$,
	it also holds that $\langle m', \post{\varphi}{t} \rangle \in\absReachConf{A}{\postFun}$.

\begin{definition}
	\label{def:safe-proof}
	An invariant domain $(A, \postFun)$ for $\P$ is \emph{safe} %
	if for every reachable abstract configuration $\langle m, \varphi \rangle \in\absReachConf{A}{\postFun}$ with $m\cap\errorPlaces\neq\emptyset$,
	it holds that $\varphi$ is equal to $\bot$.
\end{definition}
Safe invariant domains for Petri programs can be computed with many classical algorithmic verification techniques.
For instance, when verifying a program with \emph{abstract interpretation}~\cite{Cousot:AI},
an invariant domain can be extracted from the fixpoint computed by the analysis.%
\footnote{
  The attentive reader may notice that an invariant domain is itself a special case of an abstract domain.
  We introduce the distinct terminology in part due to the special conditions (such as finiteness),
  and in part to avoid confusing the invariant domain derived by applying e.g.\ interval analysis on a program (which consists of finitely many intervals, and a program-specific $\postFun$)
  with the domain used during the analysis itself (here, the interval domain).
}
A safe \emph{predicate abstraction}~$B$~\cite{Graf:Pred-Abs}, as found by \emph{predicate abstraction refinement}~\cite{Ball-Rajamani:CEGAR},
induces a safe invariant domain $(A, \mathit{post}_B)$,
with $A=\{\,\bigwedge B' \mid B' \subseteq B\,\}$ the set of all conjunctions over~$B$
and $\mathit{post}_B$ the standard post-operator of predicate abstraction
(i.e., $\mathit{post}_B(\bigwedge B',\st)=\bigwedge \big\{\,\psi \in B \mid \hoareTriple{\bigwedge B'}{\st}{\psi}\text{ holds}\,\big\}$).
Invariant domains can similarly be derived from %
the abstract reachability graph constructed by the \emph{Impact} algorithm~\cite{McMillan:Impact}.
In \emph{trace abstraction refinement}~\cite{Heizmann:TAR},
the product of the computed \emph{Floyd-Hoare automata} also allows to define a safe invariant domain.
While some of these techniques have been primarily applied for sequential programs,
they extend to concurrent programs as well,
either by considering a ``na\"ive sequentialization'' derived from the Petri program's reachability graph,
or through adaptations~\cite{VMCAI21:Petri-TAR} that help to avoid the state space explosion problem.

\begin{proposition}%
  If there exists a safe invariant domain for $\P$, then $\P$ is correct.
\end{proposition}

\subsection{Owicki-Gries Annotations}
\label{sec:og-annotation}

This section introduces \emph{Owicki-Gries annotations}, the format of correctness proofs for concurrent programs that our approach computes.
As the name implies, these annotations are based on the principles laid out by \citeauthor{Owicki-Gries}~\cite{Owicki-Gries}.
They are thread-modular, in the sense that each thread of a concurrent program is given its own correctness proof, for the most part ignoring the local states of other threads.
Where necessary, limited information about these local states can be captured by the introduction of auxiliary \emph{ghost variables} that record relevant information about the executed interleaving,
do not influence the program's own computation, and can be referred to in the formulae of the proofs of each thread.
This confinement of information about local states of other threads and the executed interleaving to only the relevant aspects
allows Owicki-Gries style proofs (in most cases) to remain concise and efficiently checkable,
while still forming a complete proof system~\cite{Owicki:Consistent-Complete,Lamport:Proving-Multiprocess}.
Whereas \citeauthor{Owicki-Gries} present their proof system as a Hoare-style calculus on the program syntax,
we adapt it here to a more control flow-oriented view,
which is more suitable for algorithmic verification tools.
Therefore, we define \emph{Owicki-Gries annotations} for our notion of Petri programs.
For this purpose, and for the remainder of the paper, we always assume a fixed Petri program $\P = (\Sigma, P, T, F, m_\init, \lambda, \errorPlaces)$.

\begin{definition}
	An \emph{Owicki-Gries annotation} of $\P$ is a
	tuple $\og=(\ghostVars, \omega, \gamma, \rho)$, such that:
	\begin{itemize}
	  \item
		$\ghostVars$ is a finite set of \emph{ghost variables},
		which is disjoint from the set of program variables ($\ghostVars \cap \programVars = \emptyset$).
		Each ghost variable has a (possibly infinite) domain.
	  \item
	    The \emph{invariant mapping} $\omega : P \to \formulas{\programVars \cup \ghostVars}$
	    maps each place $p$ to a formula $\omega(p)$ over program and ghost variables.
	  \item
		The \emph{ghost update mapping} $\gamma : T \to \ghostUpdates$
		labels each transition with a set of updates of the ghost variables.
		$\ghostUpdates$ is the set of partial functions from $\ghostVars$ to expressions over $\programVars \cup \ghostVars$.
		When executing such ghost updates, all these expressions are evaluated,
		and then simultaneously assigned to the respective ghost variables.
		If $\gamma(t)$ is the empty partial function, all ghost variables remain unchanged.
	  \item
		The \emph{initial ghost valuation} $\rho$ assigns each ghost variable $v\in \ghostVars$ an initial value $\rho(v)$ in the respective variable's domain.
	\end{itemize}
\end{definition}

For an Owicki-Gries annotation, we also need to define what it means to be \emph{valid}.

\begin{definition}
  \label{def:og-valid}
	An Owicki-Gries annotation $\og=(\ghostVars, \omega, \gamma, \rho)$ is \emph{valid},
	if the following holds:
	\begin{description}
		\item[{\propdef[initial\textsubscript{$\og$}]{InitialOG}}:] For every initial place $p\in m_\init$,
		the following implication holds:
		\[
		  \left( \bigwedge_{v\in\ghostVars}v=\rho(v) \right) \models \omega(p)
		\]

		\item[{\propdef[inductive\textsubscript{$\og$}]{InductiveOG}}:] For every transition~$t$, the following Hoare triple holds:
		\[
		  \hoareTriple{  \bigwedge_{p \in \pred{t}}\omega(p)  }{  \lambda(t)\,;\,\gamma(t)  }{  \bigwedge_{p \in \succ{t}}\omega(p)  }
		\]
		\item[{\propdef[interference-free\textsubscript{$\og$}]{InterferenceFreeOG}}:] For every transition~$t$, and every place~$p$ that is \emph{co-marked} with~$t$,
		the following Hoare triple holds:
		\[
		  \hoareTriple{  \omega(p) \land \bigwedge_{p'\in\pred{t}}\omega(p')  }{  \lambda(t)\,;\,\gamma(t)  }{  \omega(p)  }
		\]
		\item[{\propdef[safe\textsubscript{$\og$}]{SafeOG}}:] For every error place $p\in \errorPlaces$, the annotation $\omega(p)$ is equivalent to $\bot$.
	\end{description}
\end{definition}
\goodbreak

\begin{propositionrep}%
	\label{prop:og-soundness}
	If there exists a valid Owicki-Gries annotation of $\P$,
	then $\P$ is correct.
\end{propositionrep}
\begin{proofsketch}
  Let $\og = (\ghostVars, \omega, \gamma, \rho)$ be a valid Owicki-Gries annotation of $\P$.
  Then the function $\beta : \reachMarkings \to \formulas{\programVars}$
  with $\beta(m) := \exists \ghostVars\,.\, \bigwedge_{p\in m} \omega(p)$ is a Floyd-Hoare annotation,
  i.e.,
  we have that \emph{(i)} $\beta(m_\init) \equiv \top$,
  that \emph{(ii)}
  $\hoareTriple{\beta(m)}{\lambda(t)}{\beta(m')}$ whenever $m \fire{t} m'$,
  and that \emph{(iii)} $\beta(m) \equiv \bot$ if $m\cap \errorPlaces \neq \emptyset$.
  By sequential composition of Hoare triples,
  the Hoare triple $\hoareTriple{\top}{\lambda(t_1)\ldots\lambda(t_n)}{\bot}$ holds
  for every firing sequence $m_\init = m_0 \fire{t_1} \ldots \fire{t_n} = m_n$
  where $m_n\cap \errorPlaces \neq \emptyset$,
  and $\P$ is correct.
\end{proofsketch}
\begin{proof}
	Let $\og = (\ghostVars, \omega, \gamma, \rho)$ be a valid Owicki-Gries annotation of $\P$.
	We show that the function $\beta : \reachMarkings \to \formulas{\programVars}$
	with $\beta(m) := \exists \ghostVars\,.\, \bigwedge_{p\in m} \omega(p)$ is a Floyd-Hoare annotation,
	i.e.,
	that \emph{(i)} $\beta(m_\init) \equiv \top$,
	that \emph{(ii)}
	$\hoareTriple{\beta(m)}{\lambda(t)}{\beta(m')}$ whenever $m \fire{t} m'$,
	and that \emph{(iii)} $\beta(m) \equiv \bot$ if $m\cap \errorPlaces \neq \emptyset$.
	From this, it follows by sequential composition of Hoare triples that for every firing sequence $m_\init = m_0 \fire{t_1} \ldots \fire{t_n} = m_n$
	where $m_n\cap \errorPlaces \neq \emptyset$,
	the Hoare triple $\hoareTriple{\top}{\lambda(t_1)\ldots\lambda(t_n)}{\bot}$ holds, and thus $\P$ is correct.

    \begin{description}
    \item[Step (i):]
      It is trivial that $\beta(m_\init) \models \top$.
      By destructive equality resolution,
      we have $\top \equiv \big(\exists \ghostVars\,.\, \bigwedge_{v\in\ghostVars} v=\rho(v)\big)$.
      Furthermore, we have $\big(\bigwedge_{v\in\ghostVars} v=\rho(v)\big) \models \big(\bigwedge_{p\in m_\init}\omega(p)\big)$ by \prop{InitialOG}.
      Hence, it follows that
      \[
        \top
        \equiv \big( \exists \ghostVars\,.\,\bigwedge_{v\in\ghostVars} v=\rho(v) \big)
        \models \big( \exists \ghostVars\,.\, \bigwedge_{p\in m_\init}\omega(p) \big)
        \equiv \beta( m_\init)
      \]
      and we are done.
    \item[Step (ii):]
      Let $m\fire{t} m'$.
      Using \prop{InductiveOG} and \prop{InterferenceFreeOG},
      we show that for every $p'\in m'$, it holds that
      \[
        \hoareTriple{ \bigwedge_{p\in m} \omega(p) }{ \lambda(t);\gamma(t) }{ \omega(p') }
      \]
      By conjunction, it follows that
      \[
        \hoareTriple{ \bigwedge_{p\in m} \omega(p) }{ \lambda(t);\gamma(t) }{ \bigwedge_{p'\in m'} \omega(p') }
      \]
      holds.
      As $\lambda(t)$ does not refer to the ghost variables,
      and $\gamma(t)$ consists only of assignments to the ghost variables (without any guards) and does not affect the program variables,
      it follows that the Hoare triple below holds:
      \[
        \hoareTriple{ \exists \ghostVars \,.\, \bigwedge_{p\in m} \omega(p) }{ \lambda(t) }{ \exists \ghostVars \,.\, \bigwedge_{p'\in m'} \omega(p') }
      \]
    \item[Step (iii):]
      Let $m\in\reachMarkings$ with $m\cap\errorPlaces \neq \emptyset$,
      $p\in m\cap\errorPlaces$.
      As $\omega(p) \equiv \bot$ by \prop{SafeOG},
      we have
      \[
        \beta(m) \equiv \big( \exists \ghostVars\,.\, \bigwedge_{p\in m} \omega(p) \big) \equiv \big( \exists \ghostVars\,.\, \bot \big) \equiv \bot \ .
      \]
    \end{description}
\end{proof}

\begin{example}
  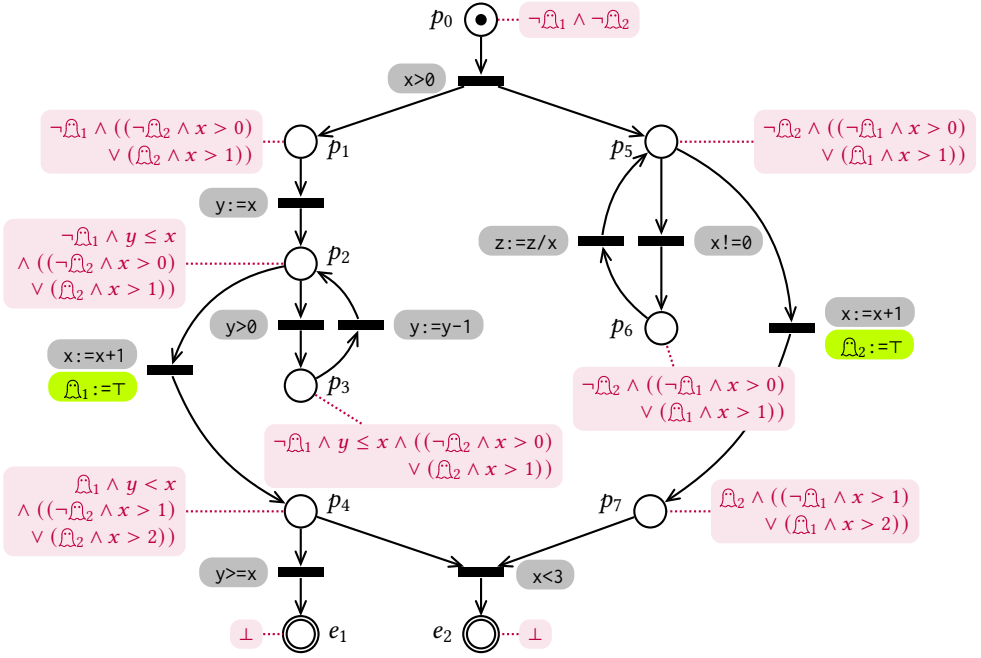
\begin{figure}
  	\centering
  	\begin{externalize}{og-example}
  	\begin{tikzpicture}[thick,scale=1,node distance=0.5cm,
  		trans/.style={->,>=angle 45,thick},
  		place/.style={draw,circle,inner sep=1.5mm},
  		anno/.style={purple, fill=purple!10!white, rounded corners, thin, align=right, font=\footnotesize},
  		annolink/.style={purple,densely dotted}
  		]

  		\node[place,label=left:$p_0$] (p0) {};
  		\node[token] (token) at (p0) {};
  		\node[transition,below=5mm of p0,label={[xshift=-0.3mm,yshift=0.3mm]left:$\stsmcol{x>0}{lightgray}$}] (t0) {};

  		\node[place,below=5mm of t0,label={[yshift=-1mm]right:$p_1$}, xshift=-24mm] (p1) {};
  		\node[transition,below=5mm of p1,label=left:$\stsmcol{y:=x}{lightgray}$] (t1) {};
  		\node[place,below=5mm of t1,label={[yshift=1mm]right:$p_2$}] (p2) {};
  		\node[transition,below=5mm of p2,label=left:$\stsmcol{y>0}{lightgray}$] (t2) {};
  		\node[transition,below=5mm of p2,label=right:$\stsmcol{y:=y-1}{lightgray}$, xshift=8mm] (t3) {};
  		\node[place,below=5mm of t2,label=right:$p_3$] (p3) {};
  		\node[transition,left=12mm of p3,label={[align=right]left:\stsmcol{x:=x+1}{lightgray}\\\stsmcol{$\ghost_1$:=$\top$}{lime}},yshift=2mm] (t4) {};
  		\node[place,below=12mm of p3,label={[yshift=1mm]right:$p_4$}] (p4) {};
  		\node[transition,below=5mm of p4,label=left:$\stsmcol{y>=x}{lightgray}$] (te1) {};

  		\node[place,below=5mm of t0,label={[yshift=-1mm]left:$p_5$}, xshift=24mm] (p5) {};
  		\node[transition,below=10mm of p5,label=right:$\stsmcol{x!=0}{lightgray}$] (t5) {};
  		\node[transition,below=10mm of p5,label=left:$\stsmcol{z:=z/x}{lightgray}$, xshift=-8mm] (t6) {};
  		\node[place,below=20mm of p5,label=left:$p_6$] (p6) {};
  		\node[transition,right=12mm of p6,label={[align=left]right:\stsmcol{x:=x+1}{lightgray}\\\stsmcol{$\ghost_2$:=$\top$}{lime}}] (t7) {};
  		\node[place,right=42mm of p4,label={[yshift=1mm]left:$p_7$}] (p7) {};

  		\node[transition,below=5mm of p4,label={[xshift=0.3mm,yshift=-0.3mm]right:$\stsmcol{x<3}{lightgray}$}, xshift=24mm] (te2) {};

  		\node[place,below=5mm of te1,label=right:$e_1$, accepting] (e1) {};
  		\node[place,below=5mm of te2,label=left:$e_2$, accepting] (e2) {};

  		\draw [trans] (p0) to (t0);
  		\draw [trans] (t0) to (p1);
  		\draw [trans] (p1) to (t1);
  		\draw [trans] (t1) to (p2);
  		\draw [trans] (p2) to (t2);
  		\draw [trans] (t2) to (p3);
  		\draw [trans,bend right=20] (p3) to (t3);
  		\draw [trans,bend right=20] (t3) to (p2);
  		\draw [trans,bend right=30] (p2) to (t4);
  		\draw [trans,bend right=20] (t4) to (p4);
  		\draw [trans] (p4) to (te1);
  		\draw [trans] (te1) to (e1);
  		\draw [trans] (p4) to (te2);

  		\draw [trans] (t0) to (p5);
  		\draw [trans] (p5) to (t5);
  		\draw [trans] (t5) to (p6);
  		\draw [trans,bend left=20] (p6) to (t6);
  		\draw [trans,bend left=20] (t6) to (p5);
  		\draw [trans,bend left=30] (p5) to (t7);
  		\draw [trans,bend left=20] (t7) to (p7);
  		\draw [trans] (p7) to (te2);
  		\draw [trans] (te2) to (e2);
  		
  		\node[anno,right of=p0,anchor=west] (o0) {$\lnot \ghost_1\land \lnot \ghost_2$};
  		\draw[annolink] (o0) -- (p0);

  		\node[anno,left of=p1,anchor=east] (o1) {$\lnot \ghost_1\land ((\lnot \ghost_2\land x>0)$\\$\null\lor(\ghost_2\land x>1))$};
  		\draw[annolink] (o1) -- (p1);

  		\node[anno,left=13mm of p2,anchor=east] (o2) {$\lnot \ghost_1\land y\leq x$\\$\null\land ((\lnot \ghost_2\land x>0)$\\$\null\lor(\ghost_2\land x>1))$};
  		\draw[annolink] (o2) -- (p2);

  		\node[anno,below of=p3,anchor=north,xshift=15mm] (o3) {$\lnot \ghost_1\land y\leq x
  		  			\land ((\lnot \ghost_2\land x>0)$\\$\null\lor(\ghost_2\land x>1))$};
  		\draw[annolink] (o3) -- (p3);

  		\node[anno,left=13mm of p4,anchor=east] (o4) {$\ghost_1\land y<x$\\$\null \land ((\lnot \ghost_2\land x>1)$\\$\null\lor(\ghost_2\land x>2))$};
  		\draw[annolink] (o4) -- (p4);

  		\node[anno,right=10mm of p5,anchor=west] (o5) {$\lnot \ghost_2\land ((\lnot \ghost_1\land x>0)$\\$\null\lor(\ghost_1\land x>1))$};
  		\draw[annolink] (o5) -- (p5);

  		\node[anno,below of=p6,anchor=north,xshift=3mm] (o6) {$\lnot\ghost_2\land ((\lnot \ghost_1\land x>0)$\\$\null\lor(\ghost_1\land x>1))$};
  		\draw[annolink] (o6) -- (p6);

  		\node[anno,right of=p7,anchor=west,xshift=3mm] (o7) {$\ghost_2\land ((\lnot \ghost_1\land x>1)$\\$\null\lor(\ghost_1\land x>2))$};
  		\draw[annolink] (o7) -- (p7);

  		\node[anno,left of=e1,anchor=east] (oe1) {$\bot$};
  		\draw[annolink] (oe1) -- (e1);

  		\node[anno,right of=e2,anchor=west] (oe2) {$\bot$};
  		\draw[annolink] (oe2) -- (e2);
  	\end{tikzpicture}
  	\end{externalize}
  	\caption{Owicki-Gries annotation for the program in \cref{fig:petri-program}, with $\rho(\protect\ghost_1)=\rho(\protect\ghost_2)=\bot$. The transitions $t$ are labeled with the ghost updates $\gamma(t)$ (in green). Each place $p$ is annotated (in purple) with the formula $\omega(p)$.}
  	\label{fig:og-example}
  \end{figure}

  \Cref{fig:og-example} presents an Owicki-Gries annotation $\og=(\ghostVars, \omega, \gamma, \rho)$ for the Petri program of~\cref{fig:petri-program}.
  The idea of this Owicki-Gries annotation is the same as discussed in the introduction.
  There are two ghost variables $\ghostVars=\{\ghost_1, \ghost_2\}$, used to track whether each of the ``threads'' has already incremented $x$.
  Both ghost variables are initialized to $\bot$ (i.e., $\rho(\ghost_1)=\rho(\ghost_2)=\bot$) and updated when $x$ is incremented.
  The annotation~$\omega(p)$ for each place~$p$ includes a case distinction on the ghost variable of the other thread to reflect whether it has already incremented $x$.
  For example, in~$p_5$, $x>1$ holds if the left-hand thread has already executed its increment, and $x>0$ otherwise.
\end{example}
As the example shows, Owicki-Gries annotations can provide a concise correctness argument for Petri programs,
as opposed to, say, specifying an inductive invariant for each of the 17~reachable markings of the example program.
Owicki-Gries annotations can be \emph{validated} efficiently: given the components $(\ghostVars, \omega, \gamma, \rho)$,
it requires only quadratically many (in the size of the program) satisfiability and Hoare triple checks to determine whether the four validity conditions of \cref{def:og-valid} hold.
Special attention has to be paid to \prop{InterferenceFreeOG},
as it references the \emph{co-marked} relation between places and transitions.
Computing this relation from scratch for arbitrary Petri programs would lead to exponential validation cost,
though unfolding techniques~\cite{McMillan:Unfoldings} can reduce the practical impact.
However, for Petri programs derived from other representations such as program code~\cite{Heizmann:Petrification},
the co-marked relation can be constructed efficiently along with the Petri program itself.

It is well-known that Owicki-Gries proofs of correct programs can always be constructed
by introducing ghost variables that track the control locations of all threads,
as shown by \citeauthor{Lamport:Proving-Multiprocess}~\cite{Lamport:Proving-Multiprocess}.
A similar result applies in our formalism.
We refer to the resulting annotations as the ``na\"ive'' Owicki-Gries annotations, to distinguish them from the method developed in the rest of the paper.
As Petri programs do not have a clear notion of threads, we encode the control state by introducing boolean ghost variables for each place, indicating whether the program is currently at that place.
This yields a linear overhead in the number of places compared to \citeauthor{Lamport:Proving-Multiprocess}'s approach.
\begin{definition}
  \label{def:naive-og}
  Let $(A,\postFun)$ be a safe invariant domain for $\P$.
  The corresponding \emph{na\"ive Owicki-Gries annotation} $\og_\textrm{na\"ive}(A,\postFun) = (\ghostVars, \omega, \gamma, \rho)$
  is defined as follows:
  \begin{itemize}
  \item We introduce boolean ghost variables $\ghostVars = \{\, \ghost_p \mid p\in P \,\}$ for each place,
    indicating whether or not the place currently has a token.
  \item Consequently, we initialize these ghost variables to $\rho(\ghost_p) = \top$ for all initial places $p\in m_\init$, and to $\rho(\ghost_p) = \bot$ for all non-initial places $p\in P\setminus m_\init$.
  \item The updates of the ghost variables at every transition $t$
    ensure the intended meaning of the ghost variables,
    by setting $\gamma(t)(\ghost_p) = \top$ if $p\in\succ{t}\setminus\pred{t}$,
    and $\gamma(t)(\ghost_p) = \bot$ if $p\in\pred{t}\setminus\succ{t}$.
  \item Each place $p$ is annotated with the formula
  \[
    \omega(p) \coloneq \bigvee \left\{\, \thenha{m} \land \beta(m) \mid m\in\reachMarkings \land p\in m \,\right\}\ ,
  \]
  where $\thenha{m} \coloneq \bigwedge_{p\in m} \ghost_p \land \bigwedge_{p\in P\setminus m} \lnot\ghost_p$
  and $\beta(m) \coloneq \bigvee_{\langle m,\varphi \rangle\in\absReachConf{A}{\postFun}} \varphi$.
  \end{itemize}
\end{definition}
In the above definition, the function $\beta$ is an inductive annotation (or \emph{Floyd-Hoare annotation}) of the reachability graph formed by all reachable markings.
The formula $\thenha{m}$ constrains the ghost variables to express that the current marking of $\P$ is exactly the marking $m$.
\begin{observationrep}
  The na\"ive Owicki-Gries annotation $\og_\textrm{na\"ive}(A,\postFun)$ is valid.
\end{observationrep}
\begin{proof}
  We show each of the validity conditions separately.
  \begin{description}
  \item[\prop{InitialOG}]
    Let $p\in m_\init$.
    Observe that $\beta(m_\init) \equiv \top$, as $\langle m_\init, \top \rangle \in\absReachConf{A}{\postFun}$ holds.
    Hence, we have that
    \[
      \left( \bigwedge_{v\in\ghostVars} v=\rho(v) \right) \models \thenha{m_\init} \equiv \thenha{m_\init} \land \beta(m_\init) \models \omega(p)
    \]
  \item[\prop{InductiveOG}]
    We begin by observing that the equivalence
    \[
      \left( \bigwedge_{p\in X} \omega(p) \right) \equiv \bigvee_{X \subseteq m \in \reachMarkings} \thenha{m} \land \beta(m)
    \]
    holds for every set of places $X$.
    This follows by inserting the definition of $\omega$,
    applying distributivity,
    and noting that for all $m\neq m'$, the conjunction $\thenha{m} \land \thenha{m'}$ is equivalent to~$\bot$.

    By this observation, we must in fact show the Hoare triple
    {\small\[
      \hoareTriple{ \bigvee \big\{\thenha{m} \land \beta(m)\mid m \in \reachMarkings \land \pred{t} \subseteq m \big\} }{ \lambda(t);\gamma(t) }{ \bigvee \big\{\thenha{m'} \land \beta(m')\mid m' \in \reachMarkings \land \succ{t} \subseteq m\big\}}
    \]}
    To see that this holds,
    observe that it suffices to show for every disjunct of the precondition,
    i.e., for every $m\in\reachMarkings$ with $\pred{t} \subseteq m$,
    that there exists a disjunct in the postcondition,
    i.e., there exists $m'\in\reachMarkings$ with $\succ{t} \subseteq m'$,
    such that
    \[
      \hoareTriple{ \thenha{m}\land\beta(m) }{ \lambda(t);\gamma(t) }{ \thenha{m'} \land \beta(m') }
    \]
    holds.
    We choose the marking $m'$ with $m \fire{t} m'$.
    We observe that $\hoareTriple{\thenha{m}}{\gamma(t)}{\thenha{m'}}$ holds (by definition of $\gamma$)
    and that $\hoareTriple{\beta(m)}{\lambda(t)}{\beta(m')}$ holds (for each disjunct in $\beta(m)$, i.e., for each $\langle m,\varphi\rangle \in \absReachConf{A}{\postFun}$, there exists a disjunct for $\langle m', \post{\varphi}{t}\rangle$ in $\beta(m')$).
    As $\lambda(t)$ and $\gamma(t)$ operate on disjoint sets of variables,
    we can combine the Hoare triples.
  \item[\prop{InterferenceFreeOG}]
    The proof is largely analogous to that of \prop{InductiveOG}:
    We show that for every marking $m$ with $\pred{t} \subseteq m$ and $p\in m$,
    the Hoare triple
    \[
      \hoareTriple{ \thenha{m}\land\beta(m) }{ \lambda(t);\gamma(t) }{ \thenha{m'} \land \beta(m') }
    \]
    holds, where $m\fire{t} m'$.
    Noting that $p\in m'$ still holds, we arrive at the desired Hoare triple.
  \item[\prop{SafeOG}]
    As $(A,\postFun)$ is safe, it must hold that for every $\langle m,\varphi \rangle\in\absReachConf$ with $m\cap\errorPlaces \neq \emptyset$,
    we have $\varphi$ equal to $\bot$.
    Then it follows directly that $\omega(p) \equiv \bot$ for all $p\in\errorPlaces$.
  \end{description}
\end{proof}
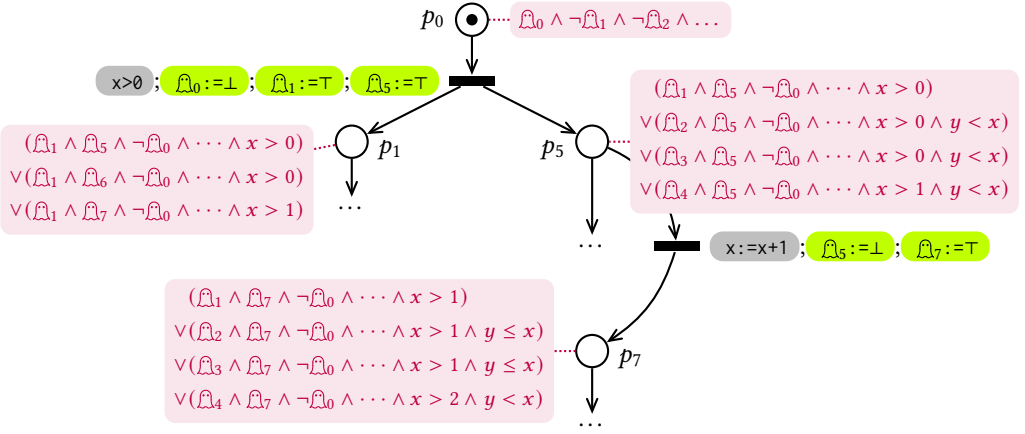
\begin{figure}
	\centering
	\begin{externalize}{naive-og-example}
	\begin{tikzpicture}[thick,scale=1,node distance=0.5cm,
		trans/.style={->,>=angle 45,thick},
		place/.style={draw,circle,inner sep=1.5mm},
		anno/.style={purple, fill=purple!10!white, rounded corners, thin, align=right, font=\footnotesize},
		annolink/.style={purple,densely dotted}
		]

		\node[place,label=left:$p_0$] (p0) {};
		\node[token] (token) at (p0) {};
		\node[transition,below=5mm of p0,label=left:\stsmcol{x>0}{lightgray};\stsmcol{$\ghost_0$:=$\bot$}{lime};\stsmcol{$\ghost_1$:=$\top$}{lime};\stsmcol{$\ghost_5$:=$\top$}{lime}] (t0) {};

		\node[place,below=5mm of t0,label={[yshift=-1mm]right:$p_1$}, xshift=-16mm] (p1) {};
		\node[below=5mm of p1] (t1) {\dots};

		\node[place,below=5mm of t0,label={[yshift=-1mm]left:$p_5$}, xshift=16mm] (p5) {};
		\node[below=10mm of p5] (p6) {\dots};
		\node[transition,right=5mm of p6,label=right:\stsmcol{x:=x+1}{lightgray};\stsmcol{$\ghost_5$:=$\bot$}{lime};\stsmcol{$\ghost_7$:=$\top$}{lime}] (t7) {};
		\node[place,below=10mm of p6,label={[yshift=-1mm]right:$p_7$}] (p7) {};
		\node[below=6mm of p7] (t8) {\dots};

		\draw [trans] (p0) to (t0);
		\draw [trans] (t0) to (p1);
		\draw [trans] (p1) to (t1);
		\draw [trans] (t0) to (p5);
		\draw [trans] (p5) to (p6);
		\draw [trans,bend left=30] (p5) to (t7);
		\draw [trans,bend left=20] (t7) to (p7);
		\draw [trans] (p7) to (t8);

        \node[anno,right of=p0,anchor=west] (o0) {$\ghost_0\land \lnot \ghost_1\land\lnot\ghost_2\land\dots$};
        \draw[annolink] (o0) -- (p0);

        \node[anno,left of=p1,anchor=east,yshift=-5mm] (o1) {$\begin{aligned}
        			&(\ghost_1 \land \ghost_5 \land \lnot \ghost_0 \land \dots \land x>0) \\
        			\lor &(\ghost_1 \land \ghost_6 \land \lnot \ghost_0 \land \dots \land x>0) \\
        			\lor &(\ghost_1 \land \ghost_7 \land \lnot \ghost_0 \land \dots \land x>1)
        		\end{aligned}$};
       \draw[annolink] (o1) -- (p1);

       \node[anno, right of=p5,anchor=west] (o5) {$\begin{aligned}
       			&(\ghost_1 \land \ghost_5 \land \lnot \ghost_0 \land \dots \land x>0) \\
       			\lor &(\ghost_2 \land \ghost_5 \land \lnot \ghost_0 \land \dots \land x>0 \land y< x) \\
       			\lor &(\ghost_3 \land \ghost_5 \land \lnot \ghost_0 \land \dots \land x>0\land y< x)\\
       			\lor &(\ghost_4 \land \ghost_5 \land \lnot \ghost_0 \land \dots \land x>1\land y< x)\\
       		\end{aligned}$};
       \draw[annolink] (o5) -- (p5);

       \node[anno,left of=p7,anchor=east] (o7) {$\begin{aligned}
       			&(\ghost_1 \land \ghost_7 \land \lnot \ghost_0 \land \dots \land x>1) \\
       			\lor &(\ghost_2 \land \ghost_7 \land \lnot \ghost_0 \land \dots \land x>1 \land y\leq x) \\
       			\lor &(\ghost_3 \land \ghost_7 \land \lnot \ghost_0 \land \dots \land x>1\land y\leq x)\\
       			\lor &(\ghost_4 \land \ghost_7 \land \lnot \ghost_0 \land \dots \land x>2\land y< x)\\
       		\end{aligned}$};
       \draw[annolink] (o7) -- (p7);
	\end{tikzpicture}
	\end{externalize}
	\caption{Excerpt of the na\"ive Owicki-Gries annotation for the Petri program from \cref{fig:petri-program}.}
	\label{fig:naive-og-example}
\end{figure}
\begin{example}
  \label{ex:naive-og}
  Consider our example program (\cref{fig:petri-program}),
  and the invariant domain $(A,\postFun)$
  with
  \[
    A = \{\, \top,\ x>0,\ x>0\land y\leq x,\ x > 1,\ x>1 \land y \leq x,\ x>1 \land y < x,\ x>2 \land y<x,\ \bot \}\ ,
  \]
  where $\post{\varphi}{t}$ is the strongest formula in $A$ such that $\hoareTriple{\varphi}{\lambda(t)}{\post{\varphi}{t}}$ holds,
  excepting that $\post{x>0\land y\leq x}{\stsmcol{x:=x+1}{t2}} = (x>0\land y\leq x)$.
  \Cref{fig:naive-og-example}~shows an excerpt of the na\"ive Owicki-Gries annotation $\og_\textrm{na\"ive}(A,\postFun)$.
\end{example}
\goodbreak

This construction provides a canonical way to construct Owicki-Gries annotations from safe invariant domains.
However, the resulting annotations are, of course, deeply unsatisfying.
The annotation $\omega(p)$ of each place $p$ consists of a large disjunction over all reachable markings containing~$p$, encoding the full reachability graph of the Petri program.
The size of these formulas is typically exponential in the size of~$\P$.
Hence, we cannot hope to efficiently validate such Owicki-Gries annotations.
Our evaluation in \cref{sec:evaluation} shows that the explosion in size and validation time is not purely theoretic, but leads to prohibitively bad performance on practical benchmarks.

\section{Empires: An Abstract Representation of Interleaving Information}
\label{sec:empires}

The key challenge of constructing compact Owicki-Gries annotations lies in the synthesis of ghost variables that contain enough interleaving information, without being too fine-grained.
As the na\"ive encoding of the entire global control state leads to unreasonably large Owicki-Gries annotations,
a coarser abstraction of interleaving information is needed.
To this end, we present the notion of \emph{empires}, as an intermediate step on the way from invariant domains to Owicki-Gries annotations.
Empires capture sufficient interleaving information to allow for the construction of valid Owicki-Gries annotations,
while allowing to abstract away irrelevant details of the interleaving, and specifically to summarize multiple global control states if they need not be distinguished.
We begin by explaining the basic concepts underlying empires,
and subsequently show how to construct an Owicki-Gries annotation by using ghost variables to encode an empire.

The first step towards empires consists in defining a suitable abstraction of the given program's global control state.
That is, instead of distinguishing every reachable marking,
we need an abstract representation for a set of markings.
To this end, observe that in the na\"ive Owicki-Gries annotation of~\cref{fig:naive-og-example},
the ghost variables record precisely for all places whether they have a token or not.
However, we have seen in \cref{fig:og-example},
that it suffices to record for each thread, whether it has already incremented $x$.
We do not need to distinguish whether the right-hand thread is currently in $p_5$ or in $p_6$.
Hence, we group places such as this in a \emph{region}.
\begin{definition}
	A \emph{region} $r$ is a non-empty set of places ($\emptyset\subset r\subseteq P$)
	such that for all~$p_1,p_2 \in r$,
	it holds that $p_1$~and~$p_2$ are not co-related.
	The set of all regions is denoted by~$\Regions$.
\end{definition}
A region can be thought of as an abstraction of the \emph{local} control state of a thread.
As the notion of threads does not exist in Petri programs, we here capture this intuition by the condition that the places must not be co-related.
Typically (and always in the construction presented in \cref{sec:algo}), regions are segments of places that are connected by sequential transitions, i.e., transitions with only a single predecessor and successor place (no fork or join transitions).
Branching inside a region is also possible, i.e., a place in the region may have multiple outgoing (or incoming) sequential transitions to (resp.\ from) other places in the region.

An abstraction of the local control state gives rise to a corresponding abstraction of the global control state, i.e., of the Petri program's markings.
Intuitively, we take one region to represent the local control state of each thread.
Formally, this idea is captured by \emph{territories}.
\begin{definition}
    A \emph{territory} $\tau$ is a set of pairwise disjoint regions~$\{r_1,\dots,r_n\}$,
	such that every set~$m$ derived by taking exactly one place from each region in~$\tau$ is a reachable marking.
	The set of all such markings is denoted by~$\treaty{\tau}$.
	We write~$\Territories$ for the set of all territories.
\end{definition}
A territory $\tau$ is an abstract representation of the markings in $\treaty{\tau}$.
To update this representation when the program takes a step, we extend the firing relation from markings to territories.
\begin{definition}
	\label{def:territory-fires}
    A transition $t$ is \emph{enabled} in a territory $\tau$, denoted by $t \in \enabled{\tau}$,
    if for every predecessor place $p\in\pred{t}$,
    there exists a distinct region $r_p\in\tau$ with $p\in r_p$
    (i.e., $r_{p_1} \neq r_{p_2}$ if $p_1\neq p_2$).
    We write $\tau \fire{t} \tau'$ for territory $\tau'$
    if
    \[
      \tau' = \big( \tau \setminus \{\, r_p \mid p\in\pred{t} \,\} \big) \cup \{\, r'_{p'} \mid p' \in \succ{t} \,\} \ ,
    \]
    where for each successor place $p'\in\succ{t}$, $r'_{p'}$ is a distinct region containing $p'$
    (i.e., $r'_{p_1'} \neq r'_{p_2'}$ if $p_1'\neq p_2'$).
    The regions in the set $\tau \setminus \{\, r_p \mid p\in\pred{t} \,\}$ are called \emph{bystanders},
    and can be equivalently written as
    \[
      \bystanders{t}{\tau} = \{\, r\in\tau \mid r \cap \pred{t} = \emptyset \,\}\ .
    \]
\end{definition}
The notion of enabled transitions and the firing relation $\fire{t}$ on territories naturally reflect the corresponding notions on the markings represented by the territories.
\begin{propositionrep}
	\label{prop:territory-succ}
	Transition $t$ is enabled in a territory~$\tau$ iff there exists $m\in\treaty{\tau}$ that enables~$t$.
	If $\tau \fire{t} \tau'$ holds,
	then for all $m\in\treaty{\tau}$ and $m'$ with $m \fire{t} m'$, it follows that $m'\in\treaty{\tau'}$.
\end{propositionrep}
\begin{proof}
	We prove the two statements separately:
    \begin{description}
    	\item[(1)] Transition $t$ is enabled in a territory~$\tau$ iff there exists a marking $m\in\treaty{\tau}$ that enables~$t$: \\
    	$\Rightarrow:$ Assume that the transition $t$ is enabled in $\tau$.
    	Then, for each $p_i \in \pred{t}$, there exists a distinct region $r_{p_i}$ such that $p_i \in r_{p_i}$ (i.e. $r_{p_1} \neq r_{p_2}$ if $p_1 \neq p_2$ for $p_1, p_2 \in \pred{t}$).
    	By the definition of a treaty, there must exist at least one marking $m\in\treaty{\tau}$ constructed by selecting each $p_i \in \pred{t}$ from each distinct region $r_{p_1}$.
    	Thus, $\pred{t} \subseteq m$ which implies that $m$ enables~$t$. \\
    	$\Leftarrow:$ Assume that there exists a marking $m \in \treaty{\tau}$ that enables~$t$.
    	Then, for each $p \in \pred{t}$, it holds that $p \in m$.
    	By the definition of treaty, this implies that for each such $p$, there exists a distinct region $r_p \in \tau$ containing it.
    	Therefore, it directly follows that $t \in \enabled{\tau}$.
    	
    	\item[(2)] If $\tau \fire{t} \tau'$ holds, then for all $m\in\treaty{\tau}$ and $m'$ with $m \fire{t} m'$, it follows that $m'\in\treaty{\tau'}$: \\
    	Consider an arbitrary transition $t$ and two territories $\tau, \tau'$ with $\tau \fire{t} \tau'$.
    	By definition of $\triangleright$, this means that $\tau'$ has the following form: $\tau' = \big( \tau \setminus \{\, r_p \mid p\in\pred{t} \,\} \big) \cup \{\, r'_{p'} \mid p' \in \succ{t} \,\}$ (where for each successor place $p'\in\succ{t}$, $r'_{p'}$ is a distinct region containing $p'$).
    	Now, consider two markings $m, m'$ with $m \fire{t} m'$ which means that $m'$ has the following form: $m' = (m \setminus \pred{t}) \cup \succ{t}$.
    	As $m\in\treaty{\tau}$, it holds that $(m \setminus \pred{t})\in\treaty{\bystanders{t}{\tau}}$.
    	Furthermore, $\succ{t}\in\treaty{\{  r'_{p'} \in (\tau' \setminus \bystanders{t}{\tau}) \mid p' \in \succ{t}  \}}$ holds trivially, which implies that $m' \in \treaty{\tau'}$.
    \end{description}
\end{proof}

Note that, while the firing relation on markings is deterministic (for every~$m$ and~$t$, there is at most one marking~$m'$ with $m\fire{t} m'$),
the firing relation on territories is nondeterministic.
There can be multiple territories~$\tau'$ such that $\tau\fire{t}\tau'$, each representing a different set of markings.
\Cref{prop:territory-succ} prescribes only a minimal set of markings that must be represented by~$\tau'$.
The choice of a ``good'' territory~$\tau'$ is a key step in our empire construction presented in~\cref{sec:algo}.

Having fixed our abstraction of the global control state,
we are ready to introduce empires.
As explained above, the purpose of empires is to capture interleaving information that is relevant to prove correctness.
Hence, we define an empire as a \emph{state machine},
which can update its state whenever a new transition is executed and appended to the current interleaving.
Each state of the empire is labeled by a territory, describing the current global control state,
and a formula over the program variables (the ``law'' of the state),
describing the current values of the program variables.
\begin{definition}[Empire]\label{def:empire}
	An \emph{empire} is a finite state machine $E = (Q, q_\init, \delta, \lawFun, \terrFun)$
	consisting of
	      a finite set of states $Q$,
	      an initial state $q_\init\in Q$,
	      a (partial) transition function $\delta: Q\times T\rightharpoonup Q$,
	      a law mapping $\lawFun: Q \rightarrow \formulas{\programVars}$ labeling each state with a formula, and
	      a territory mapping $\terrFun: Q \rightarrow \textbf{Terr}$ labeling each state with a territory.
\end{definition}
\begin{example}
	An example of an empire is shown in \cref{fig:empire-example}.
	Here, the states use territories to group multiple markings in the same state (e.g., the territory for $q_2$ abstracts the four markings $\{p_2,p_5\},\{p_2,p_6\},\{p_3,p_5\},\{p_3,p_6\}$).
	As a result, the annotation remains compact, since it avoids separate states for different markings with the same law. %
\end{example}
\begin{definition}
\label{def:empire-valid}
	An empire $E = (Q, q_\init, \delta, \lawFun, \terrFun)$ is \emph{valid}
	if the following conditions hold:
	\begin{description}
	  \item[{\propdef[initial-law\textsubscript{$E$}]{initial-law}}:]
	    $\law{q_\init}\equiv\top$
	  \item[{\propdef[initial-territory\textsubscript{$E$}]{initial-territory}}:]
	    $m_\init\in\treaty{\terr{q_\init}}$
	  \item[{\propdef[inductive-law\textsubscript{$E$}]{inductive-law}}:]
	    For all $q\in Q, t\in\enabled{\terr{q}}$, one of the following holds:
	    \begin{itemize}
		  \item $q' \coloneq \delta(q, t)$ is defined and the Hoare triple
		    $\hoareTriple{\law{q}}{\lambda(t)}{\law{q'}}$ holds, or
		  \item the Hoare triple $\hoareTriple{\law{q}}{\lambda(t)}{\bot}$ holds.
	    \end{itemize}
	  \item[{\propdef[inductive-territory\textsubscript{$E$}]{inductive-territory}}:]
	    For all $q,q'\in Q, t\in T$ with $\delta(q, t)=q'$, $\terr{q}\fire{t}\terr{q'}$ holds.
	  \item[{\propdef[safe\textsubscript{$E$}]{safe}}:]
	    For all states $q\in Q$ such that there exists a region $r\in\terr{q}$ with $r \cap \errorPlaces\neq\emptyset$, it holds that $\law{q}\equiv\bot$.
	\end{description}
\end{definition}
The first four validity conditions ensure that the law and territory of the state reached after some interleaving
soundly reflect the values of the program variables resp.\ the control state at that point.
The second alternative of \prop{inductive-law} optionally allows omitting edges
if the law is sufficient to prove that executing some transition is infeasible.
The last condition \prop{safe} ensures that the empire has sufficient information to prove correctness of the program.
Note that, in particular, \prop{safe} may hold vacuously,
as all edges to states with law $\bot$ can be omitted.
\goodbreak
\begin{theoremrep}\label{thm:empire-soundness}
	If there exists a valid empire for the Petri program $\P$, then $\P$ is correct.
\end{theoremrep}
\begin{proofsketch}
    Let $E = (Q, q_\init, \delta, \lawFun, \terrFun)$ be a valid empire for $\P$.
    Assume wlog.\ that for all~$q,q'\in Q$ with $q\neq q'$,
    the formulae $\law{q}$ and $\law{q'}$ differ (at least syntactically).
    Then $(A,\postFun)$ is  a safe invariant domain,
    with $A = \{\,\law{q} \mid q \in Q \,\} \cup \{\bot\}$, where $\post{\bot}{t} = \bot$ and
    \[
      \post{\law{q}}{t} = \begin{cases}
        \law{\delta(q,t)} & \textbf{if } \delta(q,t) \text{ is defined}\\
        \top & \textbf{if } t \text{ is not enabled in } \terr{q}\\
        \bot & \textbf{otherwise}
      \end{cases}
      \vspace*{-1.5em}
    \]
\end{proofsketch}
\begin{proof}
    Let $E = (Q, q_\init, \delta, \lawFun, \terrFun)$ be a valid empire for $\P$.
    Assume wlog.\ that for all states $q,q'\in Q$ with $q\neq q'$,
    the formulae $\law{q}$ and $\law{q'}$ are different (at least syntactically).
    The set $A = \{\,\law{q} \mid q \in Q \,\} \cup \{\bot\}$ and the function $\postFun : A \times T \to A$ with
    \[
      \post{\law{q}}{t} = \begin{cases}
        \law{\delta(q,t)} & \textbf{if } \delta(q,t) \text{ is defined}\\
        \top & \textbf{if } t \text{ is not enabled in } \terr{q}\\
        \bot & \textbf{otherwise}
      \end{cases}
      \quad \text{ and } \quad \post{\bot}{t} = \bot
    \]
    form a safe invariant domain.

    The Hoare triples $\hoareTriple{ \law{q} }{\lambda(t)}{\post{\law{q}}{t}}$ follow from \prop{inductive-law}.
    It can be shown inductively that the reachable abstract configurations all have the form $\langle m,\law{q} \rangle$ for $m\in\treaty{\terr{q}}$
    or $\langle m,\bot \rangle$.
    Hence, by property \prop{safe} of $E$,
    the invariant domain is safe.
\end{proof}
\begin{toappendix}
  Alternatively, we show below (\cref{thm:imperial-og-correct}) that a valid Owicki-Gries annotation can be constructed from a valid empire.
  From \cref{prop:og-soundness}, the correctness of $\P$ then also follows.
\end{toappendix}

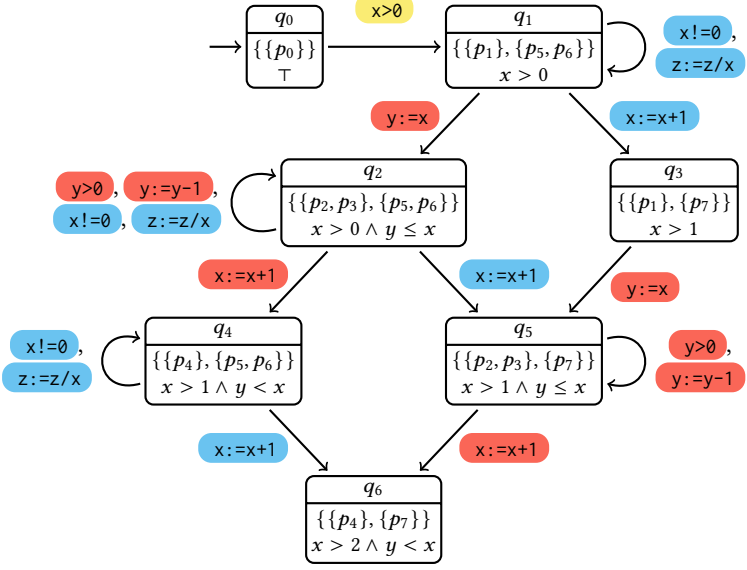
\begin{figure}
	\centering
	\begin{externalize}{empire-example}
	\begin{tikzpicture}[thick,scale=1,node distance=8mm,
		every node/.style={inner sep=2pt, outer sep=2pt, align=center}
		]

		\node [empstate]                         (0) {$q_0$ \nodepart{two} $\{\{p_0\}\}$\\$\top$};
		\node [empstate,right=15mm of 0]         (1) {$q_1$ \nodepart{two} $\{\{p_1\}, \{p_5,p_6\}\}$\\$x>0$};
		\node [empstate,below=of 1,xshift=-20mm] (2) {$q_2$ \nodepart{two} $\{\{p_2,p_3\}, \{p_5,p_6\}\}$\\$x>0\land y\leq x$};
		\node [empstate,below=of 1,xshift=20mm]  (3) {$q_3$ \nodepart{two} $\{\{p_1\}, \{p_7\}\}$\\$x>1$};
		\node [empstate,below=of 2,xshift=-20mm] (4) {$q_4$ \nodepart{two} $\{\{p_4\}, \{p_5,p_6\}\}$\\$x>1\land y<x$};
		\node [empstate,below=of 2,xshift=20mm]   (5) {$q_5$ \nodepart{two} $\{\{p_2,p_3\}, \{p_7\}\}$\\$x>1\land y\leq x$};
		\node [empstate,below=of 5,xshift=-20mm]         (7) {$q_6$ \nodepart{two} $\{\{p_4\}, \{p_7\}\}$\\$x>2\land y<x$};

		\draw[->] ($(0)+(-1,0)$) to (0);
		\draw[->] (0) -- node[above] {$\stsmcol{x>0}{main}$} (1);
		\draw[->] (1) -- node[left,yshift=1mm] {$\stsmcol{y:=x}{t1}$} (2);
		\draw[->] (1) -- node[right,yshift=1mm] {$\stsmcol{x:=x+1}{t2}$} (3);
		\draw[->] (2) -- node[left,yshift=1mm] {$\stsmcol{x:=x+1}{t1}$} (4);
		\draw[->] (2) -- node[right,yshift=1mm] {$\stsmcol{x:=x+1}{t2}$} (5);
		\draw[->] (3) -- node[right,yshift=-1mm] {$\stsmcol{y:=x}{t1}$} (5);
		\draw[->] (4) -- node[left,yshift=-1mm] {$\stsmcol{x:=x+1}{t2}$} (7);
		\draw[->] (5) -- node[right,yshift=-1mm] {$\stsmcol{x:=x+1}{t1}$} (7);

		\draw[->] (1) edge [right,loop right,looseness=3] node {$\stsmcol{x!=0}{t2},$\\$\stsmcol{z:=z/x}{t2}$} (1);
		\draw[->] (2) edge [left,loop left,looseness=3] node {$\stsmcol{y>0}{t1},\stsmcol{y:=y-1}{t1},$\\$\stsmcol{x!=0}{t2},\stsmcol{z:=z/x}{t2}$} (2);
		\draw[->] (4) edge [left,loop left,looseness=3] node {$\stsmcol{x!=0}{t2},$\\$\stsmcol{z:=z/x}{t2}$} (4);
		\draw[->] (5) edge [right,loop right,looseness=3] node {$\stsmcol{y>0}{t1},$\\$\stsmcol{y:=y-1}{t1}$} (5);
	\end{tikzpicture}
	\end{externalize}
	\caption{A valid empire for the Petri program in \cref{fig:petri-program}.}
	\label{fig:empire-example}
\end{figure}

As empires are intended to be an intermediate step on the way from invariant domains to Owicki-Gries annotations,
we have to show (1)~how to derive an empire from an invariant domain $(A,\postFun)$, and (2)~how to construct an Owicki-Gries annotation from an empire $E = (Q,q_\init,\delta,\lawFun, \terrFun)$.
We present our construction for problem~(1) in \cref{sec:algo} below.
As for problem~(2), the key idea is to introduce a single ghost variable, named $\ghost$\,, that tracks the current state of the empire.
We use the empire's territories to determine the possible empire states whenever a token is in a given place $p$,
namely the states $Q_p \coloneq \{\, q\in Q \mid \exists r\in\terr{q}\,.\,p\in r \,\}$.
By taking the laws of these possible states, we constrain the possible program states when there is a token in~$p$.

\begin{definition}
\label{def:empire-og}
Let $E = (Q, q_\init, \delta, \lawFun, \terrFun)$ be an empire.
The \emph{imperial Owicki-Gries annotation} $\og_E = (\ghostVars, \omega, \gamma, \rho)$ is defined as follows:
\begin{itemize}
	\item We introduce a single ghost variable, i.e., $\ghostVars = \{\ghost\}$, which ranges over the states in~$Q$, and tracks the current state in the empire.
	\item Therefore, this ghost variable $\ghost$ is initialized with the initial state $q_\init$ of $E$, i.e., $\rho(\ghost) = q_\init$.
	\item The ghost updates $\gamma$ encode the transition function $\delta$.
	\[
		\gamma(t)(\ghost) = \begin{cases}
		  \delta(\ghost\,,t) & \textbf{if defined}\\
		  \ghost & \textbf{otherwise}
		\end{cases}
	\]
    If $\delta(q,t)=q$ for all $q$ where $\delta(q,t)$ is defined, the ghost update is omitted. 

	\item Each place $p$ is annotated with a case distinction over the possible empire states $Q_p$:
	\[
		\omega(p) \coloneq \bigvee_{q\in Q_p} \big( \ghost=q \land \law{q} \big)
	\]
\end{itemize}
\end{definition}

As $Q$ is finite, the transition function $\delta$ can be encoded through conditional expressions.
It suffices to consider cases where the value of $\ghost$ is in $Q_p$, for all predecessor places~$p$ of~$t$.

\begin{figure}
	\centering
	\begin{externalize}{imperial-og-example}
	\begin{tikzpicture}[thick,scale=1,node distance=0.5cm,
		trans/.style={->,>=angle 45,thick},
		place/.style={draw,circle,inner sep=1.5mm},
  		anno/.style={purple, fill=purple!10!white, rounded corners, thin, align=right, font=\footnotesize},
  		annolink/.style={purple,densely dotted}
		]

		\node[place,label=left:$p_0$] (p0) {};
		\node[token] (token) at (p0) {};
		\node[transition,below=5mm of p0,label={[xshift=-0.3mm,yshift=0.3mm]left:\stsmcol{x>0}{lightgray};\stsmcol{$\ghost$:=$\delta(\ghost, \texttt{x>0})$}{lime}}] (t0) {};

		\node[place,below=5mm of t0,label={[yshift=-1mm]right:$p_1$}, xshift=-24mm] (p1) {};
		\node[transition,below=5mm of p1,label=left:\stsmcol{y:=x}{lightgray};\stsmcol{$\ghost$:=$\delta(\ghost, \texttt{y:=x})$}{lime}] (t1) {};
		\node[place,below=5mm of t1,label={[yshift=1mm]right:$p_2$}] (p2) {};
		\node[transition,below=5mm of p2,label=left:$\stsmcol{y>0}{lightgray}$] (t2) {};
		\node[transition,below=5mm of p2,label=right:$\stsmcol{y:=y-1}{lightgray}$, xshift=8mm] (t3) {};
		\node[place,below=5mm of t2,label={[yshift=0.2mm]right:$p_3$}] (p3) {};
		\node[transition,left=12mm of p3,label={left,align=right}:\stsmcol{x:=x+1}{lightgray}\\\stsmcol{$\ghost$:=$\delta(\ghost, \texttt{x:=x+1})$}{lime}] (t4) {};
		\node[place,below=12mm of p3,label={[yshift=1mm,xshift=-0.5mm]right:$p_4$}] (p4) {};
		\node[transition,below=5mm of p4,label=left:$\stsmcol{y>=x}{lightgray}$] (te1) {};

		\node[place,below=5mm of t0,label={[yshift=-1mm]left:$p_5$}, xshift=24mm] (p5) {};
		\node[transition,below=10mm of p5,label=right:$\stsmcol{x!=0}{lightgray}$] (t5) {};
		\node[transition,below=10mm of p5,label=left:$\stsmcol{z:=z/x}{lightgray}$, xshift=-8mm] (t6) {};
		\node[place,below=20mm of p5,label=left:$p_6$] (p6) {};
		\node[transition,right=12mm of p6,label={right,align=left}:\stsmcol{x:=x+1}{lightgray}\\\stsmcol{$\ghost$:=$\delta(\ghost, \texttt{x:=x+1})$}{lime}] (t7) {};
		\node[place,right=42mm of p4,label={[yshift=1mm,xshift=0.5mm]left:$p_7$}] (p7) {};

		\node[transition,below=5mm of p4,label={[xshift=0.3mm,yshift=-0.3mm]right:$\stsmcol{x<3}{lightgray}$}, xshift=24mm] (te2) {};

		\node[place,below=5mm of te1,label=right:$e_1$, accepting] (e1) {};
		\node[place,below=5mm of te2,label=left:$e_2$, accepting] (e2) {};

		\draw [trans] (p0) to (t0);
		\draw [trans] (t0) to (p1);
		\draw [trans] (p1) to (t1);
		\draw [trans] (t1) to (p2);
		\draw [trans] (p2) to (t2);
		\draw [trans] (t2) to (p3);
		\draw [trans,bend right=20] (p3) to (t3);
		\draw [trans,bend right=20] (t3) to (p2);
		\draw [trans,bend right=30] (p2) to (t4);
		\draw [trans,bend right=30] (t4) to (p4);
		\draw [trans] (p4) to (te1);
		\draw [trans] (te1) to (e1);
		\draw [trans] (p4) to (te2);

		\draw [trans] (t0) to (p5);
		\draw [trans] (p5) to (t5);
		\draw [trans] (t5) to (p6);
		\draw [trans,bend left=20] (p6) to (t6);
		\draw [trans,bend left=20] (t6) to (p5);
		\draw [trans,bend left=30] (p5) to (t7);
		\draw [trans,bend left=20] (t7) to (p7);
		\draw [trans] (p7) to (te2);
		\draw [trans] (te2) to (e2);

        \node[anno,right of=p0,anchor=west] (o0) {${\ghost=q_0}$};
        \draw[annolink] (o0) -- (p0);
        \node[anno,left of=p1,anchor=east] (o1) {$\begin{aligned}
        				&(\ghost=q_1\land x>0)\\
        				\lor &(\ghost=q_3\land x>1)
        			\end{aligned}$};
        \draw[annolink] (o1) -- (p1);
        \node[anno, left=10mm of p2, anchor=east,yshift=-2mm] (o2) {$\begin{aligned}
        				&(\ghost=q_2\land x>0\land y\leq x)\\
        				\lor &(\ghost=q_5\land x>1 \land y\leq x)
        			\end{aligned}$};
        \draw[annolink] (o2) -- (p2);
        \node[anno, right of=p3,anchor=north west,yshift=-4mm,xshift=0.8mm] (o3) {$\begin{aligned}
        				&(\ghost=q_2\land x>0\land y\leq x)\\
        				\lor &(\ghost=q_5\land x>1 \land y\leq x)
        			\end{aligned}$};
       \draw[annolink] (o3) -- (p3);
       \node[anno,left=10mm of p4,anchor=east] (o4) {$\begin{aligned}
       				&(\ghost=q_4\land x>1\land y<x)\\
       				\lor &(\ghost=q_6\land x>2\land y<x)
       			\end{aligned}$};
       \draw[annolink] (o4) -- (p4);
       
       \node[anno,right of=p5,anchor=west,yshift=5mm] (o5) {$\begin{aligned}
       			&(\ghost=q_1\land x>0)\\
       			\lor &(\ghost=q_2\land x>0\land y\leq x)\\
       			\lor &(\ghost=q_4\land x>1\land y<x)
       		\end{aligned}$};
       \draw[annolink] (o5) -- (p5);
       \node[anno, below of=p6,anchor=north,xshift=12mm] (o6) {$\begin{aligned}
       				&(\ghost=q_1\land x>0)\\
       				\lor &(\ghost=q_2\land x>0\land y\leq x)\\
       				\lor &(\ghost=q_4\land x>1\land y<x)
       			\end{aligned}$};
       \draw[annolink] (o6) -- (p6);
       \node[anno, below right of=p7,anchor=north west,yshift=2mm] (o7) {$\begin{aligned}
       				&(\ghost=q_3\land x>0)\\
       				\lor &(\ghost=q_5\land x>1\land y\leq x)\\
       				\lor &(\ghost=q_6\land x>2\land y<x)
       			\end{aligned}$};
       	\draw[annolink] (o7) -- (p7);

  		\node[anno,left of=e1,anchor=east] (oe1) {$\bot$};
  		\draw[annolink] (oe1) -- (e1);

  		\node[anno,right of=e2,anchor=west] (oe2) {$\bot$};
  		\draw[annolink] (oe2) -- (e2);
	\end{tikzpicture}
	\end{externalize}
	\caption{Imperial Owicki-Gries annotation constructed from the empire shown in \cref{fig:empire-example}.}
	\label{fig:imperial-og-example}
\end{figure}
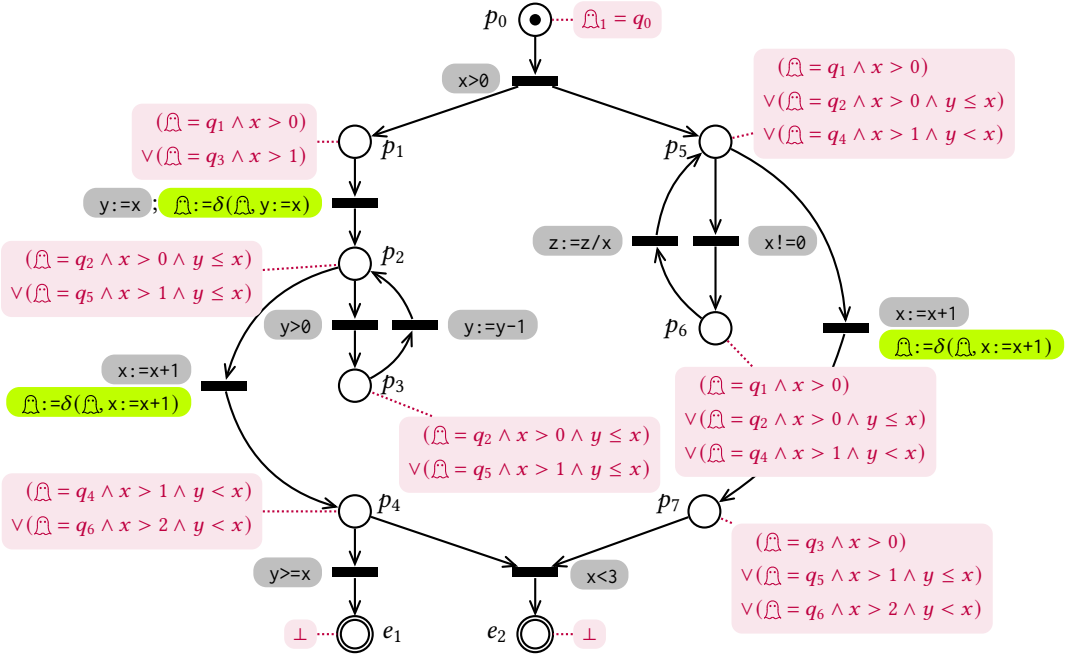

\begin{example}
	\Cref{fig:imperial-og-example} shows the imperial Owicki-Gries annotation for the empire from \cref{fig:empire-example}.
	The ghost variable $\ghost$ is initialized to $\rho(\ghost)=q_0$.
	The three ghost updates are shown in green; all other updates are omitted as the respective transitions occur only as self-loops in the empire.

	This Owicki-Gries annotation resembles that in~\cref{fig:og-example}.
	It also encodes case distinctions over which increments have already been performed.
	There are however some differences:
	\begin{enumerate}
		\item The annotations for $p_5$, $p_6$, and $p_7$ include unnecessary information about the variable \texttt{y}.
		\item Instead of two boolean ghost variables $\ghost_1$ and $\ghost_2$, we use a single ghost variable $\ghost$ that ranges over the states of the empire.
		\item There are additional ghost updates for the transitions \stsmcol{x>0}{lightgray} and \stsmcol{y:=x}{lightgray}.
	\end{enumerate}
	We address difference~(1) in \cref{sec:focus}.
	(2) is mainly a matter of preference, since the ghost variable encodes the same information. %
	We plan to address the slight overhead introduced by the additional ghost update (3) in future work.
\end{example}

It remains to show that if the empire $E$ is valid, then $\og_E$ is a valid Owicki-Gries annotation.
The following lemma presents a first key observation towards this result.

\begin{lemmarep}
	\label{lem:place-conjunction}
	Let $\hat{P}$ be a nonempty subset of a reachable marking,
	and let us write $Q_{\hat{P}}$ for the intersection $\bigcap_{p\in\hat{P}} Q_p$.
	Then it holds that
	\[
	\bigwedge_{p\in\hat{P}}\omega(p)
	\equiv
	\bigvee_{q\in Q_{\hat{P}}} \big(\ghost=q\land\law{q} \big) \ .
	\]
\end{lemmarep}
\begin{proof}
	Proven in \cref{lemma:focus-place-conjunction} for an arbitrary focus. We define $\legalIndices{q}{r} = \{1, \dots, n\}$ for (product) invariant domain $(A, \postFun) = (A_1,\postFun[1]) \otimes \ldots \otimes (A_n, \postFun[n])$, i.e. the trivial focus where each conjunct of the law is focused. Then the proof also holds here.
\end{proof}
Informally, the lemma shows that we can combine the annotations of multiple places in a marking,
and recover the same information as given by directly computing the possible empire states for the set of places $\hat{P}$.
This is crucial to show the \prop{InterferenceFreeOG} property of Owicki-Gries annotations,
where we consider $\hat{P} = \pred{t} \cup \{p\}$ for a place~$p$ that is co-marked with~$t$,
as well as for the \prop{InductiveOG} property of transitions $t$ with multiple predecessor places (``joins''), with $\hat{P} = \pred{t}$.
\Cref{lem:place-conjunction} allows us to break down these properties to individual empire states,
and we observe:
\begin{lemma}\label{lem:state-transition}
  Let $q,q'\in Q$ and $t$ a transition with $\delta(q,t)=q'$.
  The following Hoare triple holds:
  \[
    \hoareTriple{ \ghost=q \land \law{q} }{ \lambda(t);\gamma(t) }{ \ghost=q' \land \law{q'} }
  \]
\end{lemma}
\begin{inlineproof}
  This follows from \prop{inductive-law} and the definition of $\gamma$.
\end{inlineproof}

\begin{toappendix}
\begin{lemma}\label{lem:partition}
	Given a set of places $\hat{P}=\{p_1,\dots,p_n\}$ that is a subset of some reachable marking and a territory $\tau$ with $\hat{P}\subseteq places(\tau)$, it holds that there exist regions $r_1,\dots,r_n$ with $p_i\in r_i$ and $p_i\notin r_j$ (for $i\neq j$) and a territory $\hat{\tau}$ with $places(\hat{\tau})\cap\hat{P}=\emptyset$, such that $\tau = \{r_1,\dots,r_n\}\cup\hat{\tau}$.
\end{lemma}
\begin{proof}
	Since $\hat{P}=\{p_1,\dots,p_n\}$ is a subset of some reachable marking, $p_i$ and $p_j$ are co-related for $i\neq j$.
	Therefore, there cannot exist a region $r$ with $p_i\in r$ and $p_j\in r$ (for $i\neq j$), as a region may only contain places that are not co-related.
	As $\hat{P}\subseteq places(\tau)$, we can partition $\tau = \{r_1,\dots,r_n\}\cup\hat{\tau}$ with $p_i\in r_i$, $p_i\notin r_j$ (for $i\neq j$), and $places(\hat{\tau})\cap\hat{P}=\emptyset$.
\end{proof}

\begin{lemma}\label{lem:places-transition}
	Given a set of places $\hat{P}$ that is a subset of some reachable marking, for any $q\in Q, t\in T$ where $\delta(q, t)=q'$ is defined with $\pred{t}\subseteq\hat{P}\subseteq places(\terr{q})$, it holds that $(\hat{P}\setminus\pred{t})\cup\succ{t}\subseteq places(\terr{q'})$
\end{lemma}
\begin{proof}
	Let $\hat{P} = {p_1,\dots,p_n}$ be a set of places that is a subset of some reachable marking and let $q\in Q, t\in T$ where $\delta(q, t)=q'$ is defined with $\pred{t}\subseteq\hat{P}\subseteq places(\terr{q})$.
	From \cref{lem:partition} it follows that $\terr{q}$ can be partitioned as $\terr{q} = \{r_1,\dots,r_n\}\cup\hat{\tau}$ with $p_i\in r_i$, $p_i\notin r_j$ (for $i\neq j$), and $places(\hat{\tau})\cap\hat{P}=\emptyset$.

	Since $E$ is a valid empire, it follows from \prop{inductive-territory} in \cref{def:empire} that $\terr{q}\fire{t}\terr{q'}$ holds.
	This implies that the following holds:
	\[
		\terr{q'} = \bystanders{t}{\terr{q}} \cup \{ f(p) \mid p \in \succ{t} \}
	\]
	where $f : \succ{t} \to \mathbf{Regions}$ is an injective function
	such that $p\in f(p)$ for all $p\in \succ{t}$.
	
	W.l.o.g., assume that $\pred{t}=\{p_{m+1},\dots,p_n\}$.
	Therefore, it holds that $\bystanders{t}{\terr{q}}=\{r_1,\dots,r_m\}\cup\hat{\tau}$ and
	$\terr{q'}=\{r_1,\dots,r_m\}\cup\hat{\tau}\cup \{ f(p) \mid p \in \succ{t} \}$, and thus:
	\begin{align*}
		places(\terr{q'})&\supseteq r_1\cup\dots\cup r_m\cup\bigcup_{r\in\hat{\tau}}r\cup\succ{t}\\
		&\supseteq r_1\cup\dots\cup r_m\cup\succ{t}\\
		&\supseteq \{p_1,\dots,p_m\}\cup\succ{t}\\
		&\supseteq (\hat{P}\setminus\pred{t})\cup\succ{t}
	\end{align*}
\end{proof}

\end{toappendix}

Putting these observations together,
we show:

\begin{lemmarep}\label{lem:hoare-transition}
  Let $\hat{P}$ be a nonempty subset of a reachable marking,
  and let $t$ be a transition such that we have $\pred{t} \subseteq \hat{P}$.
  Then the following Hoare triple holds:
  \[
	\hoareTriple{\bigwedge_{p\in \hat{P}} \omega(p)}{\lambda(t);\gamma(t)}{\bigwedge_{p\in (\hat{P}\setminus\pred{t})\cup\succ{t}} \omega(p)}
  \]
\end{lemmarep}
\begin{proof}
	Proven in \cref{lemma:focus-hoare-transition} for an arbitrary $focus$. We define $\legalIndices{q}{r} = \{1, \dots, n\}$ for (product) invariant domain $(A, \postFun) = (A_1,\postFun[1]) \otimes \ldots \otimes (A_n, \postFun[n])$, i.e. the trivial focus where each conjunct of the law is focused. Then the proof also holds here.\\
\end{proof}

\begin{theorem}
	\label{thm:imperial-og-correct}
	If $E$ is a valid empire,
	the imperial Owicki-Gries annotation $\og_E$ is valid.
\end{theorem}
\begin{proof}
  Let $E=(Q,q_\init,\delta,\lawFun,\terrFun)$ be a valid empire, and $\og_E = (\ghostVars,\omega,\gamma,\rho)$.
  We show each of the four validity conditions separately.
  \begin{description}
    \item[\prop{InitialOG}:] Let $p\in m_\init$.
      By \prop{initial-territory}, it holds that $m_\init\in\treaty{\terr{q_\init}}$.
      Therefore, we know that $q_\init\in Q_p$.
      By \prop{initial-law},
      we have that $\law{q_\init} \equiv \top$,
      and so
      \[
        \bigwedge_{v\in\ghostVars}v=\rho(v)
        \quad\equiv\quad \ghost=q_\init
        \quad\equiv\quad \ghost=q_\init \land \law{q_\init}
        \quad\models\quad \omega(p) \ .
      \]

    \item[\prop{SafeOG}:] Let $p\in \errorPlaces$ be an error place.
      By \prop{safe}, it holds for every $q\in Q_p$
      that $\law{q} \equiv \bot$.
      Hence, we have $\omega(p) \equiv \bot$.

    \item[\prop{InductiveOG}:] Let $t$ be a transition.
      From \cref{lem:hoare-transition},
      using $\hat{P}=\pred{t}$,
      it follows that the Hoare triple $\hoareTriple{\bigwedge_{p\in\pred{t}} \omega(p)}{\lambda(t);\gamma(t)}{\bigwedge_{p\in (\pred{t}\setminus\pred{t})\cup\succ{t}} \omega(p)}$ holds.
      By simplifying the postcondition, we observe that the Hoare triple $\hoareTriple{\bigwedge_{p\in\pred{t}} \omega(p)}{\lambda(t);\gamma(t)}{\bigwedge_{p\in \succ{t}} \omega(p)}$ holds.

	\item[\prop{InterferenceFreeOG}:] Let $t$ be a transition and $p$ a place that is co-marked with $t$.
	  It follows that
	  the Hoare triple $\hoareTriple{\bigwedge_{p'\in\pred{t}\cup\{p\}} \omega(p')}{\lambda(t);\gamma(t)}{\bigwedge_{p'\in (\pred{t}\cup\{p\}\setminus\pred{t})\cup\succ{t}} \omega(p')}$ must hold,
	  by applying \cref{lem:hoare-transition} for $\hat{P}=\pred{t}\cup\{p\}$.
	  By weakening the postcondition, we conclude that the Hoare triple $\hoareTriple{\omega(p)\land\bigwedge_{p'\in\pred{t}} \omega(p')}{\lambda(t);\gamma(t)}{\omega(p)}$ also holds.
	\end{description}
	As all four validity conditions must hold, $\og_E$ is valid.
\end{proof}

\section{Constructing Empires by Symbolic Execution}
\label{sec:algo}

As \cref{sec:empires} shows how Owicki-Gries annotations can be derived from empires,
the remaining central step of our method is to construct an empire for a Petri program from a given safe invariant domain.
In this section, we present two algorithms that construct such empires step-by-step, through a symbolic execution of the program.
The first algorithm is straightforward, but produces unnecessarily large empires.
We then show how this algorithm can be improved by grouping more places into regions, resulting in the second algorithm, the construction of \emph{saturated empires}.

\subsection{A Straightforward Approach to Constructing Empires}
\label{sec:empire-naive}

In the following, let $(A,\postFun)$ be a safe invariant domain for the Petri program~$\P$.
We begin our construction by defining the type of states in the empire.
By the definition of empires, each state~$q$ of an empire must be labeled with a territory~$\terr{q}$ and a formula~$\law{q}$.
Hence, the most straightforward representation of an empire state is a pair $q = \langle \tau, \varphi \rangle$ of a territory~$\tau$ and a formula~$\varphi$,
allowing us to set $\terr{q} = \tau$ and $\law{q} = \varphi$.
Specifically, we always use a formula~$\varphi \in A$ given by the invariant domain.
We proceed by considering each of the four validity conditions for empires (\cref{def:empire-valid}) in turn,
each time choosing the simplest way to satisfy it.

Let us begin by considering the initial state $q_\init$.
Per \prop{initial-territory}, the territory $\terr{q_\init}$ must represent at least the initial marking $m_\init$.
Hence, the simplest choice is the territory which represents exactly $m_\init$ and no other markings.
Each place in $m_\init$ forms a singleton region.
\[
	\tau_\init \coloneq \big\{\, \{p\} \mid p \in m_\init \,\big\}
\]
Furthermore, as \prop{initial-law} fixes the formula $\law{q_\init}$ to be~$\top$,
we must define $q_\init \coloneq \langle \tau_\init, \top\rangle$.

The remaining states of the empire are constructed step-by-step, through symbolic execution.
Recursively, for every state $\langle \tau,\varphi \rangle$, and every transition~$t$ that is enabled in~$\tau$,
a successor state $\delta(\langle \tau,\varphi\rangle, t) = \langle \tau',\varphi' \rangle$ is computed.
To determine a suitable successor state, let us first examine the territory~$\tau'$.
The condition \prop{inductive-territory} prescribes that $\tau \fire{t} \tau'$ must hold.
This means that $\tau'$~must contain a (distinct) region for each successor place of~$t$,
and all bystander regions of~$\tau$ (i.e., regions that are not involved in the transition, in the sense that they do not contain a predecessor place of~$t$) must be preserved.
Hence, the simplest choice for~$\tau'$ is as follows:
\begin{definition}%
  Let $\tau$~be a territory, and $t \in \enabled{\tau}$ a transition.
  The \emph{replaced territory} is
  \[
    \replaced{t}{\tau} \coloneq \bystanders{t}{\tau} \cup \big\{\, \{p\} \mid p \in t^{\bullet} \,\big\} \ .
  \]
\end{definition}
\begin{lemmarep}
  \label{lem:replaces-terr}
  $\replaced{t}{\tau}$ is a territory, and it holds that $\tau \fire{t} \replaced{t}{\tau}$.
\end{lemmarep}
\begin{proof}
  We prove the two statements subsequently:
 \begin{enumerate}
   \item First, we have to prove that $\replaced{t}{\tau}$ is a territory.
     This holds if each $\tilde{m} \in \treaty{\replaced{t}{\tau}}$ is a reachable marking of $P$.
     Note that for each set $s \in \treaty{\bystanders{t}{\tau}}$, it holds that $m \in \treaty{\tau}$ for $m = s \cup \pred{t}$ and $m$ enables $t$.(We here extend the notion of treaty to arbitrary sets of regions.) This holds, since otherwise $\tau$ would not enable $t$ because there would be a region $r \in (\tau \setminus \bystanders{t}{\tau})$ with $r \cap \pred{t} = \emptyset$ which contradicts the definition of $enabled$ for territories.

     It also holds that $m \fire{t}  m'$ for $m' = s \cup \succ{t} \in \treaty{\tau}$.
     By construction, $\treaty{\replaced{t}{\tau}}$ contains exactly all such $m'$.
     Therefore, each $m' \in \treaty{\replaced{t}{\tau}}$ is a reachable marking.
     This also implies, that $\replaced{t}{\tau}$ is disjoint, since otherwise the Petri program would not be one-safe.

   \item Second, we have to prove that $\tau \fire{t} \replaced{t}{\tau}$.
     For each successor $\tilde{\tau}$ with $\tau \fire{t} \tilde{\tau}$,
     it has to hold that $\tilde{\tau} = \big( \tau \setminus \{\, r_p \mid p\in\pred{t} \,\} \big) \cup \{\, r'_{p'} \mid p' \in \succ{t} \,\}$ where $\big( \tau \setminus \{\, r_p \mid p\in\pred{t} \,\} \big) = \bystanders{t}{\tau}$.
     We can define $r'_{p'} = \{p'\}$ for each $p' \in \succ{\tau}$ which fits the requirements of $r'_{p'}$.
     Then it holds that $\tilde{\tau} = \replaced{t}{\tau}$ and therefore $\tau \fire{t} \replaced{t}{\tau}$.
 \end{enumerate}
\end{proof}
It remains only to choose a suitable successor law~$\varphi'$.
Condition \prop{inductive-law} requires that the Hoare triple $\hoareTriple{\varphi}{\lambda(t)}{\varphi'}$ must hold.
By the definition of invariant domains, this is always satisfied if we choose $\varphi' := \post{\varphi}{t}$.
Moreover, \prop{inductive-law} allows us to omit the edge from the empire if $\hoareTriple{\varphi}{\lambda(t)}{\bot}$ holds.
Hence, we let $\delta(\langle \tau,\varphi \rangle, t)$ be undefined whenever $\post{\varphi}{t} = \bot$.
The following definition summarizes our discussion:
\begin{definitionrep}
	\label{def:naive-empire}
	The \emph{na\"ive empire} $E_\textrm{na\"ive}(A,\postFun)$ corresponding to an invariant domain $(A,\postFun)$
	is the reachable part of the empire $(Q, q_\init, \delta, \lawFun, \terrFun)$,
	where
	\begin{itemize}
		\item the states $Q \coloneq \mathbf{Terr} \times A$ are pairs of territories and laws;
		\item the initial state is $q_\init \coloneq \langle \tau_\init, \top \rangle$,
		  with the territory $\tau_\init = \big\{\, \{p\} \mid p \in m_\init \,\big\}$;

		\item the partial transition function $\delta$ is defined as
		\[
		\delta(\langle \tau,\varphi\rangle,t) \coloneq \begin{cases}
			\text{undefined} & \textbf{if } t \not \in \enabled{\tau} \text{ or } \post{\varphi}{t} \equiv \bot\\
			\langle\replaced{t}{\tau}, \post{\varphi}{t}\rangle & \textbf{otherwise}
		\end{cases}
		\]
		\item and the law and territory mappings are given by $\law{\langle \tau, \varphi \rangle} \coloneq \varphi$ resp.\ $\terr{\langle \tau, \varphi \rangle} \coloneq \tau$.
		\goodbreak
	\end{itemize}
\end{definitionrep}
The restriction to the reachable part is required to ensure the condition \prop{safe}. We show:
\begin{propositionrep}
	\label{prop:naive-emp-valid}
	If the invariant domain $(A,\postFun)$ is safe,
	the na\"ive empire $E_\textrm{na\"ive}(A,\postFun)$ is valid.
\end{propositionrep}
\begin{proof}
	To show that an empire is valid, we need to prove four conditions:
	\begin{description}
	  \item[\prop{initial-law}:] Since $q_\init = \langle\tau_\init, \top\rangle$, it holds that $\law{q_\init} \equiv \top$.
	  \item[\prop{initial-territory}:]
		  The initial territory, i.e. $\terr{q_\init} = \tau_\init$, is defined as $\tau_\init \coloneq \big\{\, \{p\} \mid p \in m_\init \,\big\}$.
		  Hence, each region contains exactly one place $p$ and for each of these places it holds that $p \in m_\init$.
		  From this definition we also know, that for each $p' \in m_\init$, there exists a region in $\tau_\init$ that contains only $p'$.
		  Therefore, the single set constructed by taking exactly one place each region is equal to the initial marking.
	  \item[\prop{inductive-law}:]
	    Let $\langle \tau,\varphi\rangle\in Q$ be an arbitrary state of the empire and $t \in \enabled{\tau}$ be a transition. There are two cases for the successor state:
		\begin{description}
			\item[case 1: $\delta(\langle \tau,\varphi\rangle, t)$ is undefined.]
			  Then $t \not  \in \enabled{\tau}$ which is not possible by assumption on $t$,
			  or $\post{\varphi}{t} \equiv \bot$ holds.
			  The latter implies that the Hoare triple $\hoareTriple{\varphi}{\lambda(t)}{\bot}$ holds, by the definition of $\postFun$,
			  and thus \prop{inductive-law} is satisfied.
			\item[case 2: $\delta(\langle \tau,\varphi\rangle, t) = \langle\replaced{t}{\tau}, \post{\varphi}{t}\rangle$.]
			  Then, by definition of $\postFun$, the Hoare triple $\hoareTriple{\varphi}{\lambda(t)}{\post{\varphi}{t}}$ holds.
		\end{description}
        Hence, \prop{inductive-law} holds in both cases.
	   \item[\prop{inductive-territory}:]
	     Let $\langle \tau,\varphi\rangle\in Q$ be an arbitrary state of the empire and $t$ be a transition where $\delta(\langle \tau,\varphi\rangle, t)$ is defined.
	     Then $t$ is enabled in $\tau$, and
	     the successor territory $\tau'$ is the replaced territory $\replaced{t}{\tau}$ by definition of $\delta$.
	     As stated in lemma \cref{lem:replaces-terr}, $\tau \fire{t} \replaced{t}{\tau}$ always holds.
	   \item[\prop{safe}:]
	     Suppose there exists a reachable state $\langle \tau,\varphi \rangle$ and a region $r\in\tau$ with $r\cap\errorPlaces \neq \emptyset$.
	     Then there exists a (reachable) marking $m\in\treaty{\tau}$ with $m\cap \errorPlaces \neq \emptyset$.
	     Let $m_\init = m_0 \fire{t_0} \ldots \fire{t_n} m_n = m$ be the firing sequence reaching $m$.
	     Then we can show (by induction over $n$) that, after reading the sequence $t_0\ldots t_n$, the empire reaches state $\langle \tau,\varphi \rangle$,
	     and the abstract configuration $(m,\varphi)$ is reachable.
	     As we assume that $(A,\postFun)$ is safe, we conclude that $\varphi$ is equal to $\bot$.
	     (In fact, this implies a contradiction, as a state with law $\bot$ is never reachable by definition of $\delta$.)
	\end{description}
\end{proof}

Hence, we can construct a valid Owicki-Gries construction from the na\"ive empire.
However, as one may expect when always choosing the simplest and most straightforward solution,
the results are far from satisfactory:
The na\"ive construction does not utilize regions and territories in their full capacity.
Each region in the empire consists of a single place, and every territory thus represents only a single marking.
Hence, the imperial Owicki-Gries annotation for the na\"ive empire resembles the na\"ive Owicki-Gries annotation (\cref{def:naive-og}):
Each transition is annotated with a ghost update, and
the annotation~$\omega(p)$ of each place~$p$ is a case distinction over all reachable markings containing~$p$.
The size of the annotation is typically exponential in the size of the program.

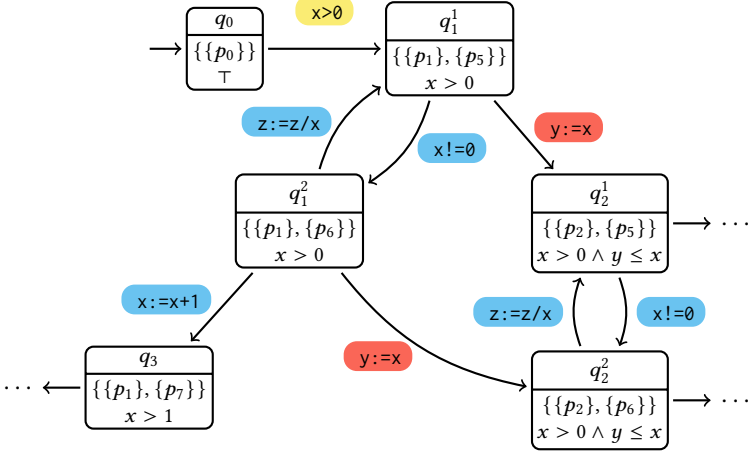
\begin{figure}
	\centering
	\begin{externalize}{naive-empire-example}
	\begin{tikzpicture}[thick,scale=1,node distance=9mm,
		every node/.style={inner sep=2pt, outer sep=2pt, align=center}
		]

		\node [empstate]                         (0) {$q_0$   \nodepart{two} $\{\{p_0\}\}$\\$\top$};
		\node [empstate,right=15mm of 0]         (1) {$q_1^1$ \nodepart{two} $\{\{p_1\}, \{p_5\}\}$\\$x>0$};
		\node [empstate,below=of 1,xshift=-20mm] (2) {$q_1^2$ \nodepart{two} $\{\{p_1\}, \{p_6\}\}$\\$x>0$};
		\node [empstate,below=of 1,xshift=20mm]  (3) {$q_2^1$ \nodepart{two} $\{\{p_2\}, \{p_5\}\}$\\$x>0\land y \leq x$};
		\node [empstate,below=of 2,xshift=-20mm] (4) {$q_3$   \nodepart{two} $\{\{p_1\}, \{p_7\}\}$\\$x>1$};
		\node [empstate,below=of 3]              (5) {$q_2^2$ \nodepart{two} $\{\{p_2\}, \{p_6\}\}$\\$x>0\land y\leq x$};

		\node [right=5mm of 5]  (dots1) {$\ldots$};
		\node [right=5mm of 3]  (dots2) {$\ldots$};
		\node [left=5mm of 4]  (dots3) {$\ldots$};

		\draw[->] ($(0)+(-1,0)$) to (0);
		\draw[->] (0) -- node[above] {$\stsmcol{x>0}{main}$} (1);
		\draw[->, bend left=20] (1) to node[right] {$\stsmcol{x!=0}{t2}$} (2);
		\draw[->] (1) -- node[right,yshift=1mm] {$\stsmcol{y:=x}{t1}$} (3);
		\draw[->, bend left=20] (2) to node[left] {$\stsmcol{z:=z/x}{t2}$} (1);
		\draw[->, bend right=20] (2) to node[left,yshift=-1mm] {$\stsmcol{y:=x}{t1}$} (5);
		\draw[->] (2) -- node[left,yshift=1mm] {$\stsmcol{x:=x+1}{t2}$} (4);
		\draw[->, bend left=20] (3) to node[right] {$\stsmcol{x!=0}{t2}$} (5);
		\draw[->, bend left=20] (5) to node[left] {$\stsmcol{z:=z/x}{t2}$} (3);

		\draw[->] (3) -- (dots2);
		\draw[->] (4) -- (dots3);
		\draw[->] (5) -- (dots1);
	\end{tikzpicture}
	\end{externalize}
	\caption{Excerpt of the na\"ive empire for the Petri program in \cref{fig:petri-program}.}
	\label{fig:naive-empire-example}
\end{figure}
For the example program from~\cref{fig:petri-program}, the excerpt of the na\"ive empire in~\cref{fig:naive-empire-example} demonstrates the disadvantages of this construction.
For each reachable marking, there exists at least one state in the empire.
As shown in \cref{fig:empire-example}, territories can represent multiple markings.
For example, state $q_1$ contains the region $\{p_5, p_6\}$ in its territory.
The na\"ive empire requires two states $q_1^1$ and $q_1^2$ to represent those markings.
For $q_2$, the na\"ive approach even requires four different states.

\subsection{Constructing Larger Regions}
To derive compact Owicki-Gries annotations,
we must construct empires with fewer states,
grouping multiple consecutive places of a thread into one region.
To this end, we introduce a condition that allows us to \emph{extend} an existing territory by executing one of its enabled transitions.
The result is a larger territory (representing a superset of markings).
We consider only \emph{sequential transitions}, with a single predecessor place and a single successor place,
as forking or joining a thread requires creating a new territory (all markings represented by a territory have the same number of tokens).
\begin{definition}%
	\label{def:extendable}
	Let $\tau$ be a territory, and let $t$ be a transition.
	We say that $\tau$ is \emph{extendable with~$t$}
	if $t$~is enabled in~$\tau$,
	we have $\pred{t} = \{p\}$, $\succ{t} = \{p'\}$ for places~$p,p'$,
	and for all places $p''$ in the region $r\in\tau$ with~$p\in r$, it holds that $p'$ and $p''$ are not co-related.
	The \emph{extended territory} is defined as
	\[
		\extended{t}{\tau} \coloneq \bystanders{t}{\tau} \cup \big\{\, r \cup \{p'\} \,\big\} \ .
	\]
\end{definition}
\begin{lemmarep}
	\label{lem:extended-terr}
	If $\tau$ is extendable with $t$, then $\extended{t}{\tau}$ is a territory and $\tau \fire{t} \extended{t}{\tau}$ holds.
\end{lemmarep}
\begin{proof}
	To show that $\extended{t}{\tau}$ is a territory, we have to show that \textbf{(1)} each set of places $r \in \extended{t}{\tau}$ is a region and \textbf{(2)} each $m \in \treaty{\extended{t}{\tau}}$ is a reachable marking of $\P$.
	\begin{description}
		\item[(1)] To show that for $r \in \extended{t}{\tau}$, $r$ is a region, first notice that $|\extended{t}{\tau} \setminus \tau| = 1$ always holds. This is true because $|\pred{t}| = 1 = |\succ{t}|$ holds by definition of extendable and therefore there exists at most one $r \in \tau$, such that $r \cap \pred{t} \neq \emptyset$ since otherwise the same place $p \in \pred{t}$ would be in two different regions of $\tau$ which cannot happen since $\tau$ is a territory by assumption and therefore can only contain reachable markings in its treaty. Because $t \in \enabled{\tau}$ such a region $r$ must exist in $\tau$. But because $\P$ is one-safe, there cannot exist a marking that contains the same place twice. Therefore, the region $r$ is unique. Hence, $|\extended{t}{\tau} \setminus \tau| = 1$. Additionally, we know that each $r' \in \extended{t}{\tau} \cap \tau$ is a region by the assumption that $\tau$ is a territory. We therefore only have to show that $r \cup \succ{t}$ is a region. Assume $\{p_1\} = \succ{t}$, we need to show that $\neg (\co{p_1}{p_2})$ for each $p_2 \in r$. This is ensured by condition 3. from \cref{def:extendable}.
		\item[(2)] To show that $m' \in \treaty{\extended{t}{\tau}}$ is a reachable marking of $\P$ we prove that there exists $m \in \treaty{\tau}$ such that $m \fire{t} m'$. This immediately implies that $m'$ is reachable since $m$ must be reachable. We define $m \coloneq (m' \setminus \succ{t}) \cup \pred{t}$. Clearly it holds $m \fire m'$ following from the definition of the firing relation, the fact that $t$ is enabled by $m$ and because $\succ{t} \subseteq m'$. The only thing left is to show that $m \in \treaty{\tau}$. As argued in \textbf{(1)}, $\tau$ and $\extended{t}{\tau}$ only have one different region which is region $r \in \tau$ respectively $(r \cup \succ{t}) \in \extended{t}{\tau}$. Therefore, for each $p \in m$ with $p \not \in \pred{t}$ it holds that there exists an unique region $r_p \in \bystanders{t}{\tau}$ such that $p \in r_p$. Since $r \in (\tau \setminus \bystanders{t}{\tau})$ we need to show that $\pred{t} \subseteq r$. This holds by condition 1. from \cref{def:extendable}.
	\end{description}
	Thus, $\extended{t}{\tau}$ is a territory.

	To show that $\tau \fire{t} \extended{t}{\tau}$ holds,
	note that $\extended{t}{\tau} = \bystanders{t}{\tau} \cup r_{\succ{t}}$ where $r_{\succ{t} = r \cup \succ{t}}$. This implies by \cref{def:territory-fires} that $\tau \fire{t} \extended{t}{\tau}$.
\end{proof}
\goodbreak

Analogous to the replaced territory, the above lemma would allow us to use $\extended{t}{\tau}$ as the territory for the successor state of a state~$\langle \tau,\varphi\rangle$,
whenever $\tau$~is extendable by~$t$.
While this would result in territories representing more than one marking,
it would still create more distinct territories (and thus states) than necessary.
The true power of extended territories is that they allow us to take shortcuts, and skip intermediate territories.
\begin{lemmarep}
  \label{lem:extended-shortcut}
  If $\tau$~is extendable with~$t$, it holds that $\extended{t}{\tau} \fire{t} \extended{t}{\tau}$.
  Moreover, if we have $\tau' \fire{t'} \tau$ and $\bystanders{t'}{\tau'} \subseteq \bystanders{t}{\tau}$,
  then it also holds that $\tau' \fire{t'}\extended{t}{\tau}$.
\end{lemmarep}
\begin{proof}
	Let $\tau$ be a territory and $t$ be a transition for which $\tau$ is extendable. We prove the two statements subsequently:
	\begin{description}
		\item[(1)] If $\tau$~is extendable with~$t$, it holds that $\extended{t}{\tau} \fire{t} \extended{t}{\tau}$.\\
		Because $\tau$~is extendable with~$t$, there exists $r \in \tau$ with $p \in r$ for $\pred{t} = \{p\}$. Additionally, it holds that $\extended{t}{\tau} \coloneq \bystanders{t}{\tau} \cup \big\{\, r \cup \{p'\} \,\big\}$ for $\succ{t} = \{p'\}$. This implies that $t \in \enabled{t}{\extended{t}{\tau}}$. Additionally, this resembles the definition of firing presented in \cref{def:territory-fires} if we define $r'_{p'} \coloneq r \cup \{p'\}$. Therefore, $\extended{t}{\tau} \fire{t} \extended{t}{\tau}$ holds.
		\item[(2)] If we have $\tau' \fire{t'} \tau$ and $\bystanders{t'}{\tau'} \subseteq \bystanders{t}{\tau}$, then it also holds that $\tau' \fire{t'}\extended{t}{\tau}$.\\
		Assume $\tau$ is extendable with $t$ because otherwise $\extended{t}{\tau}$ does not exist. Let $\pred{t} = \{p\}$ and $\succ{t} = \{p'\}$ and $r \in \tau$ be the region for which $p \in r$. Then we know that $r \not\in \bystanders{t}{\tau}$ which implies that $r \not \in \bystanders{t'}{\tau'}$ by the assumption that $\tau' \fire{t'} \tau$ and $\bystanders{t'}{\tau'} \subseteq \bystanders{t}{\tau}$. Hence, we know that $\bystanders{t'}{\tau'}  \subset \extended{t}{\tau}$ because $\bystanders{t'}{\tau'}  \subset \tau$ and extension only modifies $r$. Therefore, we can define
		\begin{align*}
		  \extended{t}{\tau} &= \bystanders{t}{\tau} \cup \big\{\, r \cup \{p'\} \,\big\}\\
		  &= \bystanders{t'}{\tau'} \cup (\bystanders{t}{\tau} \setminus \bystanders{t'}{\tau'}) \cup \big\{\, r \cup \{p'\} \,\big\}\ .
		\end{align*}
		Because $\tau' \fire{t'} \tau$ it holds that for each $r_{p'} \in (\tau \setminus \bystanders{t'}{\tau'})$ that $p' \in r_{p'}$ for an unique successor place $p' \in \succ{t'}$ by the definition of $fires$. Additionally, we know that $r$ is one such region that contains a successor place $p' \in \succ{t'}$ and by extending no places of $r$ are removed. This implies that for each $r' \in  \big((\bystanders{t}{\tau} \setminus \bystanders{t'}{\tau'}) \cup \big\{\, r \cup \{p'\} \,\big\} \big)$ each place $p \in \succ{t'}$ is contained in exactly one unique $r'$ and there exists no $r'$ with $r' \cap \succ{t'} = \emptyset$. Therefore, $\tau' \fire{t'}\extended{t}{\tau}$.
	\end{description}
\end{proof}
\begin{wrapfigure}[8]{r}{5.5cm}%
\vspace*{-1.3em}
  \begin{tikzpicture}[thick]%
    \node (0) {$\tau'$};
    \node (1) [right of=0,anchor=west] {$\tau$};
    \node (2) [right of=1,anchor=west] {$\extended{t}{\tau}$};
    \draw[->] (0) -- node[auto]{$t'$} (1);
    \draw[->,densely dotted] (1) -- node[auto]{$t$} (2);
    \draw[->,densely dashed,every loop/.style={looseness=5}] (2) edge[loop right] node[auto]{$t$} ();
    \draw[->,densely dashed] (0) edge[bend right] node[auto,swap]{$t'$} (2);
  \end{tikzpicture}
  \caption{Illustration of firing relationships implied by \cref{lem:extended-terr} (dotted line) and \cref{lem:extended-shortcut} (dashed lines).}
  \label{fig:extended-skip}
\end{wrapfigure}
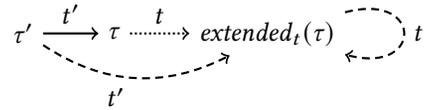
\Cref{fig:extended-skip} illustrates the firing relationships implied by the lemmas.
The key take-away from this figure is that we do not necessarily need to construct a state with the territory~$\tau$.
Instead, we can immediately construct a state with the territory~$\extended{t}{\tau}$, and use this territory for the successor of $\tau'$ under transition~$t'$.
Moreover, if $\extended{t}{\tau}$ in turn is extendable by yet another transition~$t''$,
applying the same lemma again allows us (provided the bystander condition holds) to even use the territory $\extended{t''}{\extended{t}{\tau}}$ as the successor of~$\tau'$.
This process can be repeated over and over, as long as there are extending transitions that satisfy the bystander condition.

However, when using extended territories to construct an empire, we must be careful not to extend too much.
There are two limiting factors that can prevent us from extending with a transition.
First, as we have already seen, the bystander condition must hold in order for the extended territory to remain a possible successor for the territory~$\tau'$.
This is particularly relevant when the edge from $\tau'$ to $\tau$ is not an extension.
Second, when we extend the territory for an empire state, thereby enlarging the set of markings it represents, we are claiming that the same law holds for all these markings.
To remain sound, we must be careful not to claim any law that is not proven by the given invariant domain $(A,\postFun)$.
Therefore, we only extend the territory with transitions that do not change the law.
The following definition captures these two limiting factors.
\begin{definition}%
	\label{def:saturated-successor}
  Let $\tau$ be a territory, $\varphi\in A$ be a formula, and $\mathit{rb}$ a set of regions.
  The set of \emph{saturated successors} $\saturatedSuccs{\varphi}{rb}{\tau}$
  contains all territories $\tau'$ such that (recursively) one of the following holds:
  \begin{itemize}
    \item $\tau' \in \saturatedSuccs{\varphi}{rb}{\extended{t}{\tau}}$ holds, for some transition $t$ such that $\tau$ is extendable by $t$, with $\lnot(\tau\fire{t} \tau)$, $\post{\varphi}{t} = \varphi$ and $\mathit{rb} \subseteq \bystanders{t}{\tau}$; or
    \item $\tau = \tau'$, and no such transition $t$ exists.
  \end{itemize}
\end{definition}
\begin{lemmarep}
	\label{lem:inductive-selfloop}
	If $\tau \fire{t} \tau$ and $\tau$ is extendable with $t'$, then $\extended{t'}{\tau} \fire{t} \extended{t'}{\tau}$.
\end{lemmarep}
\begin{proof}
	Because $\tau \fire{t} \tau$ and $\tau$ holds there must exist sets $\mathit{pred_t}, \mathit{succ_t} \subseteq (\tau \setminus \bystanders{t}{\tau})$ such that for each $r_1, r_2 \in \mathit{pred_t}$ where $r_1 \ne r_2$, places $p_1 \in r_1$ and $p_2 \in r_2$ for $p_1, p_2 \in \pred{t}$ such that $p_1 \ne p_2$, i.e. each region in $\mathit{pred_t}$ contains exactly one unique predecessor of the transition. The set $\mathit{succ_t}$ is defined similar, but with the successor places of the transition. Additionally, $|\mathit{pred_t}| = |\mathit{succ_t}|$ must hold, since \cref{def:territory-fires} implies that each region in $(\tau \setminus \bystanders{t}{\tau})$ contains exactly one successor and by definition of $bystanders$, each region has to contain exactly one predecessor place. Therefore, $\mathit{pred_t} = \mathit{succ_t}$. Thus, we refer to $\mathit{pred_t}, \mathit{succ_t}$ as $\mathit{firing_t}$ in the following.
	Now assume $\extended{t'}{\tau} $ does not change any region in $\mathit{firing_t}$. Then $\extended{t'}{\tau} \fire{t} \extended{t'}{\tau}$ holds immediately. If instead there exists one region $r \in \mathit{firing_t}$ that is changed by the extension then we know by definition that $(r \cup \{p\}) \in \extended{t'}{\tau}$ for $\succ{t'} = \{p\}$. This operation therefore only adds one place to region $r$. Thus, $\extended{t'}{\tau}$ still enables $t$. The extension does not modify $\bystanders{t}{\tau}$ by definition of $r$. Therefore, we only have to show that for each $r_1', r_2' \in (\extended{t'}{\tau} \setminus \bystanders{t}{\tau})$ it holds that there exist $p_1', p_2' \in \succ{t}$ and if $p_1' \ne p_2'$ then $r_1' \ne p_2'$. Additionally, for each $p\in \succ{t}$ it must hold that there exists a region in $\extended{t'}{\tau}$ that contains that place. The latter holds since it already was true for $\tau$ and extension only adds one place to $r$. The former does also immediately hold for $r_1',r_2' \ne (r \cup \{p\})$ because it was already true for $\tau$ and only $r$ changed. Assume WLOG $r_1' = (r \cup \{p\})$ then since $p_1' \in r$ for some successor place $p_1'$ it must hold that $p_1' \in r_1'$. Additionally, $p \not \in \succ{t}$ because otherwise $p$ and $p_1'$ would be co-related. Therefore, $\extended{t'}{\tau} \fire{t} \extended{t'}{\tau}$ holds.
\end{proof}
\begin{lemmarep}
	\label{lem:well-defined}
	The set $\saturatedSuccs{\varphi}{\mathit{rb}}{\tau}$ is always non-empty.
	For all territories $\tau,\tau'$ and transitions $t$ with $\tau \fire{t} \tau'$ and $\mathit{rb}=\bystanders{t}{\tau}$,
	it holds that $\tau \fire{t} \tau''$ for all $\tau'' \in \saturatedSuccs{\varphi}{\mathit{rb}}{\tau'}$.
\end{lemmarep}
\begin{proof}
	We show the two statements subsequently:
	\begin{description}
		\item[(1)] The set $\saturatedSuccs{\varphi}{\mathit{rb}}{\tau}$ is always non-empty.\\
		If there exists no transition $t$ such that $\tau$ is extendable by $t$, with $\lnot(\tau\fire{t} \tau)$, $\post{\varphi}{t} = \varphi$ and $\mathit{rb} \subseteq \bystanders{t}{\tau}$ then $\tau \in \saturatedSuccs{\varphi}{\mathit{rb}}{\tau}$ by definition, hence the set is non-empty. \\
		If at least one such $t$ exists, then $\tau' \in \saturatedSuccs{\varphi}{\mathit{rb}}{\tau}$ if $\tau' \in \saturatedSuccs{\varphi}{rb}{\extended{t}{\tau}}$. Since this is again recursively defined, we have to prove that the number of recursive steps is finite. Because the saturated successors then are defined in terms of an extended territory, it suffices to prove that through this recursive definition for any transition $t$, $\extended{t}{ \ }$ gets invoked at most one time for only one successor of $\tau$. In other words, after in some recursion step a successor of $\tau$ is extended with $t$, in no later recursion step $\extended{t}{ \ }$ will be invoked again. This holds because after one successor is extended by $t$ all later territories $\tilde{\tau}$ constructed during the recursion  satisfy $\tilde{\tau} \fire{t} \tilde{\tau}$. This holds with \cref{lem:extended-shortcut} which ensures that the extended territory has a self-loop for $t$ and inductively then holds for all successors with \cref{lem:inductive-selfloop}. Hence, the number of transitions in $T$ is finite and $\extended{t}{ \ }$gets called at most one time for $t \in T$ in the recursion. This implies that after at most $T$ recursion steps there is no transition left for which all four conditions for the recursive call hold. Then $\saturatedSuccs{\varphi}{\mathit{rb}}{\tau}$ contains at least one element by the recursive definition. Hence, it is non-empty.

		\item[(2)] For all territories $\tau,\tau'$ and transitions $t$ with $\tau \fire{t} \tau'$ and $\mathit{rb}=\bystanders{t}{\tau}$, it holds that $\tau \fire{t} \tau''$ for all $\tau'' \in \saturatedSuccs{\varphi}{\mathit{rb}}{\tau'}$. \\
		\textbf{Proof by Induction: }
		During this proof we refer to the set of transitions for which $\mathit{extended}$ is recursively called by saturate as $T' \subseteq T$.\\
		\textbf{Base Case:} Let \( \tau'' = \tau' \) and \(T' = \emptyset\).\\
		Then $\tau \fire{t} \tau'$ holds by definition.\\

		\textbf{Inductive Hypothesis: } For a $\tau$, $\tau'$ and a set of transitions $T'$, $\tau \fire{t} \tau''$ for all $\tau'' \in \saturatedSuccs{\varphi}{\mathit{rb}}{\tau'}$ holds.\\

		\textbf{Inductive Step: } Assume for $\tau'$ and $T'$ the statement $\tau \fire{t} \tau''$ for all $\tau'' \in \saturatedSuccs{\varphi}{\mathit{rb}}{\tau'}$ (I.H.) holds. Now let $t'$ be a transition that is not in $T'$ and for which the four conditions for a recursive call on $\tau''$ hold. For $T' \cup t'$ it holds that $\tau_1 \in \saturatedSuccs{\varphi}{\mathit{rb}}{\tau'}$ iff $\tau_1 \in \saturatedSuccs{\varphi}{\mathit{rb}}{\extended{t'}{\tau''}}$. Since $\tau \fire{t} \tau''$ we can apply the second statement of \cref{lem:extended-shortcut} and conclude that $\tau \fire{t} \tau_1$ does also hold.
	\end{description}
\end{proof}

Typically, the set $\saturatedSuccs{\varphi}{rb}{\tau}$ contains exactly one territory, i.e., there is a unique saturated successor,
as the order in which extending transitions $t$ are applied does not matter.
The only exception are cases in which there is no clear separation of the places into threads,
namely, when two places $p_1,p_2$ can each be reached from a common predecessor place through only sequential transitions,
but $p_1$ and $p_2$ also occur together in a reachable marking (after a different sequence of transitions).
In our benchmarks in \cref{sec:evaluation}, the saturated successor is always unique.

With the notion of saturated successors, we have a condition to extend a given territory as much as possible.
Whenever our empire construction creates a new territory $\tau$,
it computes a saturated successor $\tau'$ of $\tau$ and uses $\tau'$ in place of $\tau$.

\begin{definition}%
  \label{def:saturated-empire}
  Given a Petri program $\P$ and a safe invariant domain $(A, \postFun)$ for $\P$,
  a \emph{saturated empire} is the reachable part of an empire $E = (Q, q_\init, \delta, \lawFun, \tau)$,
  where
  \begin{itemize}
	\item the states $Q \coloneq \mathbf{Terr} \times A$ are pairs of territories and laws;
	\item the initial state $q_\init \in \{\, \langle \tau, \top \rangle \mid \tau\in \saturatedSuccs{\top}{\emptyset}{\tau_\init}\,\}$
	  consists of a saturated successor of the territory $\tau_\init$ and the law $\top$;
	\item the transition function $\delta$ satisfies
	  \[
	    \delta(\langle \tau, \varphi \rangle, t) = \begin{cases}
	      \text{undefined} & \textbf{if } t \not \in \enabled{\tau} \text{ or } \post{\varphi}{t} \equiv \bot\\
	      \langle\tau, \varphi\rangle & \textbf{if } \tau\fire{t}\tau \text{ and } \post{\varphi}{t} \equiv \varphi\\
	      \langle \tau', \varphi' \rangle & \begin{aligned}[t]\textbf{otherwise, }\text{where } &\varphi'=\post{\varphi}{t},\ \mathit{rb} = \bystanders{t}{\tau}\\ &\text{ and } \tau'\in\saturatedSuccs{\varphi'}{\mathit{rb}}{\replaced{t}{\tau}}\end{aligned}
	    \end{cases}
	  \]
	\item and the law and territory mappings are given by $\law{\langle \tau, \varphi \rangle} \coloneq \varphi$ resp.\ $\terr{\langle \tau, \varphi\rangle} \coloneq \tau$.
\end{itemize}
\end{definition}

As before, we restrict to the reachable part in order to ensure that the validity condition \prop{safe} holds.
The following lemma shows that the validity conditions \prop{inductive-territory} and \prop{inductive-law} continue to hold for our improved transition function $\delta$.
\begin{lemmarep}
	\label{lem:max-successor}
	Let $\langle \tau, \varphi \rangle$ be an empire state, and $t$ a transition,
	such that $\langle\tau', \varphi'\rangle \coloneq \delta(\langle \tau, \varphi \rangle, t)$ is defined.
	Then $\tau \fire{t} \tau'$ holds, and the Hoare triple $\hoareTriple{\varphi}{\lambda(t)}{\varphi'}$ holds.
\end{lemmarep}
\begin{proof}
	Let $\langle\tau, \varphi\rangle$ be a state and $t \in \enabled{\tau}$ be an arbitrary transition enabled in $\tau$ for which $\delta(q, t)$ is defined. There are two cases for the successor:
	\begin{description}
		\item[case 1:] $\delta(\langle\tau, \varphi\rangle, t) = \langle\tau, \varphi\rangle$. In this case $\tau \fire{t} \tau$ holds by definition of $\delta$ and $\post{\varphi}{t} \equiv \varphi$ which implies that the Hoare triple $\hoareTriple{\varphi}{\lambda(t)}{\varphi}$ holds.
		\item[case 2:] $\delta(\langle\tau, \varphi\rangle, t) = \langle \tau', \post{\varphi}{t}\rangle$. Where $\mathit{rb} = \bystanders{t}{\tau}$ and $\tau'\in\saturatedSuccs{\varphi'}{\mathit{rb}}{\replaced{t}{\tau}}$.
		As show in \cref{lem:well-defined} at least one such $\tau'$ always exists. As shown in \cref{lem:replaces-terr} it holds that $\tau \fire{t} \replaced{t}{\tau}$. Again applying \cref{lem:well-defined} reveals that for all $\tau' \in \saturatedSuccs{\varphi'}{\mathit{rb}}{\replaced{t}{\tau}}$, it holds that $\tau \fire{t} \tau'$. Because of the definition of $\postFun$, the Hoare triple $\hoareTriple{\varphi}{\lambda}{\post{\varphi}{t}}$ holds.
	\end{description}
	We have shown that the lemma holds.
\end{proof}
\begin{theoremrep}
	\label{thm:saturated-empire-valid}
	Every saturated empire constructed from a safe invariant domain $(A,\postFun)$ is valid.
\end{theoremrep}
\begin{proof}
		Similar to the proof for the na\"ive empire, we need to prove that the four conditions from \cref{def:empire} for validity hold:
	\begin{description}
		\item[\prop{initial-law}:] Since $q_\init \coloneq \langle \tau, \top \rangle$, it holds that $\law{q_\init} \equiv \top$.
		\item[\prop{initial-territory}:] The initial territory is $\tau \in \saturatedSuccs{\emptyset}{\top}{\tau_\init}$. Since $m_\init \in \treaty{\tau_\init}$
		and $\mathit{extended}$ only adds places to regions, $m_\init \in \treaty{\tau}$ does also hold.
		\item[\prop{inductive-law}:] Let $q \coloneq \langle \tau, \varphi\rangle$ be a state and $t \in \enabled{\tau}$ be a transition. If $\delta(q, t)$ is defined, \prop{inductive-law} holds as shown in \cref{lem:max-successor}. Otherwise, if $\delta(q, t)$ is undefined, this can only happen if $\post{\varphi}{t} \equiv \bot$ by definition. Hence, in this case \prop{inductive-law} also holds.
		\item[\prop{inductive-territory}:] Holds as proven in \cref{lem:max-successor} when $\delta(q, t)$ is defined.
	   \item[\prop{safe}:]
		Suppose there exists a reachable state $\langle \tau,\varphi \rangle$ and a region $r\in\tau$ with $r\cap\errorPlaces \neq \emptyset$.
		Then there exists a (reachable) marking $m\in\treaty{\tau}$ with $m\cap \errorPlaces \neq \emptyset$.
		Let $m_\init = m_0 \fire{t_0} \ldots \fire{t_n} m_n = m$ be the firing sequence reaching $m$.
		Then we can show (by induction over $n$) that, after reading the sequence $t_0\ldots t_n$, the empire reaches state $\langle \tau,\varphi \rangle$,
		and the abstract configuration $(m,\varphi)$ is reachable.
		As we assume that $(A,\postFun)$ is safe, we conclude that $\varphi$ is equal to $\bot$.
		(In fact, this implies a contradiction, as a state with law $\bot$ is never reachable by definition of $\delta$.)
	\end{description}
\end{proof}

\begin{toappendix}
\begin{corollary}
	\label{cor:no-error-states}
	For a saturated empire~$E$ and each place $p \in \errorPlaces$, there exists no state $q \in Q$ with $p \in \places{q}$.
\end{corollary}
\end{toappendix}

The empire in \cref{fig:empire-example} is a saturated empire (i.e., it can be computed using our approach) for the program in~\cref{fig:petri-program} and the invariant domain in~\cref{ex:naive-og}.
In contrast to the na\"ive empire (\cref{fig:naive-empire-example}), the saturated empire groups places and constructs larger regions, decreasing the number of states.

\section{Focus: Modular and Succinct Owicki-Gries Annotations}
\label{sec:focus}
As discussed in~\cref{sec:empires},
imperial Owicki-Gries annotations still sometimes include unnecessary information.
For instance, in \cref{fig:imperial-og-example},
the annotations of the places $p_5,p_6,p_7$ need not include any information about the variable $y$,
which can be thought of as a local variable of the left-hand thread.
So far, imperial Owicki-Gries annotations have no mechanism to split the law associated with an empire state and distribute this information among the places in the state's territory.
Splitting this information allows us to construct more compact, succinct Owicki-Gries annotations.

Towards achieving such succinct Owicki-Gries annotations,
the idea of the \emph{focus} is to make use of the modular structure inherent in the invariant domains computed by certain verification techniques.
Formally, this modular structure is captured by requiring that the invariant domain arises as a product of several smaller invariant domains.
\begin{definition}
	Let $(A_1, \postFun[1]),\dots (A_n, \postFun[n])$ be invariant domains.
	The \emph{product invariant domain} $(A, \postFun) = (A_1,\postFun[1]) \otimes \ldots \otimes (A_n, \postFun[n])$
	consists of the invariants $A \coloneq \{\, \bigwedge_{i=1}^n \varphi_i \mid \varphi_i\in A_i \,\} \cup \{\bot, \top\}$
	and the function $\postFun$ with
	\[
	  \post{\varphi}{t} \coloneq \begin{cases}
	  	\bigwedge_{i=1}^n \post[i]{\top}{t} & \textbf{if } \varphi=\top\\
	  	\bigwedge_{i=1}^n \post[i]{\varphi_i}{t} & \textbf{if } \varphi=\bigwedge_{i=1}^n \varphi_i \text{ and } \post[i]{\varphi_i}{t} \neq \bot \text{ for all } i\\
	  	\bot & \textbf{otherwise}
	  \end{cases}
	\]
\end{definition}
Product invariant domains arise from many different verification algorithms.
For instance in abstract interpretation, it is common to use product domains~\cite{Rival:IntroStaticAnalysis,Cortesi:ProductOperators} of multiple abstract domains running in parallel.
Applying an analysis with multiple abstract domains can be useful, as one abstract domain may be able to rule out certain program defects that the other can not, and vice versa.
From the fixpoint computed by the abstract interpretation engine then arises one invariant domain per abstract domain.
As another example, following a very different verification approach, \emph{trace abstraction}-based verification algorithms~\cite{Heizmann:TAR} also produce product invariant domains,
by decomposing the set of a program's error traces into multiple so-called \emph{Floyd-Hoare automata}.
Each such automaton corresponds to one invariant domain.
Underlying both these approaches is the observation that the product of the computed invariant domains may be safe,
i.e., establish the program's correctness,
even if none of the invariant domains on its own is safe.

The modularity of product invariant domains thus differs fundamentally from the notion of thread-modularity.
The individual invariant domains constitute separate proof arguments that
may prove specific sub-properties,
or establish the correctness of different program parts, i.e., subsets of interleavings (rather than particular threads).
\begin{example}
	\label{ex:product-invariant-dom}
  As discussed in the introduction,
  an algorithmic verifier might verify our example program (\cref{fig:petri-program}) by showing the two assert statements separately.
  This yields two invariant domains $(A_1,\postFun[1])$ and $(A_2,\postFun[2])$,
  where
  \[
    A_1 = \{\,\top,\ y\leq x,\ y < x,\ \bot\,\} \quad \text{ resp.\ } \quad A_2 = \{\, \top,\ x>0,\ x>1,\ x>2,\ \bot\,\}\ ,
  \]
  and $\post[i]{\varphi}{t}$ is the strongest formula in $A_i$ where $\hoareTriple{\varphi}{\lambda(t)}{\post[i]{\varphi}{t}}$ holds, for $i=1,2$,
  excepting that $\post[1]{y\leq x}{\stsmcol{x:=x+1}{t2}} = (y\leq x)$.
  The product of these domains is similar to the invariant domain given in \cref{ex:naive-og},
  (they yield the same reachable abstract configurations, and the same saturated empires).
  The domain $(A_1,\postFun[1])$ proves correctness of all interleavings reaching error place~$e_1$,
  whereas $(A_2,\postFun[2])$  proves correctness of all interleavings that reach~$e_2$.
  Hence, the product invariant domain's modularity lies in the fact that each individual domain proves one of the assertions (which differs significantly from \emph{thread}-modularity).
\end{example}

In the following, let $(A,\postFun) = (A_1,\postFun[1]) \otimes \ldots \otimes (A_n, \postFun[n])$ be a safe product invariant domain,
and let $E = (Q,q_\init, \delta, \lawFun, \terrFun)$ be a saturated empire for $(A,\postFun)$.
Intuitively, the idea of a \emph{focused} Owicki-Gries annotation is that each thread of a concurrent program carries in its annotation only the information from a subset of the invariant domains~$A_i$.
Namely, the thread must carry enough information to rule out thread-local errors.
Additionally, the thread may carry information that contributes to ruling out an error in the parent thread, after the current thread has been joined into the parent.
This selection of information is formalized by the following definition.
\begin{definition}
	\label{def:focus}
	A \emph{focus} $\legalFun$ maps each state $q$ of $E$ and region $r\in\terr{q}$
	to a subset of indices $\legalIndices{q}{r} \subseteq \{1,\ldots,n\}$,
	such that for each state $q$, with $\law{q}=\bigwedge_{i=1}^n\varphi_i$%
	\footnote{
	  In case that $\law{q}$ is not a conjunction but $\top$, let $\varphi_i=\top$ for each $i=1,\ldots,n$.
	}%
	,
	the following holds:
	\begin{description}

		\item[{\propdef[safe\textsubscript{$\legalFun$}]{B1}}:]
		  If transition~$t$ is enabled in $\terr{q}$ but proven infeasible (and has thus no edge),
		  a predecessor region~$r$ of~$t$ must have enough information to prove the infeasibility.
		  \begin{align*}
		  &t\in\enabled{\terr{q}} \land \delta(q,t) \text{ undefined}\\
		  &\qquad \Rightarrow
		  \exists r\in\terr{q}\,.\,
		  \exists i \in \{1,\ldots,n\}\,.\,
		    r\cap\pred{t} \neq \emptyset \land \post[i]{\varphi_i}{t} = \bot \land \legal{q}{r}{i}
		  \end{align*}
		\item[{\propdef[inductive-edge\textsubscript{$\legalFun$}]{B2}}:]
		  If a successor region~$r'$ of transition~$t$ has information from~$(A_i,\postFun[i])$,
		  a predecessor region~$r$ of~$t$ must already have information from~$(A_i,\postFun[i])$.
          \[
            \delta(q,t) = q' \land r'\cap\succ{t}\neq\emptyset \land \legal{q'}{r'}{i}
            \Rightarrow
            \exists r\in \terr{q}\,.\,
            r\cap\pred{t} \neq \emptyset
            \land \legal{q}{r}{i}
          \]
		\item[{\propdef[bystanders\textsubscript{$\legalFun$}]{B3}}:]
		  Executing a transition $t$ cannot increase the information of bystanders.
          \[
            \delta(q,t) = q' \land r\in\bystanders{t}{\terr{q}} \land \legal{q'}{r}{i}
            \Rightarrow
            \legal{q}{r}{i}
          \]
	\end{description}
\end{definition}

A focus can be computed from~$E$ by a fixpoint iteration.
One starts by identifying states~$q$ that are missing an outgoing edge for a transition~$t\in\enabled{\terr{q}}$,
and, following \prop{B1}, adds a suitable index~$i$ to the focus of~$q$.
Using rules \prop{B2} and \prop{B3}, focus information is propagated backwards across edges of the empire,
until a fixpoint is reached.
Different strategies for resolving the nondeterminism inherent in the choice of the predecessor region~$r$ (for transitions with multiple predecessor places) and the index~$i$ yield different, incomparable focus functions.

Given a focus, we omit non-focused conjuncts from the formulae $\omega(p)$ in the Owicki-Gries annotation constructed from the empire $E$,
resulting in a smaller and more modular annotation.
\begin{definition}
  \label{def:focus-empire-og}
  Let $\legalFun$ be a focus.
  For a state $q$ of $E$ with $\law{q}=\bigwedge_{i=1}^n\varphi_i$, let us write
  \[
    \law[p]{q} = \bigwedge \{\, \varphi_i \mid i\in\{1,\ldots,n\} \land \exists r \in \terr{q} \,.\, \legal{q}{r}{i} \land p\in r \ \}\ .
  \]
  The \emph{focused imperial Owicki-Gries annotation} $\og_{E,\legalFun} = (\ghostVars, \omega, \gamma, \rho)$
  annotates every place $p$ with
  \[
    \omega(p) = \bigvee_{q\in Q_p} \big( \ghost=q \land \law[p]{q} \big)\ .
  \]
  The components $\ghostVars$, $\gamma$ and $\rho$ are as in the imperial Owicki-Gries annotation $\og_E$ (\cref{def:empire-og}).
\end{definition}

\begin{example}
	Recall the example program in \cref{fig:petri-program} and the corresponding imperial Owicki-Gries annotation in \cref{fig:imperial-og-example}.
    As highlighted in \cref{ex:product-invariant-dom}, the invariant domain $(A_1, \postFun[1])$ is not required to prove unreachability of $e_2$.
    The variable~$y$ is local to the left-hand thread ($p_1, p_2, p_3$ and $p_4$), and information about it is only relevant for establishing unreachability of~$e_1$.
    Therefore, annotations for $p_5, p_6$ and $p_7$ do not require invariants over~$y$.
    The focused approach increases thread-modularity by removing all information about~$y$ from $\omega(p_5), \omega(p_6)$ and $\omega(p_7)$.
    For example, the annotation of~$p_5$ changes to:
	\begin{align*}
		\omega(p_5) &= (\ghost=q_1\land x>0) \lor (\ghost=q_2\land x>0) \lor (\ghost=q_4\land x>1)
	\end{align*}
\end{example}

The ``unfocused'' imperial Owicki-Gries annotation defined in~\cref{def:empire-og}
corresponds to the \emph{trivial focus},
which assigns $\legalIndices{q}{r} = \{1,\ldots,n\}$ for every~$q,r$. %
For invariant domains that are not given as a product, a focused Owicki-Gries annotation can also be computed,
by applying the above constructions with~$n=1$.
We call the resulting focus functions \emph{monolithic}, as we have $\legalIndices{q}{r}\in\big\{\emptyset, \{1\} \big\}$ for all~$q,r$.
It is however to be expected that monolithic focus functions provide only minimal advantages over the trivial focus,
as this requires that for some region in a state $q$, none of the information in~$\law{q}$ is relevant.
A typical instance of this are regions from which no error place is reachable.

The validity proof of the focused imperial Owicki-Gries annotation largely follows the same structure as for the unfocused imperial Owicki-Gries annotation in \cref{sec:empires}.
However, we must first establish the following key lemma, a generalization of \prop{inductive-law} to the focused setting:
\begin{lemmarep}
	\label{lemma:focus-law-inductive}
	Let $\hat{P}$ be a subset of a reachable marking, and let $t$ be a transition with $\pred{t} \subseteq \hat{P}$.
	Let $q,q'\in Q$ be states of $E$ with $\delta(q, t)=q'$.
	Given a set $X$ of places, we write $\law[X]{q}$ for $\bigwedge_{p\in X} \law[p]{q}$.
	The following Hoare triple holds:
	\[
	  \hoareTriple{ \law[\hat{P}]{q} }{ \lambda(t) }{ \law[(\hat{P}\setminus\pred{t})\cup\succ{t}]{q'} }
	\]
\end{lemmarep}
\begin{proofsketch}
  Let $\law{q} = \varphi_1\land\ldots\land \varphi_n$ and $\law{q'} = \varphi'_1\land\ldots\land\varphi'_n$.
  Using \prop{B2} and \prop{B3},
  as well as the conjunction rule for Hoare triples,
  the Hoare triple above can be broken down to simpler Hoare triples of the form $\hoareTriple{ \varphi_i }{ \lambda(t) }{ \varphi'_i }$.
  As $\varphi'_i = \post[i]{\varphi_i}{t}$ by construction of $E$, these Hoare triples must hold (by the definition of invariant domains).
\end{proofsketch}
\begin{proof}
  Let $\law{q} = \varphi_1\land\ldots\land \varphi_n$ and $\law{q'} = \varphi'_1\land\ldots\land\varphi'_n$.
  By the conjunction rule of Hoare triples,
  it suffices to show the individual Hoare triples $\hoareTriple{ \law[\hat{P}]{q} }{ \lambda(t) }{ \law[p']{q'} }$ for each place $p'\in (\hat{P}\setminus\pred{t})\cup\succ{t}$.
  Furthermore, to show such a Hoare triple,
  it suffices to show $\hoareTriple{ \law[\hat{P}]{q} }{ \lambda(t) }{ \varphi'_i }$ for each $i\in\{1,\ldots,n\}$ such there exists a region $r'\in\terr{q'}$ with $\legal{q'}{r'}{i}$.

  Hence, let $p'\in (\hat{P}\setminus\pred{t})\cup\succ{t}$ be such a place, and let $i\in\{1,\ldots,n\}$, $r'\in\terr{q'}$ with $\legal{q'}{r'}{i}$.
  We distinguish two cases:
  \begin{description}
  \item[case $p'\in\succ{t}$.]
    We have $r'\cap \succ{t} \neq \emptyset$, and hence, by \prop{B2}, there exists a region $r\in\terr{q}$ with $r\cap\pred{t} \neq\emptyset$ and $\legal{q}{r}{i}$.
    Let in particular $p\in r\cap\pred{t}$.
    As $\hat{P} \supseteq \pred{t}$, it follows that $p\in\hat{P}$ holds,
    and we have $\law[\hat{P}]{q} \models \law[p]{q} \models \varphi_i$.
    Given that $\hoareTriple{ \varphi_i }{ \lambda(t) }{ \varphi'_i }$ holds by construction of the empire,
    we conclude that indeed $\hoareTriple{ \law[\hat{P}]{q} }{ \lambda(t) }{ \varphi'_i }$ holds.
  \item[case $p'\notin\succ{t}$.]
    It must hold that $p'\in\hat{P}\setminus\pred{t}$.
    In particular, we have $p' \in \hat{P}$, and hence the precondition $\law[\hat{P}]{q}$ includes the conjunct $\law[p']{q}$.

    We note that $r'$ cannot be a successor region of $t$, i.e., $r' \cap \succ{t} = \emptyset$ must hold.
    Otherwise, if $p'' \in r'\cap \succ{t}$ and $m$ is the reachable marking containing $\hat{P}$, we have that $m \fire{t} m'$ for a marking $m'$ such that $p',p''\in m'$.
    Thus, $p'$ and $p''$ would be co-related, contradicting the fact that they are in the same region $r'$.

    Consequently, as $r' \in \terr{q'}$ is not in a successor region of $t$, and $\terr{q} \fire{t} \terr{q'}$ holds by \prop{inductive-territory},
    it follows that $r' \in \bystanders{t}{\terr{q}}$.
    By \prop{B3}, it follows that $\legal{q}{r}{i}$ holds.
    Therefore, we have that $\law[\hat{P}]{q} \models \law[p']{q} \models \varphi_i$, and by the fact that $\hoareTriple{ \varphi_i }{ \lambda(t) }{ \varphi'_i }$ holds,
    we conclude that indeed $\hoareTriple{ \law[\hat{P}]{q} }{ \lambda(t) }{ \varphi'_i }$ holds.
  \end{description}
\end{proof}

Having re-established this property for the focused case,
we show analogous, ``focused'' versions of \cref{lem:place-conjunction,lem:state-transition,lem:hoare-transition},
and finally arrive at the validity theorem:

\begin{toappendix}
\begin{lemma}
	\label{lemma:focus-place-conjunction}
	Let $\hat{P}$ be a nonempty subset of a reachable marking,
	and let us write $Q_{\hat{P}}$ for the intersection $\bigcap_{p\in\hat{P}} Q_p$
	and $\law[\hat{P}]{q}$ for $\bigwedge_{p\in\hat{P}} \law[p]{q}$.
	Then it holds that
	\[
	\bigwedge_{p\in\hat{P}}\omega(p)
	\equiv
	\bigvee_{q\in Q_{\hat{P}}} \big(\ghost=q\land\law[\hat{P}]{q} \big) \ .
	\]
\end{lemma}
\begin{proof}
	By induction.
	\begin{description}
		\item[Induction Base] Let $\hat{P}=\{p\}$ for some $p\in P$, then:
		\[
		\bigwedge_{p\in\hat{P}}\omega(p)\equiv\omega(p)\equiv\bigvee_{q\in Q_p} \ghost=q \land \law[p]{q}\equiv\bigvee_{q\in Q_{\hat{P}}} \ghost=q \land \law[\hat{P}]{q}
		\]
		\item[Induction Hypothesis] For any $\hat{P}$ with $|\hat{P}|\leq n$ it holds that
		\[
		\bigwedge_{p\in\hat{P}}\omega(p)\equiv\bigvee_{q\in Q_{\hat{P}}} \ghost=q\land\law[\hat{P}]{q}
		\]
		\item[Induction Step] Let $\hat{P}'\coloneq\{p'\}\cup\hat{P}$ with $|\hat{P}|\leq n$, therefore
		\begin{align*}
			\bigwedge_{p\in\hat{P}'}\omega(p)
			&\equiv\omega(p')\land\bigwedge_{p\in\hat{P}}\omega(p)
			\stackrel{IH}{\equiv}\omega(p')\land\bigvee_{q\in Q_{\hat{P}}} \ghost=q\land\law[\hat{P}]{q}\\
			&\stackrel{def.}{\equiv} \bigvee_{q'\in Q_{p'}} \ghost=q' \land \law[p']{q'}\land\bigvee_{q\in Q_{\hat{P}}} \ghost=q\land\law[\hat{P}]{q}\\
			&\equiv \bigvee_{q\in Q_{\hat{P}}, q'\in Q_{p'}} \ghost=q\land\law[\hat{P}]{q} \land \ghost=q' \land \law[p']{q'}\\
			&\stackrel{q=q'}{\equiv} \bigvee_{q\in Q_{\hat{P}}\cap Q_{p'}} \ghost=q\land\law[\hat{P}]{q}\land \law[p']{q}\\
			&\stackrel{def.}{\equiv} \bigvee_{q\in Q_{\hat{P'}}} \ghost=q\land\law[\hat{P'}]{q}
		\end{align*}
	\end{description}
\end{proof}

\begin{lemma}\label{lemma:focus-hoare-transition}
	Let $\hat{P}$ be a nonempty subset of a reachable marking,
	and let $t$ be a transition such that we have $\pred{t} \subseteq \hat{P}$.
	Then the following Hoare triple holds:
	\[
	\hoareTriple{\bigwedge_{p\in \hat{P}} \omega(p)}{\lambda(t);\gamma(t)}{\bigwedge_{p\in (\hat{P}\setminus\pred{t})\cup\succ{t}} \omega(p)}
	\]
\end{lemma}
\begin{proof}
	Pick some $q\in Q$ with $\hat{P}\subseteq places(\terr{q})$.
	From the partitioning in \cref{lem:partition}, it follows by the existence of these regions that $t\in\enabled{\terr{q}}$ holds.
	Therefore, using \prop{inductive-law} of \cref{def:empire} one of the following holds:
	\begin{description}
		\item[case 1:] $\delta(q, t)=q'$ is defined.
		
		Let $\varphi\coloneq \law[\hat{P}]{q}$ and $\varphi'\coloneq \law[(\hat{P}\setminus\pred{t})\cup\succ{t}]{q'}$.
		From \cref{lemma:focus-law-inductive} it follows that the Hoare triple
		$\hoareTriple{ \varphi }{ \lambda(t) }{ \varphi' }$ holds.
		Since $\gamma(t)$ does not read or write any program variables, the Hoare triple $\hoareTriple{ \varphi' }{\gamma(t)}{ \varphi' }$ holds.
		Therefore, by sequential composition, the Hoare triple $\hoareTriple{ \varphi }{\lambda(t); \gamma(t)}{ \varphi' }$ also holds.

		From the definition of $\gamma$, it follows that the Hoare triple $\hoareTriple{\ghost=q}{\gamma(t)}{\ghost=q'}$ holds, as $\delta(q, t)=q'$. Since $\lambda(t)$ does not read or write the ghost variable $\ghost$, the Hoare triple $\hoareTriple{\ghost=q}{\lambda(t)}{\ghost=q}$ holds, and thus also the Hoare triple $\hoareTriple{\ghost=q}{\lambda(t); \gamma(t)}{\ghost=q'}$.

		Combining both reasoning steps, we conclude that the Hoare triple \[\hoareTriple{\ghost=q\land\varphi}{\lambda(t); \gamma(t)}{\ghost=q'\land\varphi'}\] holds.

		\item[case 2:] $\delta(q, t)$ is undefined.
		
		Since $t\in\enabled{\terr{q}}$, it follows from the definition of $\delta$ in \cref{def:saturated-empire} that $\post{\law{q}}{t} \equiv \bot$.
		Then, by \prop{B1}, there exists a $p \in \pred{t}$ such that the Hoare triple $\hoareTriple{\law[p]{q}}{\lambda(t)}{\bot}$ holds.
		From $\pred{t} \subseteq \hat{P}$ it follows that the Hoare triple $\hoareTriple{\law[\hat{P}]{q}}{\lambda(t)}{\bot}$ holds.
		Since the Hoare triple $\hoareTriple{\bot}{\gamma(t)}{\bot}$ holds trivially, by sequential composition the Hoare triple $\hoareTriple{\law[\hat{P}]{q}}{\lambda(t); \gamma(t)}{\bot}$ also holds.
		By strengthening the precondition, we conclude that the Hoare triple $\hoareTriple{\ghost=q\land\law[\hat{P}]{q}}{\lambda(t); \gamma(t)}{\bot}$ holds.
	\end{description}
	From the two cases it follows using the disjunction rule that the Hoare triple
	\[\hoareTriple{\bigvee_{q\in Q_{\hat{P}}}\ghost=q\land\law[\hat{P}]{q}}{\lambda(t); \gamma(t)}{\bigvee_{q'\in Q'}\ghost=q'\land\law[(\hat{P}\setminus\pred{t})\cup\succ{t}]{q'}}\]
	holds, where $Q'\coloneq \{\delta(q,t)\mid q\in Q_{\hat{P}} \land \delta(q,t) \text{ is defined}\}$.
	From \cref{lemma:focus-place-conjunction} it follows that the equivalence
	\[\bigvee_{q\in Q_{\hat{P}}}\ghost=q\land\law[\hat{P}]{q}\equiv\bigwedge_{p\in\hat{P}}\omega(p)\] holds.
	From \cref{lem:places-transition} and \cref{lemma:focus-place-conjunction}, we obtain that the following holds:
	\[\bigvee_{q\in Q'}\ghost=q'\land\law[(\hat{P}\setminus\pred{t})\cup\succ{t}]{q'}\models\bigvee_{q'\in Q_{(\hat{P}\setminus\pred{t})\cup\succ{t}}}\ghost=q'\land\law[(\hat{P}\setminus\pred{t})\cup\succ{t}]{q'}\equiv\bigwedge_{p\in(\hat{P}\setminus\pred{t})\cup\succ{t}}\omega(p)\].

	Therefore, we conclude that the following Hoare triple holds:
	\[\hoareTriple{\bigwedge_{p\in \hat{P}} \omega(p)}{\lambda(t);\gamma(t)}{\bigwedge_{p\in (\hat{P}\setminus\pred{t})\cup\succ{t}} \omega(p)}\]
\end{proof}
\end{toappendix}

\begin{theoremrep}
	Let $(A,\postFun) = (A_1,\postFun[1]) \otimes \ldots \otimes (A_n,\postFun[n])$ be a safe invariant domain,
	and let~$E$ be a corresponding saturated empire.
	Let $\legalFun$ be a focus on~$E$.
	The focused imperial Owicki-Gries annotation~$\og_{E,\legalFun}$ is valid.
\end{theoremrep}
\begin{proofsketch}
	This proof follows the same structure as the proof of \cref{thm:imperial-og-correct}.
	In particular, \cref{lem:hoare-transition} also holds for the focused imperial Owicki-Gries annotation, ensuring that both \prop{InductiveOG} and \prop{InterferenceFreeOG} are satisfied.
	Condition~\prop{InitialOG} still holds as the ghost variables remain unchanged,
	and \prop{SafeOG} follows from property \prop{B1}. %
\end{proofsketch}
\begin{proof}
		Let $E=(Q,q_\init,\delta,\lawFun,\terrFun)$ be a valid empire, $\legalFun$ be a focus on~$E$ and $\og_{E,\legalFun} = (\ghostVars,\omega,\gamma,\rho)$.
		We show each of the four validity conditions separately.
		\begin{description}
			\item[\prop{InitialOG}:] Let $p\in m_\init$.
			By \prop{initial-territory}, it holds that $m_\init\in\treaty{\terr{q_\init}}$.
			Therefore, we know that $q_\init\in Q_p$.
			We have that 
			\[
				\law[p]{q_\init} = \bigwedge \{\, \varphi_i \mid i\in\{1,\ldots,n\} \land \exists r \in \terr{q} \,.\, \legal{q}{r}{i} \land p\in r \ \}
			\] and by \prop{initial-law} and \cref{def:focus} it holds that $\varphi_i \equiv \top$ $i = 1, \dots, n$. This implies $\law[p]{q_\init} \equiv \top$ independently from the focus
			and so
			\[
			\bigwedge_{v\in\ghostVars}v=\rho(v)
			\quad\models\quad \ghost=q_\init
			\quad\models\quad \ghost=q_\init \land \law[p]{q_\init} 
			\quad\models\quad \omega(p) \ .
			\]
			
			\item[\prop{SafeOG}:] Let $p\in \errorPlaces$ be an error place.
			By \cref{cor:no-error-states} it holds for any saturated empire $E$ that $Q_p = \emptyset$. 
			Hence, we have $\omega(p) = \bigvee_{q\in Q_p} \ghost=q \land \law[p]{q} \equiv \bigvee_{q\in \emptyset} \ghost=q \land \law[p]{q} \equiv \bot$.
			
			\item[\prop{InductiveOG}:] Let $t$ be a transition.
			From \cref{lemma:focus-hoare-transition},
			using $\hat{P}=\pred{t}$,
			it follows that the Hoare triple $\hoareTriple{\bigwedge_{p\in\pred{t}} \omega(p)}{\lambda(t);\gamma(t)}{\bigwedge_{p\in (\pred{t}\setminus\pred{t})\cup\succ{t}} \omega(p)}$ holds.
			By simplifying the postcondition, we observe that the Hoare triple $\hoareTriple{\bigwedge_{p\in\pred{t}} \omega(p)}{\lambda(t);\gamma(t)}{\bigwedge_{p\in \succ{t}} \omega(p)}$ holds.
			
			\item[\prop{InterferenceFreeOG}:] Let $t$ be a transition and $p$ a place that is co-marked with $t$.
			It follows that
			the Hoare triple $\hoareTriple{\bigwedge_{p'\in\pred{t}\cup\{p\}} \omega(p')}{\lambda(t);\gamma(t)}{\bigwedge_{p'\in (\pred{t}\cup\{p\}\setminus\pred{t})\cup\succ{t}} \omega(p')}$ must hold,
			by applying \cref{lemma:focus-hoare-transition} for $\hat{P}=\pred{t}\cup\{p\}$.
			By weakening the postcondition, we conclude that the Hoare triple $\hoareTriple{\omega(p)\land\bigwedge_{p'\in\pred{t}} \omega(p')}{\lambda(t);\gamma(t)}{\omega(p)}$ also holds.
		\end{description}
		As all four validity conditions hold, $\og_{E,\legalFun}$ is valid.
\end{proof}

\section{Evaluation}
\label{sec:evaluation}

In order to evaluate our approach empirically, we implemented the generation of Owicki-Gries annotations from safe invariant domains.
Specifically, we implemented the saturated empire construction as described in \cref{sec:algo},
as well as the computation of (focused) imperial Owicki-Gries annotations (cf.~\cref{sec:empires,sec:focus}).
As a baseline for comparison, we also implemented the na\"ive Owicki-Gries annotation as defined in \cref{def:naive-og}.
\smallskip

Our evaluation aims to answer the following questions:

\begin{description}
	\item[Q1] Is it practically feasible to compute Owicki-Gries annotations from invariant domains, with reasonable resource overhead?
	\item[Q2] How do quality and thread-modularity of (focused) imperial Owicki-Gries annotations compare to the na\"ive approach?
	\item[Q3] Can the resulting Owicki-Gries annotation be checked efficiently?
	\item[Q4] What is the impact of the focus on the quality of the resulting Owicki-Gries annotations?
\end{description}

\paragraph{Benchmarks.}
For the evaluation, we used invariant domains generated by the software model checker \uautomizer~\cite{SVCOMP24:Automizer},
which verifies programs using the trace abstraction refinement~\cite{Heizmann:TAR} approach.
This approach decomposes the program into smaller, provably-correct components, resulting in a safe product invariant domain.
\automizer translates concurrent C~programs
to Petri programs~\cite{Heizmann:Petrification}, and applies a variant of trace abstraction adapted for Petri programs~\cite{VMCAI21:Petri-TAR}.
To generate benchmarks, we ran \automizer on concurrent C~programs from SV-COMP~2025~\cite{svcomp25}, comprising \totalBenchmarks\ concurrent C~programs with various specifications;
specifically, all correct tasks from the \emph{ConcurrencySafety-Main}, \emph{ConcurrencySafety-NoOverflow} and \emph{ConcurrencySafety-NoDataRace} categories.
We used \automizer\ to convert each C~program into a Petri program and to verify its correctness.
If the verification succeeded within a time limit of 900\,s and a memory limit of 8\,GB, we extracted both the Petri program and the proof.
This resulted in \totalProofs\ benchmarks,
comprising Petri programs with between \minPlaces\ and \maxPlaces\ places and between \minTransitions\ and \maxTransitions\ transitions.

\paragraph{Evaluation Setup.}
In our evaluation, we applied different methods to generate Owicki-Gries annotations from these benchmarks.
In a second benchmark run, we furthermore validated the Owicki-Gries annotations after their generation to ensure correctness, measuring the required validation time.
Validation consists in checking the correctness conditions specified in \cref{def:og-valid}, which was done using the SMT solver Z3~\cite{z3}.
A VM replicating our evaluation setup is available as artifact~\cite{popl26:artifact}.
The evaluation was performed using the \textsc{BenchExec} framework~\cite{benchexec} on an AMD Ryzen Threadripper 3970X at 3.7\,GHz, with a time limit of 300\,s and a memory limit of 8\,GB.

\paragraph{Evaluation Results.}
\begin{table}[t]
	\caption{
	  Comparison of the generation and validation of the Owicki-Gries annotations.
	  The average times are computed on the subset of benchmarks where all approaches succeeded.
	}
	\label{tab:evaluation}
	\footnotesize
	\setlength\tabcolsep{15pt}
	\begin{tabular}{lrrr}
		\toprule
		& \multirow{2}{0.15\textwidth}[-2pt]{\centering na\"ive\\ Owicki-Gries} & \multicolumn{2}{c}{imperial Owicki-Gries} \\
		\cmidrule{3-4}
		& & unfocused & focused \\
		\midrule
		\multicolumn{4}{c}{Owicki-Gries generation}\\
		\midrule
		\# successful             & \naiveSucc        & \imperialSucc        & \modularFocusSucc    \\
		\# timeout                & \naiveTO          & \imperialTO          & \modularFocusTO      \\
		\# out of memory          & \naiveOOM         & \imperialOOM         & \modularFocusOOM     \\
		avg.\ generation time (s) & \naiveGenTime     & \imperialGenTime     & \modularFocusGenTime \\
		\midrule
		\multicolumn{4}{c}{Owicki-Gries generation \& validation}\\
		\midrule
		\# successful (valid)     & \naiveValid       & \imperialValid       & \modularFocusValid   \\
		\# timeout                & \naiveValTO       & \imperialValTO       & \modularFocusValTO   \\
		\# out of memory          & \naiveValOOM      & \imperialValOOM      & \modularFocusValOOM  \\
		\# unknown                & \naiveValUnknown  & \imperialValUnknown  & \modularFocusValUnknown \\
		avg.\ validation time (s) & \naiveValTime     & \imperialValTime     & \modularFocusValTime \\
		\bottomrule
	\end{tabular}
\end{table}

\Cref{tab:evaluation} presents a comparison of the different algorithms used to generate Owicki-Gries annotations, applied to the \totalProofs\ proofs produced by \automizer.
Across all configurations, we encountered \numCrashes\ crashes due to technical issues.
As we generated the benchmarks (Petri programs and invariant domains) beforehand, the reported times do not include the time required for the verification itself.
\Cref{fig:places2time} shows how the generation time for Owicki-Gries annotations scales with the number of places across the different approaches.
For the imperial approach, \imperialUnderOneSec\,\% of the annotations are computed within less than 1\,s (see \cref{fig:places2time-imperial,fig:places2time-modularFocus});
some outliers require up to \imperialMaxGenTime\,s.
The generation always succeeded within the given resource limits.
We conclude that the imperial Owicki-Gries annotation is a feasible postprocessing step after verification, with only minimal overhead (\textbf{Q1}).
By contrast, \cref{fig:places2time-naive} shows that the generation time for na\"ive Owicki-Gries annotations increases much faster within the number of places.
As a result, the na\"ive approach fails in \naiveOutOfRessources~cases due to resource exhaustion, and is significantly slower (see \cref{tab:evaluation}).

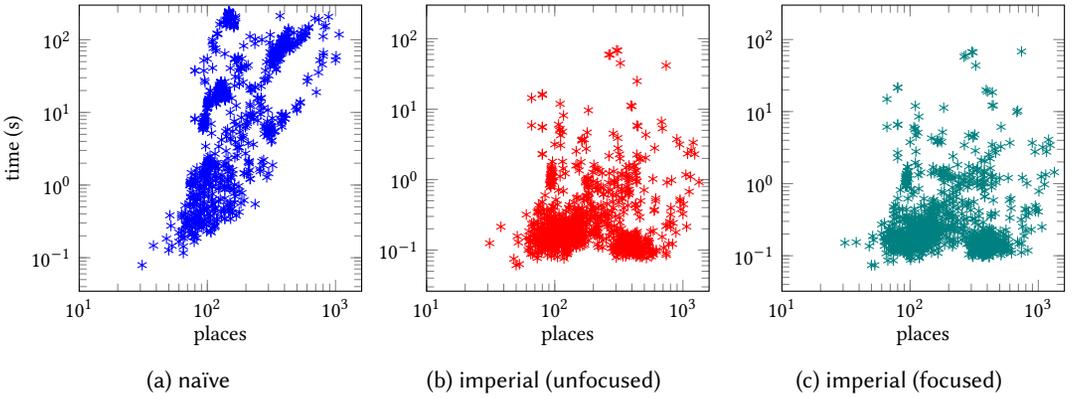
\begin{figure}[b]
	\centering
	\begin{subfigure}[t]{0.32\textwidth}
		\centering
		\hspace*{-1em}
		\begin{externalize}{places2time-naive}
		\begin{tikzpicture}
			\begin{loglogaxis}[
				xlabel=places,
				ylabel=time (\second),
				xmin=10,
				xmax=1600,
				ymin=0,
				ymax=300,
				domain=1:300,
				clip mode=individual,
				width=1.2\linewidth,
				height=1.2\linewidth,
				font=\footnotesize,
				outer sep=3pt,
				inner sep=0,
				]
				\addplot+[blue, mark=asterisk,only marks]
				table[x index=0, y index=1] {content/evaluation-results/places2time_naive.csv};
			\end{loglogaxis}
		\end{tikzpicture}
		\end{externalize}
		\vspace*{-1em}
		\caption{na\"ive}
		\label{fig:places2time-naive}
	\end{subfigure}
	\hfill
	\begin{subfigure}[t]{0.32\textwidth}
		\centering
		\begin{externalize}{places2time-imperial}
		\begin{tikzpicture}
			\begin{loglogaxis}[
				xlabel=places,
				xmin=10,
				xmax=1600,
				ymin=0,
				ymax=300,
				domain=1:300,
				clip mode=individual,
				width=1.2\linewidth,
				height=1.2\linewidth,
				font=\footnotesize,
				outer sep=3pt,
				inner sep=0,
				]
				\addplot+[red, mark=asterisk,only marks]
				table[x index=0, y index=1] {content/evaluation-results/places2time_imperial.csv};
			\end{loglogaxis}
		\end{tikzpicture}
		\end{externalize}
		\vspace*{-1em}
		\caption{imperial (unfocused)}
		\label{fig:places2time-imperial}
	\end{subfigure}
	\hfill
	\begin{subfigure}[t]{0.32\textwidth}
		\centering
		\begin{externalize}{places2time-modularFocus}
		\begin{tikzpicture}
			\begin{loglogaxis}[
				xlabel=places,
				xmin=10,
				xmax=1600,
				ymin=0,
				ymax=300,
				domain=1:300,
				clip mode=individual,
				width=1.2\linewidth,
				height=1.2\linewidth,
				font=\footnotesize,
				outer sep=3pt,
				inner sep=0,
				]
				\addplot+[teal, mark=asterisk,only marks]
				table[x index=0, y index=1] {content/evaluation-results/places2time_modularFocus.csv};
			\end{loglogaxis}
		\end{tikzpicture}
		\end{externalize}
		\vspace*{-1em}
		\caption{imperial (focused)}
		\label{fig:places2time-modularFocus}
	\end{subfigure}
	\caption{Comparison of the time needed to generate an Owicki-Gries annotation (Y-axis) for a program with a given number of places (X-axis), using the different approaches.}
	\label{fig:places2time}
\end{figure}

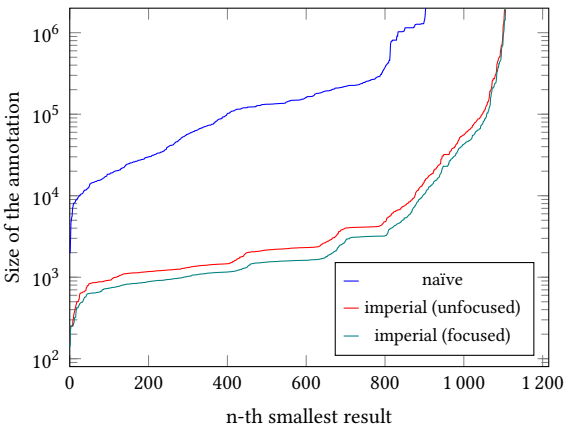
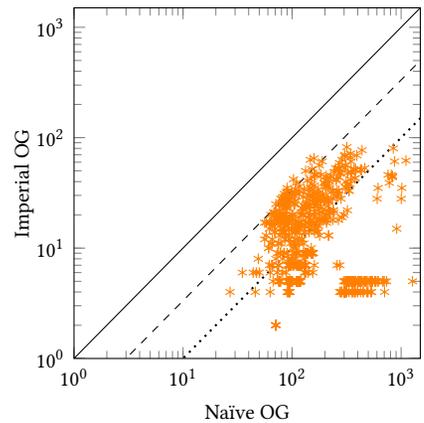
\begin{figure}[b]
	\centering
	\begin{subfigure}[t]{0.52\textwidth}
	\begin{externalize}{og-size}%
	\pgfplotsset{/pgfplots/myline legend/.style={
	  legend image code/.code={
	    \draw[mark repeat=2,mark phase=2,##1] plot coordinates {
	      (0cm,0cm) (0.2cm,0cm)
	    };%
	  }
	}}
	\begin{tikzpicture}
		\begin{semilogyaxis}[
			width=1.1\linewidth,
			height=.87\linewidth,
		    font=\footnotesize,
			/pgfplots/table/y index=0,
			/pgfplots/table/header=false,
			xlabel=n-th smallest result,
			ylabel=Size of the annotation,
			y label style={yshift=-.2em},
			log basis y=10,
			xmin=0,
			ymin=80,
			ymax=2000000,
			legend style={font=\scriptsize},
			legend entries={na\"ive, imperial (unfocused), imperial (focused)},
			legend pos=south east,
			myline legend,
			]
			\addplot+ [mark=none,blue] table [y index=1] {content/evaluation-results/formulaSizes_naive.csv};
			\addplot+ [mark=none,red] table [y index=1] {content/evaluation-results/formulaSizes_imperial.csv};
			\addplot+ [mark=none,teal] table [y index=1] {content/evaluation-results/formulaSizes_modularFocus.csv};
		\end{semilogyaxis}
	\end{tikzpicture}%
	\end{externalize}%
	\caption{Quantile plot comparing the sizes of the generated Owicki-Gries annotations.}
	\label{fig:og-size}
\end{subfigure}
\hfill
\begin{subfigure}[t]{0.43\textwidth}
	\centering
	\begin{externalize}{ghost-updates}
	\begin{tikzpicture}
		\begin{loglogaxis}[
		    width=1.2\linewidth,
		    font=\footnotesize,
			xlabel=Na\"ive OG,
			ylabel=Imperial OG,
			y label style={yshift=-.6em},
			xmin=1,
			xmax=1500,
			ymin=1,
			ymax=1500,
			domain=1:1500,
			clip mode=individual,
			axis equal image,
			]
			\addplot[black] {x};
			\addplot[dashed] {x/3};
			\addplot[thick,dotted] {x/10};
			\addplot+[orange, mark=asterisk, only marks]
			table[x index=0, y index=1] {content/evaluation-results/ghostUpdates.csv};
		\end{loglogaxis}
	\end{tikzpicture}%
	\end{externalize}%
	\caption{Comparison of the number of ghost updates
	(solid line: equal numbers, dashed line: $\times$3 fewer updates, dotted line: $\times$10 fewer updates).}
	\label{fig:ghost-updates}
\end{subfigure}
  \caption{Quality comparison for the generated Owicki-Gries annotations.}
\end{figure}

To determine the quality of the imperial Owicki-Gries annotation~(\textbf{Q2}), we used two metrics: the size of the annotation, and the number of ghost updates. %
We define the size of the annotation as the sum of the sizes of its different components, i.e., the size of the formulas for each place, the expressions in ghost updates, and the initial values of the ghost variables.
For each formula resp.\ expression, the size is given by the number of nodes in a DAG representation of its syntax.
Smaller annotations are advantageous, as they are generally easier to check and to understand for humans.
The number of ghost updates, on the other hand, can be seen as an indicator of thread-modularity.
A larger number of ghost updates allows to encode more information about interleavings, which can make the annotation less thread-modular.
\Cref{fig:og-size} shows a quantile plot with the size of the annotations on the Y-axis and the number of Owicki-Gries annotations with size up to that value on the X-axis.
In terms of annotation size, we observe that the imperial Owicki-Gries annotations are approximately \naiveToImperialSizeRatio~times smaller on average than the na\"ive annotations.
\Cref{fig:ghost-updates} compares the number of ghost updates required by the na\"ive Owicki-Gries annotation (X-axis) with the imperial Owicki-Gries annotation (Y-axis), where each data point corresponds to a task for which both approaches succeeded.
The na\"ive Owicki-Gries annotations require one ghost update for each transition, except for self-loops.
By contrast, the imperial Owicki-Gries annotations require substantially fewer ghost updates.
For \threeTimesLessGU\,\% of the tasks, the number of ghost updates is less than factor of 3 compared to the na\"ive approach, while for \tenTimesLessGU\,\%, even fewer than a factor of 10 ghost updates are required.
In all cases, the imperial approach needs strictly fewer ghost updates than the na\"ive approach, as indicated by all data points lying below the diagonal.
These results demonstrate that the imperial Owicki-Gries annotations are considerably more thread-modular (\textbf{Q2}) than the annotation from the na\"ive approach.

\Cref{tab:evaluation} also reports the number of annotations shown as valid, along with the average validation time in case it succeeded. %
None of the generated Owicki-Gries annotations were found to be invalid.
For the (focused) imperial Owicki-Gries annotations, validity checks are quite efficient (\textbf{Q3}). %
The validation is faster than the verification (which produced the given invariant domains)
for most (\percValFasterVeriOverall\,\%) of the benchmarks.
When focusing on benchmarks with a non-trivial verification time ($\geq 10\,s$),
the validation is faster in \percValFasterVeriTenSeconds\,\% of cases, and more than 3 times faster in \percValThriceFasterVeriTenSeconds\,\% of cases.
The remaining \percValSlowerVeriTenSeconds\,\% of benchmarks (with non-trivial verification time) where validation time exceeds verification time require further technical work, e.g.\ some cases may be due to the fact that our validation always uses Z3 to check Hoare triples, while the verification uses a range of different SMT solvers.
Moreover, \imperialValidPercent\,\%~of the imperial annotations were successfully validated.
In comparison, the na\"ive Owicki-Gries annotations are significantly worse to validate, with only \naiveValidPercent\,\% succeeding and a much higher average validation time.

Regarding the quality of the focused Owicki-Gries annotations (\textbf{Q4}), we observe that focus reduces the size of the annotation by \focusSmallerPercent\,\% (see \cref{fig:og-size}).
The focused approach also slightly decreases the average validation time and enables the validation of \additionalModularFocus\ additional annotations compared to the unfocused approach.
Although the improvements are not as significant as the difference compared to the na\"ive Owicki-Gries annotations, focus is still advantageous.
As shown in \cref{tab:evaluation,fig:places2time}, the generation times with focus are comparable to those without focus (with only minimal overhead).

\section{Discussion}
\label{sec:discussion}

We have presented a construction that converts interleaving-based correctness proofs of concurrent programs into thread-modular proofs, in the form of Owicki-Gries annotations.
To our knowledge, this is the first work to approach the construction of thread-modular proofs from this angle.
Previous work instead modifies verification algorithms themselves,
by a~priori restricting the space of available proofs to be in a thread-modular form.
The advantage of our approach is that the construction we present can be combined with many existing, mature verification algorithms and verification tools,
without requiring major changes to the verification itself.
Furthermore, such combinations do not suffer from completeness issues that can, in some cases, be introduced when verification algorithms themselves are modified.
Instead, the verification algorithm retains the freedom to consider as much interleaving information as required to prove the correctness.
In our evaluation, we have applied our construction to benchmarks generated from C~programs,
showing empirically that the generation of Owicki-Gries annotations incurs only a reasonable overhead,
and the generated Owicki-Gries annotations are compact and can be efficiently validated.

In the following, we discuss four key aspects of our approach, and present an outlook on possible future work for each aspect.

\paragraph{Implementation for Other Verifiers.}
To implement our approach for other verification tools,
two areas that likely need to be addressed are the \emph{program model} and the \emph{interleaving-based proofs}.

We describe our approach using a generic Petri net-based program model,
without assuming additional information such as a concept of threads, special ``fork'' and ``join'' instructions,
or other concurrency paradigms.
Our model should be generic enough to capture many models of concurrent programs, as is the case for POSIX threads-based programs~\cite{Heizmann:Petrification}.
In practice, it may also be beneficial to instead \emph{specialize} our approach to a particular representation,
taking advantage of additional information it provides.
Such a specialization can improve performance
(e.g.\ for checking validity of Owicki-Gries annotations, see \cref{sec:og-annotation})
and simplify certain definitions.
For instance, if the program is split into threads,
a territory~$\tau$ being extendable by a transition~$t$ (\cref{def:extendable}) simplifies to requiring that $t$~is enabled in~$\tau$ and $t$~is a sequential transition, i.e., not a fork or join.

Regarding the interleaving-based proofs produced by the verifier,
we base our approach on the abstract notion of invariant domains.
As discussed in~\cref{sec:interleaving-proofs},
various established program verification techniques can be seen to produce invariant domains,
including e.g.\ abstract interpretation~\cite{Cousot:AI} and predicate abstraction~\cite{Graf:Pred-Abs,Ball-Rajamani:CEGAR}.
However, some hurdles exist for techniques which produce proofs that do not directly cover all interleavings.
For instance, verification based on \emph{partial order reduction}~\cite{Farzan:ReductionForSafety,PLDI22:SoundSeq,Cassez:TARConcurrentPrograms,Wachter:MultiThreadedImpact} gives a direct proof only for a subset of interleavings,
and relies on an indirect meta-argument for other interleavings
(namely, each interleaving is in some sense equivalent to one of the directly-proven interleavings).
Constructing an Owicki-Gries proof from the output of such a verification would require encoding this meta-argument in the proof.
Whether it is possible to synthesize ghost variables encoding this meta-argument is an open question we plan to address in future work.

\paragraph{Interaction with Verification Algorithms.}
Naturally, leaving the freedom to the verification algorithm to decide how much interleaving information is considered
also allows the verifier to produce proofs (i.e., invariant domains) that rely on more interleaving information than is strictly required.
In this scenario, our presented construction will still succeed in computing an Owicki-Gries annotation,
but the annotation may have many ghost updates and (in the worst case, exponentially) large formulas.
However, our evaluation shows that in practice, our approach performs considerably better than encoding all interleaving information, at least for the verifier we evaluated.

In combining our construction with different verification algorithms, it may be worthwhile to consider heuristics that guide the verifier to avoid including unnecessary interleaving information.
Additionally, our construction may be modified to further improve how it handles such cases.
The \emph{focused} imperial Owicki-Gries annotation (\cref{sec:focus}) is a first step in this direction,
taking advantage of the modular nature of invariant domains produced by some verification algorithms.
The more fine-grained the modules are, the greater the benefit resulting from the focus.
Future work may extend the notion of focus to other modular proof combinations,
drawing on the literature on \emph{products of abstract domains}~\cite{Cortesi:ProductOperators},
or go beyond the modularity provided by the structure of the invariant domain and use other criteria to simplify the generated Owicki-Gries annotation.

\paragraph{Granularity of Ghost Variables.}
The ghost variables introduced by our construction capture an abstraction of the \emph{history}, i.e., the executed interleaving, by tracking the state of the empire.
Currently, all this interleaving information is encoded in a single ghost variable $\ghost$.
It may be possible to split this information across multiple ghost variables, each capturing a particular aspect of the interleaving information.
While it may seem counterintuitive to increase the number of ghost variables,
this would enable annotations to refer only to specific aspects of interleaving information relevant to a particular place.
For instance, one aspect of the interleaving information may be to track whether some thread is currently accessing a shared resource.
Threads that never access the same resource need not consider this aspect in their annotation.
Ideally, this could also allow identifying certain specialized kinds of ghost variables, such as \emph{resource invariants}.

\paragraph{Practical Application.}
As our benchmark Petri programs were generated from C~code, it is a natural next step to translate the generated Owicki-Gries annotations to the C~level
and export them as proofs in a format readable by program verifiers.
One such format has been defined and is supported by many tools~\cite{Ayaziova:Witnesses20}, and an extension to concurrent programs has recently been proposed~\cite{VMCAI24:Ghost-Witnesses}.
Beyond technical hurdles, exporting proofs in such a format requires dealing with certain challenges that are abstracted away in our Petri program model, such as variable scope;
in particular, annotations referring to local variables of other threads.
Once these challenges are addressed, proofs generated by one tool can be validated or reused by other tools, enabling increased confidence in the verification results and fruitful cooperation between verifiers.

\paragraph{Conclusion}
We have presented a novel approach to constructing thread-modular proofs, in the style of \citeauthor{Owicki-Gries}, from interleaving-based proofs found by verification algorithms. %
Our approach enables a broader application of the resulting correctness proofs especially in combination with existing verification tools.
The key concept of our approach are \emph{empires}, which concisely capture the necessary interleaving information.
This interleaving information is then automatically encoded in ghost variables.
Our evaluation shows that our approach is efficient in practice and %
produces significantly more compact thread-modular proofs compared to a baseline.
\goodbreak

\section*{Data Availability Statement}
The artifact for this paper (a VM containing our implementation, benchmarks and evaluation setup) is available on Zenodo~\cite{popl26:artifact}.
\ifextended\else Proofs of our results can be found in the extended version~\cite{popl26:extendedVersion}.\fi
\goodbreak

\section*{Acknowledgements}

This work was partially funded by the Deutsche Forschungsgemeinschaft (DFG) under the project number 503812980.

\bibliographystyle{ACM-Reference-Format}
\bibliography{./references}


\begin{thebibliography}{48}


\ifx \showCODEN    \undefined \def \showCODEN     #1{\unskip}     \fi
\ifx \showISBNx    \undefined \def \showISBNx     #1{\unskip}     \fi
\ifx \showISBNxiii \undefined \def \showISBNxiii  #1{\unskip}     \fi
\ifx \showISSN     \undefined \def \showISSN      #1{\unskip}     \fi
\ifx \showLCCN     \undefined \def \showLCCN      #1{\unskip}     \fi
\ifx \shownote     \undefined \def \shownote      #1{#1}          \fi
\ifx \showarticletitle \undefined \def \showarticletitle #1{#1}   \fi
\ifx \showURL      \undefined \def \showURL       {\relax}        \fi
\providecommand\bibfield[2]{#2}
\providecommand\bibinfo[2]{#2}
\providecommand\natexlab[1]{#1}
\providecommand\showeprint[2][]{arXiv:#2}

\bibitem[Apt et~al\mbox{.}(2009)]%
        {Apt-Olderog:Book}
\bibfield{author}{\bibinfo{person}{Krzysztof~R. Apt}, \bibinfo{person}{Frank~S.
  de Boer}, {and} \bibinfo{person}{Ernst{-}R{\"{u}}diger Olderog}.}
  \bibinfo{year}{2009}\natexlab{}.
\newblock \bibinfo{booktitle}{\emph{Verification of Sequential and Concurrent
  Programs}}.
\newblock \bibinfo{publisher}{Springer}.
\newblock
\showISBNx{978-1-84882-744-8}
\href{https://doi.org/10.1007/978-1-84882-745-5}{doi:\nolinkurl{10.1007/978-1-84882-745-5}}


\bibitem[Ayaziov{\'{a}} et~al\mbox{.}(2024)]%
        {Ayaziova:Witnesses20}
\bibfield{author}{\bibinfo{person}{Paul{\'{\i}}na Ayaziov{\'{a}}},
  \bibinfo{person}{Dirk Beyer}, \bibinfo{person}{Marian~Lingsch Rosenfeld},
  \bibinfo{person}{Martin Spiessl}, {and} \bibinfo{person}{Jan Strejcek}.}
  \bibinfo{year}{2024}\natexlab{}.
\newblock \showarticletitle{Software Verification Witnesses 2.0}. In
  \bibinfo{booktitle}{\emph{Model Checking Software - 30th International
  Symposium, {SPIN} 2024, Luxembourg City, Luxembourg, April 8-9, 2024,
  Proceedings}} \emph{(\bibinfo{series}{Lecture Notes in Computer Science},
  Vol.~\bibinfo{volume}{14624})}, \bibfield{editor}{\bibinfo{person}{Thomas
  Neele} {and} \bibinfo{person}{Anton Wijs}} (Eds.).
  \bibinfo{publisher}{Springer}, \bibinfo{pages}{184--203}.
\newblock
\href{https://doi.org/10.1007/978-3-031-66149-5\_11}{doi:\nolinkurl{10.1007/978-3-031-66149-5\_11}}


\bibitem[Ball and Rajamani(2001)]%
        {Ball-Rajamani:CEGAR}
\bibfield{author}{\bibinfo{person}{Thomas Ball} {and}
  \bibinfo{person}{Sriram~K. Rajamani}.} \bibinfo{year}{2001}\natexlab{}.
\newblock \showarticletitle{Automatically Validating Temporal Safety Properties
  of Interfaces}. In \bibinfo{booktitle}{\emph{Model Checking Software, 8th
  International {SPIN} Workshop, Toronto, Canada, May 19-20, 2001,
  Proceedings}} \emph{(\bibinfo{series}{Lecture Notes in Computer Science},
  Vol.~\bibinfo{volume}{2057})}, \bibfield{editor}{\bibinfo{person}{Matthew~B.
  Dwyer}} (Ed.). \bibinfo{publisher}{Springer}, \bibinfo{pages}{103--122}.
\newblock
\href{https://doi.org/10.1007/3-540-45139-0\_7}{doi:\nolinkurl{10.1007/3-540-45139-0\_7}}


\bibitem[Beyer et~al\mbox{.}(2019)]%
        {benchexec}
\bibfield{author}{\bibinfo{person}{Dirk Beyer}, \bibinfo{person}{Stefan
  L{\"{o}}we}, {and} \bibinfo{person}{Philipp Wendler}.}
  \bibinfo{year}{2019}\natexlab{}.
\newblock \showarticletitle{Reliable benchmarking: requirements and solutions}.
\newblock \bibinfo{journal}{\emph{Int. J. Softw. Tools Technol. Transf.}}
  \bibinfo{volume}{21}, \bibinfo{number}{1} (\bibinfo{year}{2019}),
  \bibinfo{pages}{1--29}.
\newblock
\href{https://doi.org/10.1007/S10009-017-0469-Y}{doi:\nolinkurl{10.1007/S10009-017-0469-Y}}


\bibitem[Beyer and Strejcek(2022)]%
        {Beyer:Witnesses-Where-We-Go}
\bibfield{author}{\bibinfo{person}{Dirk Beyer} {and} \bibinfo{person}{Jan
  Strejcek}.} \bibinfo{year}{2022}\natexlab{}.
\newblock \showarticletitle{Case Study on Verification-Witness Validators:
  Where We Are and Where We Go}. In \bibinfo{booktitle}{\emph{Static Analysis -
  29th International Symposium, {SAS} 2022, Auckland, New Zealand, December
  5-7, 2022, Proceedings}} \emph{(\bibinfo{series}{Lecture Notes in Computer
  Science}, Vol.~\bibinfo{volume}{13790})},
  \bibfield{editor}{\bibinfo{person}{Gagandeep Singh} {and}
  \bibinfo{person}{Caterina Urban}} (Eds.). \bibinfo{publisher}{Springer},
  \bibinfo{pages}{160--174}.
\newblock
\href{https://doi.org/10.1007/978-3-031-22308-2\_8}{doi:\nolinkurl{10.1007/978-3-031-22308-2\_8}}


\bibitem[Beyer and Strejcek(2025)]%
        {svcomp25}
\bibfield{author}{\bibinfo{person}{Dirk Beyer} {and} \bibinfo{person}{Jan
  Strejcek}.} \bibinfo{year}{2025}\natexlab{}.
\newblock \showarticletitle{Improvements in Software Verification and Witness
  Validation: {SV-COMP} 2025}. In \bibinfo{booktitle}{\emph{Tools and
  Algorithms for the Construction and Analysis of Systems - 31st International
  Conference, {TACAS} 2025, Held as Part of the International Joint Conferences
  on Theory and Practice of Software, {ETAPS} 2025, Hamilton, ON, Canada, May
  3-8, 2025, Proceedings, Part {III}}} \emph{(\bibinfo{series}{Lecture Notes in
  Computer Science}, Vol.~\bibinfo{volume}{15698})},
  \bibfield{editor}{\bibinfo{person}{Arie Gurfinkel} {and}
  \bibinfo{person}{Marijn Heule}} (Eds.). \bibinfo{publisher}{Springer},
  \bibinfo{pages}{151--186}.
\newblock
\href{https://doi.org/10.1007/978-3-031-90660-2\_9}{doi:\nolinkurl{10.1007/978-3-031-90660-2\_9}}


\bibitem[Brookes(2007)]%
        {Brookes:CSL-Semantics}
\bibfield{author}{\bibinfo{person}{Stephen Brookes}.}
  \bibinfo{year}{2007}\natexlab{}.
\newblock \showarticletitle{A semantics for concurrent separation logic}.
\newblock \bibinfo{journal}{\emph{Theor. Comput. Sci.}} \bibinfo{volume}{375},
  \bibinfo{number}{1-3} (\bibinfo{year}{2007}), \bibinfo{pages}{227--270}.
\newblock
\href{https://doi.org/10.1016/J.TCS.2006.12.034}{doi:\nolinkurl{10.1016/J.TCS.2006.12.034}}


\bibitem[Brookes and O'Hearn(2016)]%
        {Brookes:CSL-2016}
\bibfield{author}{\bibinfo{person}{Stephen Brookes} {and}
  \bibinfo{person}{Peter~W. O'Hearn}.} \bibinfo{year}{2016}\natexlab{}.
\newblock \showarticletitle{Concurrent separation logic}.
\newblock \bibinfo{journal}{\emph{{ACM} {SIGLOG} News}} \bibinfo{volume}{3},
  \bibinfo{number}{3} (\bibinfo{year}{2016}), \bibinfo{pages}{47--65}.
\newblock
\href{https://doi.org/10.1145/2984450.2984457}{doi:\nolinkurl{10.1145/2984450.2984457}}


\bibitem[Cassez and Ziegler(2015)]%
        {Cassez:TARConcurrentPrograms}
\bibfield{author}{\bibinfo{person}{Franck Cassez} {and} \bibinfo{person}{Frowin
  Ziegler}.} \bibinfo{year}{2015}\natexlab{}.
\newblock \showarticletitle{Verification of Concurrent Programs Using Trace
  Abstraction Refinement}. In \bibinfo{booktitle}{\emph{Logic for Programming,
  Artificial Intelligence, and Reasoning - 20th International Conference,
  {LPAR-20} 2015, Suva, Fiji, November 24-28, 2015, Proceedings}}
  \emph{(\bibinfo{series}{Lecture Notes in Computer Science},
  Vol.~\bibinfo{volume}{9450})}, \bibfield{editor}{\bibinfo{person}{Martin
  Davis}, \bibinfo{person}{Ansgar Fehnker}, \bibinfo{person}{Annabelle McIver},
  {and} \bibinfo{person}{Andrei Voronkov}} (Eds.).
  \bibinfo{publisher}{Springer}, \bibinfo{pages}{233--248}.
\newblock
\href{https://doi.org/10.1007/978-3-662-48899-7\_17}{doi:\nolinkurl{10.1007/978-3-662-48899-7\_17}}


\bibitem[Cortesi et~al\mbox{.}(2013)]%
        {Cortesi:ProductOperators}
\bibfield{author}{\bibinfo{person}{Agostino Cortesi}, \bibinfo{person}{Giulia
  Costantini}, {and} \bibinfo{person}{Pietro Ferrara}.}
  \bibinfo{year}{2013}\natexlab{}.
\newblock \showarticletitle{A Survey on Product Operators in Abstract
  Interpretation}. In \bibinfo{booktitle}{\emph{Semantics, Abstract
  Interpretation, and Reasoning about Programs: Essays Dedicated to David A.
  Schmidt on the Occasion of his Sixtieth Birthday, Manhattan, Kansas, USA,
  19-20th September 2013}} \emph{(\bibinfo{series}{{EPTCS}},
  Vol.~\bibinfo{volume}{129})}, \bibfield{editor}{\bibinfo{person}{Anindya
  Banerjee}, \bibinfo{person}{Olivier Danvy}, \bibinfo{person}{Kyung{-}Goo
  Doh}, {and} \bibinfo{person}{John Hatcliff}} (Eds.).
  \bibinfo{pages}{325--336}.
\newblock
\href{https://doi.org/10.4204/EPTCS.129.19}{doi:\nolinkurl{10.4204/EPTCS.129.19}}


\bibitem[Cousot and Cousot(1977)]%
        {Cousot:AI}
\bibfield{author}{\bibinfo{person}{Patrick Cousot} {and}
  \bibinfo{person}{Radhia Cousot}.} \bibinfo{year}{1977}\natexlab{}.
\newblock \showarticletitle{Abstract Interpretation: A Unified Lattice Model
  for Static Analysis of Programs by Construction or Approximation of
  Fixpoints}. In \bibinfo{booktitle}{\emph{Conference Record of the Fourth
  {ACM} Symposium on Principles of Programming Languages, Los Angeles,
  California, USA, January 1977}}, \bibfield{editor}{\bibinfo{person}{Robert~M.
  Graham}, \bibinfo{person}{Michael~A. Harrison}, {and} \bibinfo{person}{Ravi
  Sethi}} (Eds.). \bibinfo{publisher}{{ACM}}, \bibinfo{pages}{238--252}.
\newblock
\href{https://doi.org/10.1145/512950.512973}{doi:\nolinkurl{10.1145/512950.512973}}


\bibitem[de~Moura and Bj{\o}rner(2008)]%
        {z3}
\bibfield{author}{\bibinfo{person}{Leonardo~Mendon{\c{c}}a de Moura} {and}
  \bibinfo{person}{Nikolaj~S. Bj{\o}rner}.} \bibinfo{year}{2008}\natexlab{}.
\newblock \showarticletitle{{Z3:} An Efficient {SMT} Solver}. In
  \bibinfo{booktitle}{\emph{Tools and Algorithms for the Construction and
  Analysis of Systems, 14th International Conference, {TACAS} 2008, Held as
  Part of the Joint European Conferences on Theory and Practice of Software,
  {ETAPS} 2008, Budapest, Hungary, March 29-April 6, 2008. Proceedings}}
  \emph{(\bibinfo{series}{Lecture Notes in Computer Science},
  Vol.~\bibinfo{volume}{4963})}, \bibfield{editor}{\bibinfo{person}{C.~R.
  Ramakrishnan} {and} \bibinfo{person}{Jakob Rehof}} (Eds.).
  \bibinfo{publisher}{Springer}, \bibinfo{pages}{337--340}.
\newblock
\href{https://doi.org/10.1007/978-3-540-78800-3\_24}{doi:\nolinkurl{10.1007/978-3-540-78800-3\_24}}


\bibitem[Dietsch et~al\mbox{.}(2021)]%
        {VMCAI21:Petri-TAR}
\bibfield{author}{\bibinfo{person}{Daniel Dietsch}, \bibinfo{person}{Matthias
  Heizmann}, \bibinfo{person}{Dominik Klumpp}, \bibinfo{person}{Mehdi Naouar},
  \bibinfo{person}{Andreas Podelski}, {and} \bibinfo{person}{Claus
  Sch{\"{a}}tzle}.} \bibinfo{year}{2021}\natexlab{}.
\newblock \showarticletitle{Verification of Concurrent Programs Using Petri Net
  Unfoldings}. In \bibinfo{booktitle}{\emph{Verification, Model Checking, and
  Abstract Interpretation - 22nd International Conference, {VMCAI} 2021,
  Copenhagen, Denmark, January 17-19, 2021, Proceedings}}
  \emph{(\bibinfo{series}{Lecture Notes in Computer Science},
  Vol.~\bibinfo{volume}{12597})}, \bibfield{editor}{\bibinfo{person}{Fritz
  Henglein}, \bibinfo{person}{Sharon Shoham}, {and} \bibinfo{person}{Yakir
  Vizel}} (Eds.). \bibinfo{publisher}{Springer}, \bibinfo{pages}{174--195}.
\newblock
\href{https://doi.org/10.1007/978-3-030-67067-2\_9}{doi:\nolinkurl{10.1007/978-3-030-67067-2\_9}}


\bibitem[Erhard et~al\mbox{.}(2025)]%
        {VMCAI24:Ghost-Witnesses}
\bibfield{author}{\bibinfo{person}{Julian Erhard}, \bibinfo{person}{Manuel
  Bentele}, \bibinfo{person}{Matthias Heizmann}, \bibinfo{person}{Dominik
  Klumpp}, \bibinfo{person}{Simmo Saan}, \bibinfo{person}{Frank
  Sch{\"{u}}ssele}, \bibinfo{person}{Michael Schwarz}, \bibinfo{person}{Helmut
  Seidl}, \bibinfo{person}{Sarah Tilscher}, {and} \bibinfo{person}{Vesal
  Vojdani}.} \bibinfo{year}{2025}\natexlab{}.
\newblock \showarticletitle{Correctness Witnesses for Concurrent Programs:
  Bridging the Semantic Divide with Ghosts}. In
  \bibinfo{booktitle}{\emph{Verification, Model Checking, and Abstract
  Interpretation - 26th International Conference, {VMCAI} 2025, Denver, CO,
  USA, January 20-21, 2025, Proceedings, Part {I}}}
  \emph{(\bibinfo{series}{Lecture Notes in Computer Science},
  Vol.~\bibinfo{volume}{15529})},
  \bibfield{editor}{\bibinfo{person}{Shankaranarayanan Krishna},
  \bibinfo{person}{Sriram Sankaranarayanan}, {and} \bibinfo{person}{Ashutosh
  Trivedi}} (Eds.). \bibinfo{publisher}{Springer}, \bibinfo{pages}{74--100}.
\newblock
\href{https://doi.org/10.1007/978-3-031-82700-6\_4}{doi:\nolinkurl{10.1007/978-3-031-82700-6\_4}}


\bibitem[Eriksson et~al\mbox{.}(2017)]%
        {Eriksson:Unsolvability-Certificates}
\bibfield{author}{\bibinfo{person}{Salom{\'{e}} Eriksson},
  \bibinfo{person}{Gabriele R{\"{o}}ger}, {and} \bibinfo{person}{Malte
  Helmert}.} \bibinfo{year}{2017}\natexlab{}.
\newblock \showarticletitle{Unsolvability Certificates for Classical Planning}.
  In \bibinfo{booktitle}{\emph{Proceedings of the Twenty-Seventh International
  Conference on Automated Planning and Scheduling, {ICAPS} 2017, Pittsburgh,
  Pennsylvania, USA, June 18-23, 2017}},
  \bibfield{editor}{\bibinfo{person}{Laura Barbulescu}, \bibinfo{person}{Jeremy
  Frank}, \bibinfo{person}{Mausam}, {and} \bibinfo{person}{Stephen~F. Smith}}
  (Eds.). \bibinfo{publisher}{{AAAI} Press}, \bibinfo{pages}{88--97}.
\newblock
\urldef\tempurl%
\url{https://aaai.org/ocs/index.php/ICAPS/ICAPS17/paper/view/15734}
\showURL{%
\tempurl}


\bibitem[Farzan et~al\mbox{.}(2014)]%
        {Farzan:Proofs-Count}
\bibfield{author}{\bibinfo{person}{Azadeh Farzan}, \bibinfo{person}{Zachary
  Kincaid}, {and} \bibinfo{person}{Andreas Podelski}.}
  \bibinfo{year}{2014}\natexlab{}.
\newblock \showarticletitle{Proofs that count}. In
  \bibinfo{booktitle}{\emph{The 41st Annual {ACM} {SIGPLAN-SIGACT} Symposium on
  Principles of Programming Languages, {POPL} '14, San Diego, CA, USA, January
  20-21, 2014}}, \bibfield{editor}{\bibinfo{person}{Suresh Jagannathan} {and}
  \bibinfo{person}{Peter Sewell}} (Eds.). \bibinfo{publisher}{{ACM}},
  \bibinfo{pages}{151--164}.
\newblock
\href{https://doi.org/10.1145/2535838.2535885}{doi:\nolinkurl{10.1145/2535838.2535885}}


\bibitem[Farzan et~al\mbox{.}(2022)]%
        {PLDI22:SoundSeq}
\bibfield{author}{\bibinfo{person}{Azadeh Farzan}, \bibinfo{person}{Dominik
  Klumpp}, {and} \bibinfo{person}{Andreas Podelski}.}
  \bibinfo{year}{2022}\natexlab{}.
\newblock \showarticletitle{Sound sequentialization for concurrent program
  verification}. In \bibinfo{booktitle}{\emph{{PLDI} '22: 43rd {ACM} {SIGPLAN}
  International Conference on Programming Language Design and Implementation,
  San Diego, CA, USA, June 13 - 17, 2022}},
  \bibfield{editor}{\bibinfo{person}{Ranjit Jhala} {and} \bibinfo{person}{Isil
  Dillig}} (Eds.). \bibinfo{publisher}{{ACM}}, \bibinfo{pages}{506--521}.
\newblock
\href{https://doi.org/10.1145/3519939.3523727}{doi:\nolinkurl{10.1145/3519939.3523727}}


\bibitem[Farzan and Vandikas(2020)]%
        {Farzan:ReductionForSafety}
\bibfield{author}{\bibinfo{person}{Azadeh Farzan} {and}
  \bibinfo{person}{Anthony Vandikas}.} \bibinfo{year}{2020}\natexlab{}.
\newblock \showarticletitle{Reductions for safety proofs}.
\newblock \bibinfo{journal}{\emph{Proc. {ACM} Program. Lang.}}
  \bibinfo{volume}{4}, \bibinfo{number}{{POPL}} (\bibinfo{year}{2020}),
  \bibinfo{pages}{13:1--13:28}.
\newblock
\href{https://doi.org/10.1145/3371081}{doi:\nolinkurl{10.1145/3371081}}


\bibitem[Froleyks et~al\mbox{.}(2024)]%
        {Froyleks:Certifying-PA}
\bibfield{author}{\bibinfo{person}{Nils Froleyks}, \bibinfo{person}{Emily Yu},
  \bibinfo{person}{Armin Biere}, {and} \bibinfo{person}{Keijo Heljanko}.}
  \bibinfo{year}{2024}\natexlab{}.
\newblock \showarticletitle{Certifying Phase Abstraction}. In
  \bibinfo{booktitle}{\emph{Automated Reasoning - 12th International Joint
  Conference, {IJCAR} 2024, Nancy, France, July 3-6, 2024, Proceedings, Part
  {I}}} \emph{(\bibinfo{series}{Lecture Notes in Computer Science},
  Vol.~\bibinfo{volume}{14739})}, \bibfield{editor}{\bibinfo{person}{Christoph
  Benzm{\"{u}}ller}, \bibinfo{person}{Marijn J.~H. Heule}, {and}
  \bibinfo{person}{Renate~A. Schmidt}} (Eds.). \bibinfo{publisher}{Springer},
  \bibinfo{pages}{284--303}.
\newblock
\href{https://doi.org/10.1007/978-3-031-63498-7\_17}{doi:\nolinkurl{10.1007/978-3-031-63498-7\_17}}


\bibitem[Froleyks et~al\mbox{.}(2025)]%
        {Froyleks:Certificates-HWMC}
\bibfield{author}{\bibinfo{person}{Nils Froleyks}, \bibinfo{person}{Emily Yu},
  \bibinfo{person}{Mathias Preiner}, \bibinfo{person}{Armin Biere}, {and}
  \bibinfo{person}{Keijo Heljanko}.} \bibinfo{year}{2025}\natexlab{}.
\newblock \showarticletitle{Introducing Certificates to the Hardware Model
  Checking Competition}.
\newblock   \bibinfo{volume}{15931} (\bibinfo{year}{2025}),
  \bibinfo{pages}{281--295}.
\newblock
\href{https://doi.org/10.1007/978-3-031-98668-0\_14}{doi:\nolinkurl{10.1007/978-3-031-98668-0\_14}}


\bibitem[Graf and Sa{\"{\i}}di(1997)]%
        {Graf:Pred-Abs}
\bibfield{author}{\bibinfo{person}{Susanne Graf} {and} \bibinfo{person}{Hassen
  Sa{\"{\i}}di}.} \bibinfo{year}{1997}\natexlab{}.
\newblock \showarticletitle{Construction of Abstract State Graphs with {PVS}}.
  In \bibinfo{booktitle}{\emph{Computer Aided Verification, 9th International
  Conference, {CAV} '97, Haifa, Israel, June 22-25, 1997, Proceedings}}
  \emph{(\bibinfo{series}{Lecture Notes in Computer Science},
  Vol.~\bibinfo{volume}{1254})}, \bibfield{editor}{\bibinfo{person}{Orna
  Grumberg}} (Ed.). \bibinfo{publisher}{Springer}, \bibinfo{pages}{72--83}.
\newblock
\href{https://doi.org/10.1007/3-540-63166-6\_10}{doi:\nolinkurl{10.1007/3-540-63166-6\_10}}


\bibitem[Grebenshchikov et~al\mbox{.}(2012)]%
        {Grebenshchikov:Synthesize-Verifiers}
\bibfield{author}{\bibinfo{person}{Sergey Grebenshchikov},
  \bibinfo{person}{Nuno~P. Lopes}, \bibinfo{person}{Corneliu Popeea}, {and}
  \bibinfo{person}{Andrey Rybalchenko}.} \bibinfo{year}{2012}\natexlab{}.
\newblock \showarticletitle{Synthesizing software verifiers from proof rules}.
  In \bibinfo{booktitle}{\emph{{ACM} {SIGPLAN} Conference on Programming
  Language Design and Implementation, {PLDI} '12, Beijing, China - June 11 -
  16, 2012}}, \bibfield{editor}{\bibinfo{person}{Jan Vitek},
  \bibinfo{person}{Haibo Lin}, {and} \bibinfo{person}{Frank Tip}} (Eds.).
  \bibinfo{publisher}{{ACM}}, \bibinfo{pages}{405--416}.
\newblock
\href{https://doi.org/10.1145/2254064.2254112}{doi:\nolinkurl{10.1145/2254064.2254112}}


\bibitem[Gupta et~al\mbox{.}(2011)]%
        {Gupta:OwickiGries-CHC}
\bibfield{author}{\bibinfo{person}{Ashutosh Gupta}, \bibinfo{person}{Corneliu
  Popeea}, {and} \bibinfo{person}{Andrey Rybalchenko}.}
  \bibinfo{year}{2011}\natexlab{}.
\newblock \showarticletitle{Predicate abstraction and refinement for verifying
  multi-threaded programs}. In \bibinfo{booktitle}{\emph{Proceedings of the
  38th {ACM} {SIGPLAN-SIGACT} Symposium on Principles of Programming Languages,
  {POPL} 2011, Austin, TX, USA, January 26-28, 2011}},
  \bibfield{editor}{\bibinfo{person}{Thomas Ball} {and} \bibinfo{person}{Mooly
  Sagiv}} (Eds.). \bibinfo{publisher}{{ACM}}, \bibinfo{pages}{331--344}.
\newblock
\href{https://doi.org/10.1145/1926385.1926424}{doi:\nolinkurl{10.1145/1926385.1926424}}


\bibitem[Gurfinkel et~al\mbox{.}(2016)]%
        {Gurfinkel:SMT-Based-Parameterized}
\bibfield{author}{\bibinfo{person}{Arie Gurfinkel}, \bibinfo{person}{Sharon
  Shoham}, {and} \bibinfo{person}{Yuri Meshman}.}
  \bibinfo{year}{2016}\natexlab{}.
\newblock \showarticletitle{SMT-based verification of parameterized systems}.
  In \bibinfo{booktitle}{\emph{Proceedings of the 24th {ACM} {SIGSOFT}
  International Symposium on Foundations of Software Engineering, {FSE} 2016,
  Seattle, WA, USA, November 13-18, 2016}},
  \bibfield{editor}{\bibinfo{person}{Thomas Zimmermann}, \bibinfo{person}{Jane
  Cleland{-}Huang}, {and} \bibinfo{person}{Zhendong Su}} (Eds.).
  \bibinfo{publisher}{{ACM}}, \bibinfo{pages}{338--348}.
\newblock
\href{https://doi.org/10.1145/2950290.2950330}{doi:\nolinkurl{10.1145/2950290.2950330}}


\bibitem[Heizmann et~al\mbox{.}(2009)]%
        {Heizmann:TAR}
\bibfield{author}{\bibinfo{person}{Matthias Heizmann}, \bibinfo{person}{Jochen
  Hoenicke}, {and} \bibinfo{person}{Andreas Podelski}.}
  \bibinfo{year}{2009}\natexlab{}.
\newblock \showarticletitle{Refinement of Trace Abstraction}. In
  \bibinfo{booktitle}{\emph{Static Analysis, 16th International Symposium,
  {SAS} 2009, Los Angeles, CA, USA, August 9-11, 2009. Proceedings}}
  \emph{(\bibinfo{series}{Lecture Notes in Computer Science},
  Vol.~\bibinfo{volume}{5673})}, \bibfield{editor}{\bibinfo{person}{Jens
  Palsberg} {and} \bibinfo{person}{Zhendong Su}} (Eds.).
  \bibinfo{publisher}{Springer}, \bibinfo{pages}{69--85}.
\newblock
\href{https://doi.org/10.1007/978-3-642-03237-0\_7}{doi:\nolinkurl{10.1007/978-3-642-03237-0\_7}}


\bibitem[Heizmann et~al\mbox{.}(2024)]%
        {Heizmann:Petrification}
\bibfield{author}{\bibinfo{person}{Matthias Heizmann}, \bibinfo{person}{Dominik
  Klumpp}, \bibinfo{person}{Lars Nitzke}, {and} \bibinfo{person}{Frank
  Sch{\"{u}}ssele}.} \bibinfo{year}{2024}\natexlab{}.
\newblock \showarticletitle{Petrification: Software Model Checking for Programs
  with Dynamic Thread Management}. In \bibinfo{booktitle}{\emph{Verification,
  Model Checking, and Abstract Interpretation - 25th International Conference,
  {VMCAI} 2024, London, United Kingdom, January 15-16, 2024, Proceedings, Part
  {II}}} \emph{(\bibinfo{series}{Lecture Notes in Computer Science},
  Vol.~\bibinfo{volume}{14500})}, \bibfield{editor}{\bibinfo{person}{Rayna
  Dimitrova}, \bibinfo{person}{Ori Lahav}, {and} \bibinfo{person}{Sebastian
  Wolff}} (Eds.). \bibinfo{publisher}{Springer}, \bibinfo{pages}{3--25}.
\newblock
\href{https://doi.org/10.1007/978-3-031-50521-8\_1}{doi:\nolinkurl{10.1007/978-3-031-50521-8\_1}}


\bibitem[Hoenicke et~al\mbox{.}(2017)]%
        {Hoenicke:Thread-Modularity-Levels}
\bibfield{author}{\bibinfo{person}{Jochen Hoenicke}, \bibinfo{person}{Rupak
  Majumdar}, {and} \bibinfo{person}{Andreas Podelski}.}
  \bibinfo{year}{2017}\natexlab{}.
\newblock \showarticletitle{Thread modularity at many levels: a pearl in
  compositional verification}. In \bibinfo{booktitle}{\emph{Proceedings of the
  44th {ACM} {SIGPLAN} Symposium on Principles of Programming Languages, {POPL}
  2017, Paris, France, January 18-20, 2017}},
  \bibfield{editor}{\bibinfo{person}{Giuseppe Castagna} {and}
  \bibinfo{person}{Andrew~D. Gordon}} (Eds.). \bibinfo{publisher}{{ACM}},
  \bibinfo{pages}{473--485}.
\newblock
\href{https://doi.org/10.1145/3009837.3009893}{doi:\nolinkurl{10.1145/3009837.3009893}}


\bibitem[Hoenicke and Schindler(2022)]%
        {Hoenicke:Proof-Format}
\bibfield{author}{\bibinfo{person}{Jochen Hoenicke} {and}
  \bibinfo{person}{Tanja Schindler}.} \bibinfo{year}{2022}\natexlab{}.
\newblock \showarticletitle{A Simple Proof Format for {SMT}}. In
  \bibinfo{booktitle}{\emph{Proceedings of the 20th Internal Workshop on
  Satisfiability Modulo Theories co-located with the 11th International Joint
  Conference on Automated Reasoning {(IJCAR} 2022) part of the 8th Federated
  Logic Conference (FLoC 2022), Haifa, Israel, August 11-12, 2022}}
  \emph{(\bibinfo{series}{{CEUR} Workshop Proceedings},
  Vol.~\bibinfo{volume}{3185})}, \bibfield{editor}{\bibinfo{person}{David
  D{\'{e}}harbe} {and} \bibinfo{person}{Antti E.~J. Hyv{\"{a}}rinen}} (Eds.).
  \bibinfo{publisher}{CEUR-WS.org}, \bibinfo{pages}{54--70}.
\newblock
\urldef\tempurl%
\url{https://ceur-ws.org/Vol-3185/paper9527.pdf}
\showURL{%
\tempurl}


\bibitem[Hojjat et~al\mbox{.}(2014)]%
        {Hojjat:CHC-Timed}
\bibfield{author}{\bibinfo{person}{Hossein Hojjat}, \bibinfo{person}{Philipp
  R{\"{u}}mmer}, \bibinfo{person}{Pavle Subotic}, {and} \bibinfo{person}{Wang
  Yi}.} \bibinfo{year}{2014}\natexlab{}.
\newblock \showarticletitle{Horn Clauses for Communicating Timed Systems}. In
  \bibinfo{booktitle}{\emph{Proceedings First Workshop on Horn Clauses for
  Verification and Synthesis, {HCVS} 2014, Vienna, Austria, 17 July 2014}}
  \emph{(\bibinfo{series}{{EPTCS}}, Vol.~\bibinfo{volume}{169})},
  \bibfield{editor}{\bibinfo{person}{Nikolaj~S. Bj{\o}rner},
  \bibinfo{person}{Fabio Fioravanti}, \bibinfo{person}{Andrey Rybalchenko},
  {and} \bibinfo{person}{Valerio Senni}} (Eds.). \bibinfo{pages}{39--52}.
\newblock
\href{https://doi.org/10.4204/EPTCS.169.6}{doi:\nolinkurl{10.4204/EPTCS.169.6}}


\bibitem[Jones(1981)]%
        {Jones:rely-guarantee-thesis}
\bibfield{author}{\bibinfo{person}{Cliff~B. Jones}.}
  \bibinfo{year}{1981}\natexlab{}.
\newblock \emph{\bibinfo{title}{Developing methods for computer programs
  including a notion of interference}}.
\newblock \bibinfo{thesistype}{Ph.\,D. Dissertation}.
  \bibinfo{school}{University of Oxford, {UK}}.
\newblock
\urldef\tempurl%
\url{https://ethos.bl.uk/OrderDetails.do?uin=uk.bl.ethos.259064}
\showURL{%
\tempurl}


\bibitem[Lamport(1977)]%
        {Lamport:Proving-Multiprocess}
\bibfield{author}{\bibinfo{person}{Leslie Lamport}.}
  \bibinfo{year}{1977}\natexlab{}.
\newblock \showarticletitle{Proving the Correctness of Multiprocess Programs}.
\newblock \bibinfo{journal}{\emph{{IEEE} Trans. Software Eng.}}
  \bibinfo{volume}{3}, \bibinfo{number}{2} (\bibinfo{year}{1977}),
  \bibinfo{pages}{125--143}.
\newblock
\href{https://doi.org/10.1109/TSE.1977.229904}{doi:\nolinkurl{10.1109/TSE.1977.229904}}


\bibitem[McMillan(1992)]%
        {McMillan:Unfoldings}
\bibfield{author}{\bibinfo{person}{Kenneth~L. McMillan}.}
  \bibinfo{year}{1992}\natexlab{}.
\newblock \showarticletitle{Using Unfoldings to Avoid the State Explosion
  Problem in the Verification of Asynchronous Circuits}. In
  \bibinfo{booktitle}{\emph{Computer Aided Verification, Fourth International
  Workshop, {CAV} '92, Montreal, Canada, June 29 - July 1, 1992, Proceedings}}
  \emph{(\bibinfo{series}{Lecture Notes in Computer Science},
  Vol.~\bibinfo{volume}{663})}, \bibfield{editor}{\bibinfo{person}{Gregor von
  Bochmann} {and} \bibinfo{person}{David~K. Probst}} (Eds.).
  \bibinfo{publisher}{Springer}, \bibinfo{pages}{164--177}.
\newblock
\href{https://doi.org/10.1007/3-540-56496-9\_14}{doi:\nolinkurl{10.1007/3-540-56496-9\_14}}


\bibitem[McMillan(2006)]%
        {McMillan:Impact}
\bibfield{author}{\bibinfo{person}{Kenneth~L. McMillan}.}
  \bibinfo{year}{2006}\natexlab{}.
\newblock \showarticletitle{Lazy Abstraction with Interpolants}. In
  \bibinfo{booktitle}{\emph{Computer Aided Verification, 18th International
  Conference, {CAV} 2006, Seattle, WA, USA, August 17-20, 2006, Proceedings}}
  \emph{(\bibinfo{series}{Lecture Notes in Computer Science},
  Vol.~\bibinfo{volume}{4144})}, \bibfield{editor}{\bibinfo{person}{Thomas
  Ball} {and} \bibinfo{person}{Robert~B. Jones}} (Eds.).
  \bibinfo{publisher}{Springer}, \bibinfo{pages}{123--136}.
\newblock
\href{https://doi.org/10.1007/11817963\_14}{doi:\nolinkurl{10.1007/11817963\_14}}


\bibitem[Min{\'{e}}(2012)]%
        {Mine:ThreadModularAI}
\bibfield{author}{\bibinfo{person}{Antoine Min{\'{e}}}.}
  \bibinfo{year}{2012}\natexlab{}.
\newblock \showarticletitle{Static Analysis of Run-Time Errors in Embedded
  Real-Time Parallel {C} Programs}.
\newblock \bibinfo{journal}{\emph{Log. Methods Comput. Sci.}}
  \bibinfo{volume}{8}, \bibinfo{number}{1} (\bibinfo{year}{2012}).
\newblock
\href{https://doi.org/10.2168/LMCS-8(1:26)2012}{doi:\nolinkurl{10.2168/LMCS-8(1:26)2012}}


\bibitem[Min{\'{e}}(2014)]%
        {Mine:RelationalAI-ThreadModular}
\bibfield{author}{\bibinfo{person}{Antoine Min{\'{e}}}.}
  \bibinfo{year}{2014}\natexlab{}.
\newblock \showarticletitle{Relational Thread-Modular Static Value Analysis by
  Abstract Interpretation}. In \bibinfo{booktitle}{\emph{Verification, Model
  Checking, and Abstract Interpretation - 15th International Conference,
  {VMCAI} 2014, San Diego, CA, USA, January 19-21, 2014, Proceedings}}
  \emph{(\bibinfo{series}{Lecture Notes in Computer Science},
  Vol.~\bibinfo{volume}{8318})}, \bibfield{editor}{\bibinfo{person}{Kenneth~L.
  McMillan} {and} \bibinfo{person}{Xavier Rival}} (Eds.).
  \bibinfo{publisher}{Springer}, \bibinfo{pages}{39--58}.
\newblock
\href{https://doi.org/10.1007/978-3-642-54013-4\_3}{doi:\nolinkurl{10.1007/978-3-642-54013-4\_3}}


\bibitem[Monat and Min{\'{e}}(2017)]%
        {Monat:PreciseTM-AI}
\bibfield{author}{\bibinfo{person}{Rapha{\"{e}}l Monat} {and}
  \bibinfo{person}{Antoine Min{\'{e}}}.} \bibinfo{year}{2017}\natexlab{}.
\newblock \showarticletitle{Precise Thread-Modular Abstract Interpretation of
  Concurrent Programs Using Relational Interference Abstractions}. In
  \bibinfo{booktitle}{\emph{Verification, Model Checking, and Abstract
  Interpretation - 18th International Conference, {VMCAI} 2017, Paris, France,
  January 15-17, 2017, Proceedings}} \emph{(\bibinfo{series}{Lecture Notes in
  Computer Science}, Vol.~\bibinfo{volume}{10145})},
  \bibfield{editor}{\bibinfo{person}{Ahmed Bouajjani} {and}
  \bibinfo{person}{David Monniaux}} (Eds.). \bibinfo{publisher}{Springer},
  \bibinfo{pages}{386--404}.
\newblock
\href{https://doi.org/10.1007/978-3-319-52234-0\_21}{doi:\nolinkurl{10.1007/978-3-319-52234-0\_21}}


\bibitem[Nieto(2001)]%
        {PrensaNieto:OG-Complete}
\bibfield{author}{\bibinfo{person}{Leonor~Prensa Nieto}.}
  \bibinfo{year}{2001}\natexlab{}.
\newblock \showarticletitle{Completeness of the Owicki-Gries System for
  Parameterized Parallel Programs}. In \bibinfo{booktitle}{\emph{Proceedings of
  the 15th International Parallel {\&} Distributed Processing Symposium
  (IPDPS-01), San Francisco, CA, USA, April 23-27, 2001}}.
  \bibinfo{publisher}{{IEEE} Computer Society}, \bibinfo{pages}{150}.
\newblock
\href{https://doi.org/10.1109/IPDPS.2001.925138}{doi:\nolinkurl{10.1109/IPDPS.2001.925138}}


\bibitem[O'Hearn(2004)]%
        {OHearn:CSL}
\bibfield{author}{\bibinfo{person}{Peter~W. O'Hearn}.}
  \bibinfo{year}{2004}\natexlab{}.
\newblock \showarticletitle{Resources, Concurrency and Local Reasoning}. In
  \bibinfo{booktitle}{\emph{{CONCUR} 2004 - Concurrency Theory, 15th
  International Conference, London, UK, August 31 - September 3, 2004,
  Proceedings}} \emph{(\bibinfo{series}{Lecture Notes in Computer Science},
  Vol.~\bibinfo{volume}{3170})}, \bibfield{editor}{\bibinfo{person}{Philippa
  Gardner} {and} \bibinfo{person}{Nobuko Yoshida}} (Eds.).
  \bibinfo{publisher}{Springer}, \bibinfo{pages}{49--67}.
\newblock
\href{https://doi.org/10.1007/978-3-540-28644-8\_4}{doi:\nolinkurl{10.1007/978-3-540-28644-8\_4}}


\bibitem[Owicki(1976)]%
        {Owicki:Consistent-Complete}
\bibfield{author}{\bibinfo{person}{Susan~S. Owicki}.}
  \bibinfo{year}{1976}\natexlab{}.
\newblock \showarticletitle{A Consistent and Complete Deductive System for the
  Verification of Parallel Programs}. In \bibinfo{booktitle}{\emph{Proceedings
  of the 8th Annual {ACM} Symposium on Theory of Computing, May 3-5, 1976,
  Hershey, Pennsylvania, {USA}}}, \bibfield{editor}{\bibinfo{person}{Ashok~K.
  Chandra}, \bibinfo{person}{Detlef Wotschke}, \bibinfo{person}{Emily~P.
  Friedman}, {and} \bibinfo{person}{Michael~A. Harrison}} (Eds.).
  \bibinfo{publisher}{{ACM}}, \bibinfo{pages}{73--86}.
\newblock
\href{https://doi.org/10.1145/800113.803634}{doi:\nolinkurl{10.1145/800113.803634}}


\bibitem[Owicki and Gries(1976)]%
        {Owicki-Gries}
\bibfield{author}{\bibinfo{person}{Susan~S. Owicki} {and}
  \bibinfo{person}{David Gries}.} \bibinfo{year}{1976}\natexlab{}.
\newblock \showarticletitle{Verifying Properties of Parallel Programs: An
  Axiomatic Approach}.
\newblock \bibinfo{journal}{\emph{Commun. {ACM}}} \bibinfo{volume}{19},
  \bibinfo{number}{5} (\bibinfo{year}{1976}), \bibinfo{pages}{279--285}.
\newblock
\href{https://doi.org/10.1145/360051.360224}{doi:\nolinkurl{10.1145/360051.360224}}


\bibitem[Rival and Yi(2020)]%
        {Rival:IntroStaticAnalysis}
\bibfield{author}{\bibinfo{person}{Xavier Rival} {and}
  \bibinfo{person}{Kwangkeun Yi}.} \bibinfo{year}{2020}\natexlab{}.
\newblock \bibinfo{booktitle}{\emph{Introduction to static analysis: an
  abstract interpretation perspective}}.
\newblock \bibinfo{publisher}{MIT Press}.
\newblock


\bibitem[Saan et~al\mbox{.}(2023)]%
        {Saan:Goblint-SVCOMP23}
\bibfield{author}{\bibinfo{person}{Simmo Saan}, \bibinfo{person}{Michael
  Schwarz}, \bibinfo{person}{Julian Erhard}, \bibinfo{person}{Manuel Pietsch},
  \bibinfo{person}{Helmut Seidl}, \bibinfo{person}{Sarah Tilscher}, {and}
  \bibinfo{person}{Vesal Vojdani}.} \bibinfo{year}{2023}\natexlab{}.
\newblock \showarticletitle{Goblint: Autotuning Thread-Modular Abstract
  Interpretation - (Competition Contribution)}. In
  \bibinfo{booktitle}{\emph{Tools and Algorithms for the Construction and
  Analysis of Systems - 29th International Conference, {TACAS} 2023, Held as
  Part of the European Joint Conferences on Theory and Practice of Software,
  {ETAPS} 2022, Paris, France, April 22-27, 2023, Proceedings, Part {II}}}
  \emph{(\bibinfo{series}{Lecture Notes in Computer Science},
  Vol.~\bibinfo{volume}{13994})}, \bibfield{editor}{\bibinfo{person}{Sriram
  Sankaranarayanan} {and} \bibinfo{person}{Natasha Sharygina}} (Eds.).
  \bibinfo{publisher}{Springer}, \bibinfo{pages}{547--552}.
\newblock
\href{https://doi.org/10.1007/978-3-031-30820-8\_34}{doi:\nolinkurl{10.1007/978-3-031-30820-8\_34}}


\bibitem[Schurr et~al\mbox{.}(2021)]%
        {Schurr:Alethe}
\bibfield{author}{\bibinfo{person}{Hans{-}J{\"{o}}rg Schurr},
  \bibinfo{person}{Mathias Fleury}, \bibinfo{person}{Haniel Barbosa}, {and}
  \bibinfo{person}{Pascal Fontaine}.} \bibinfo{year}{2021}\natexlab{}.
\newblock \showarticletitle{Alethe: Towards a Generic {SMT} Proof Format
  (extended abstract)}. In \bibinfo{booktitle}{\emph{Proceedings Seventh
  Workshop on Proof eXchange for Theorem Proving, PxTP 2021, Pittsburg, PA,
  USA, July 11, 2021}} \emph{(\bibinfo{series}{{EPTCS}},
  Vol.~\bibinfo{volume}{336})}, \bibfield{editor}{\bibinfo{person}{Chantal
  Keller} {and} \bibinfo{person}{Mathias Fleury}} (Eds.).
  \bibinfo{pages}{49--54}.
\newblock
\href{https://doi.org/10.4204/EPTCS.336.6}{doi:\nolinkurl{10.4204/EPTCS.336.6}}


\bibitem[Sch{\"{u}}ssele et~al\mbox{.}(2024)]%
        {SVCOMP24:Automizer}
\bibfield{author}{\bibinfo{person}{Frank Sch{\"{u}}ssele},
  \bibinfo{person}{Manuel Bentele}, \bibinfo{person}{Daniel Dietsch},
  \bibinfo{person}{Matthias Heizmann}, \bibinfo{person}{Xinyu Jiang},
  \bibinfo{person}{Dominik Klumpp}, {and} \bibinfo{person}{Andreas Podelski}.}
  \bibinfo{year}{2024}\natexlab{}.
\newblock \showarticletitle{Ultimate Automizer and the Abstraction of Bitwise
  Operations - (Competition Contribution)}. In \bibinfo{booktitle}{\emph{Tools
  and Algorithms for the Construction and Analysis of Systems - 30th
  International Conference, {TACAS} 2024, Held as Part of the European Joint
  Conferences on Theory and Practice of Software, {ETAPS} 2024, Luxembourg
  City, Luxembourg, April 6-11, 2024, Proceedings, Part {III}}}
  \emph{(\bibinfo{series}{Lecture Notes in Computer Science},
  Vol.~\bibinfo{volume}{14572})}, \bibfield{editor}{\bibinfo{person}{Bernd
  Finkbeiner} {and} \bibinfo{person}{Laura Kov{\'{a}}cs}} (Eds.).
  \bibinfo{publisher}{Springer}, \bibinfo{pages}{418--423}.
\newblock
\href{https://doi.org/10.1007/978-3-031-57256-2\_31}{doi:\nolinkurl{10.1007/978-3-031-57256-2\_31}}


\bibitem[Sch\"ussele et~al\mbox{.}(2026)]%
        {popl26:artifact}
\bibfield{author}{\bibinfo{person}{Frank Sch\"ussele},
  \bibinfo{person}{Matthias Zumkeller}, \bibinfo{person}{Miriam
  Lagunes-Rochin}, {and} \bibinfo{person}{Dominik Klumpp}.}
  \bibinfo{year}{2026}\natexlab{}.
\newblock \bibinfo{booktitle}{\emph{Artifact for the {POPL}'2026 Paper "The
  Ghosts of Empires: Extracting Modularity from Interleaving-Based Proofs"}}.
\newblock
\href{https://doi.org/10.5281/zenodo.17347697}{doi:\nolinkurl{10.5281/zenodo.17347697}}


\bibitem[Vojdani et~al\mbox{.}(2016)]%
        {Vojdani:Goblint-Approach}
\bibfield{author}{\bibinfo{person}{Vesal Vojdani}, \bibinfo{person}{Kalmer
  Apinis}, \bibinfo{person}{Vootele R{\~{o}}tov}, \bibinfo{person}{Helmut
  Seidl}, \bibinfo{person}{Varmo Vene}, {and} \bibinfo{person}{Ralf Vogler}.}
  \bibinfo{year}{2016}\natexlab{}.
\newblock \showarticletitle{Static race detection for device drivers: the
  Goblint approach}. In \bibinfo{booktitle}{\emph{Proceedings of the 31st
  {IEEE/ACM} International Conference on Automated Software Engineering, {ASE}
  2016, Singapore, September 3-7, 2016}},
  \bibfield{editor}{\bibinfo{person}{David Lo}, \bibinfo{person}{Sven Apel},
  {and} \bibinfo{person}{Sarfraz Khurshid}} (Eds.). \bibinfo{publisher}{{ACM}},
  \bibinfo{pages}{391--402}.
\newblock
\href{https://doi.org/10.1145/2970276.2970337}{doi:\nolinkurl{10.1145/2970276.2970337}}


\bibitem[Wachter et~al\mbox{.}(2013)]%
        {Wachter:MultiThreadedImpact}
\bibfield{author}{\bibinfo{person}{Bj{\"{o}}rn Wachter},
  \bibinfo{person}{Daniel Kroening}, {and} \bibinfo{person}{Jo{\"{e}}l
  Ouaknine}.} \bibinfo{year}{2013}\natexlab{}.
\newblock \showarticletitle{Verifying multi-threaded software with impact}. In
  \bibinfo{booktitle}{\emph{Formal Methods in Computer-Aided Design, {FMCAD}
  2013, Portland, OR, USA, October 20-23, 2013}}. \bibinfo{publisher}{{IEEE}},
  \bibinfo{pages}{210--217}.
\newblock
\urldef\tempurl%
\url{https://ieeexplore.ieee.org/document/6679412/}
\showURL{%
\tempurl}


\bibitem[Wetzler et~al\mbox{.}(2014)]%
        {Wetzler:DRAT}
\bibfield{author}{\bibinfo{person}{Nathan Wetzler}, \bibinfo{person}{Marijn
  Heule}, {and} \bibinfo{person}{Warren A.~Hunt Jr.}}
  \bibinfo{year}{2014}\natexlab{}.
\newblock \showarticletitle{DRAT-trim: Efficient Checking and Trimming Using
  Expressive Clausal Proofs}. In \bibinfo{booktitle}{\emph{Theory and
  Applications of Satisfiability Testing - {SAT} 2014 - 17th International
  Conference, Held as Part of the Vienna Summer of Logic, {VSL} 2014, Vienna,
  Austria, July 14-17, 2014. Proceedings}} \emph{(\bibinfo{series}{Lecture
  Notes in Computer Science}, Vol.~\bibinfo{volume}{8561})},
  \bibfield{editor}{\bibinfo{person}{Carsten Sinz} {and} \bibinfo{person}{Uwe
  Egly}} (Eds.). \bibinfo{publisher}{Springer}, \bibinfo{pages}{422--429}.
\newblock
\href{https://doi.org/10.1007/978-3-319-09284-3\_31}{doi:\nolinkurl{10.1007/978-3-319-09284-3\_31}}


\end{thebibliography}
\end{document}